%% file: main_revised.tex
\documentclass[review,onefignum,onetabnum]{siamart190516}

\usepackage{amssymb}

\input{ex_shared}

\ifpdf
\hypersetup{
  pdftitle={A Graphical Representation of Membrane Filtration},
  pdfauthor={B. Gu, L. Kondic, and L. J. Cummings}
}
\fi

\begin{document}
\nolinenumbers
\maketitle

\begin{abstract}
  We analyze the performance of membrane filters represented by pore networks using two criteria: 1) total volumetric throughput of filtrate over the filter lifetime and 2) accumulated foulant concentration in the filtrate. We first formulate the governing equations of fluid flow on a general network, and we model transport and adsorption of particles (foulants) within the network by imposing an advection equation with a sink term on each pore (edge) as well as conservation of fluid and foulant volumetric flow rates at each pore junction (network vertex).  Such a setup yields a system of partial differential equations on the network. We study the influence of three geometric network parameters on filter performance: 1) average number of neighbors of each vertex; 2) initial total void volume of the pore network; and 3) tortuosity of the network. We find that total volumetric throughput depends more strongly on the initial void volume than on average number of neighbors.  Tortuosity, however, turns out to be a universal parameter, leading to almost perfect collapse of all results for a variety of different network architectures. In particular, the accumulated foulant concentration in the filtrate shows an exponential decay as tortuosity increases.
\end{abstract}

\begin{keywords}
  Fluid mechanics, filtration, networks, graph theory and discrete calculus
\end{keywords}

\begin{AMS}
  76S05, 35R02, 94C15, 35E15
\end{AMS}

\section{Introduction}
In many real membrane filters, cavities in the membrane material (pore junctions) are connected by capillaries or channels (pore throats) of length significantly larger than their radius (see \cref{fig:schematic}). Such slender geometry allows for simplifications of the equations governing filtration. At the same time, the membrane pore structure can be viewed as a complex network with vertices and edges, which represent the pore junctions and pore throats respectively \cite{martinez,3d_image}. A fluid feed with foulants is then driven across the membrane through these channels by the transmembrane pressure, while foulant is deposited on the pore walls.

There is a broad range of studies involving dynamics on networks, such as models of the cardiovascular network \cite{changroper}, chemical reaction networks \cite{chem} and optimal control in networks of pipes and canals \cite{opt_control}, among many others. Accordingly, there are various graph models that may be used to represent complex networks, including pore networks in porous media. One such setup is the random geometric graph (RGG) model that has garnered significant attention from theoreticians~\cite{rgg_penrose,rgg_dall}, 
as well as from applied scientists working on applications such as sensor networks \cite{cluster}, social science \cite{social,social2,social3,social4}, neuroscience \cite{neuro} and, most relevant to our work, filtration modeling \cite{rgg_2011_thiedmann,griffith_jms}. The RGG setup involves generating a specified number of random points, uniformly distributed within a specified domain, and then connecting any pair of points that are closer than some search radius (according to the chosen metric). 

Within the specific RGG setup used by Griffiths {\it et al.}~\cite{griffith_jms} for filtration modeling, the chosen domain is a unit cube (normalized on the membrane thickness and representing a portion of the filter membrane); the metric is the usual 3D Euclidean distance; and to model top inlets and bottom outlets (see \cref{fig:schematic_a}) random points are uniformly distributed on the top and the bottom surfaces of the cube. This setup requires four parameters: the number of interior points; the search radius; the number of top inlets; and the number of bottom outlets. Furthermore, the fluid flow is modelled by the Hagen-Poiseuille equations, whose validity relies on the physical assumption that the edges have a small aspect ratio (radius vs. length), prompting an additional constraint on the network generation. 

Another important aspect of modeling the dynamics of filtration on the membrane network is fouling. Membrane fouling occurs when contaminants transported by the fluid feed within the pore network become trapped within it. There are three primary fouling mechanisms in membrane filtration: 1) adsorption; 2) sieving and 3) cake layer formation. Adsorption occurs when small particulate foulants are deposited on the pore walls due to physical or chemical interactions with membrane material (or with existing particle deposits). Sieving involves fouling particles of size comparable to the pore size that may partially cover or completely block upstream pore entries or internal junctions. Cake formation occurs during the later stages of filtration, where fouling particles are packed against each other inside the pore throats or on top of the membrane surface, further restricting fluid flow.

In this work, we consider adsorptive fouling only, modeling two main features. First, the fouling particles (referred to simply as particles, henceforth) deposit on pore walls as they are advected by fluid flow, thereby reducing foulant concentration in the feed as it traverses the network. Second, pore radius decreases due to the particle deposition on the pore wall. Within a continuum model for the particle concentration, the first effect can be captured by an appropriate sink term in the advection equation for the particle concentration, while the second may be modeled by an evolution equation for the pore radius. 

Classical models of particle advective transport on graphs were developed and studied by Chapman \& Mesbahi \cite{chapman_mesbahi}, but these authors did not incorporate a sink term to capture external effects such as fouling. Meanwhile, Gu {\it et al.} \cite{gu2020} considered a coarse discretization of a transport equation with deposition on regular layered pore structures with interconnections, which can be generalised to more complex networks (represented by graphs). In this work, we combine these two approaches and formulate a transport equation on the network using the graph theoretical framework. 
\begin{figure}[tbhp]
\centering
\subfloat[Experimental Image]{\label{fig:schematic_a}\includegraphics[width=.43\textwidth]{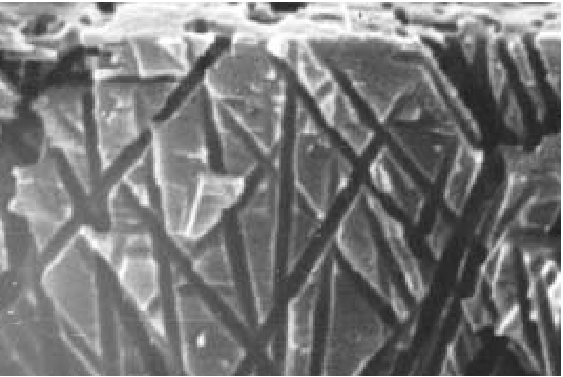}}
\subfloat[Inlets and Interior Junctions]{\label{fig:schematic_b}\includegraphics[width=.43\textwidth]{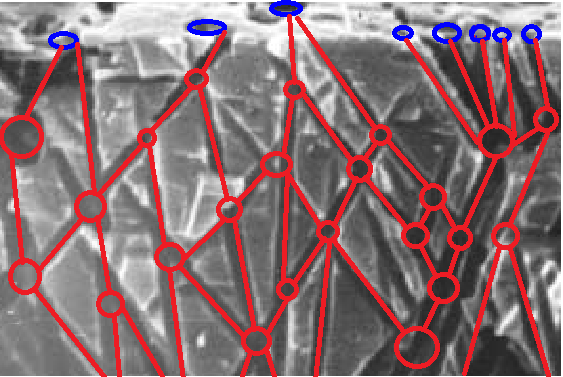}}\\
\caption{Schematic: (a) an experimental image with lateral view of a filter cross-section \cite{Apel2001}; (b) a corresponding (partial) graph representation with inlets on the top surface ({\color{blue}blue}) and interior pore junctions and throats ({\color{red}red}).} 
\label{fig:schematic}
\end{figure}

Our principal contributions are to develop and solve a fluid flow and fouling problem on a graph representing the pore network within a membrane filter, by constructing the associated operators based on the underlying graph; and to use our model to identify important correlations between measurable network properties and key filtration performance metrics. The novel random graph generation technique we use achieves two primary goals: the constructed network is such that physical assumptions for the Hagen-Poiseuille fluid flow model remain valid; and the overall number of model parameters is minimized. Our model is described in detail in \cref{sec:2}: in \cref{sec:2.1}, we introduce the aspects of graph theory relevant to the present problem; in \cref{sec:2.2} we present a specific graph construction that represents a membrane network; in \cref{sec:2.3}, we describe the fluid flow in an arbitrary pore (edge of the graph), before defining the necessary graph operators and associated function spaces on graphs in \cref{sec:2.4}. With these building-blocks we are then able to set up the governing equations for fluid flow and foulant transport on a graph in \cref{sec:2.5,sec:2.6} respectively. In \cref{sec:2.7}, we define the metrics that will be used to characterize filter performance. In \cref{sec:3} we introduce appropriate scalings to nondimensionalize the model and in \cref{sec:4} we outline the algorithm we use to solve it. \Cref{sec:5} contains our results, presented in the context of the performance metrics defined in \cref{sec:2.7}. In \cref{sec:6} we discuss our findings and summarize conclusions.

\section{Modeling}\label{sec:2}
In this section, we describe and construct a graph that models a membrane filter pore network, and set up the governing equations of fluid flow and foulant transport through the network. We then introduce the performance measures associated with our membrane filter, which we will use to characterize filtration effectiveness. We refer the reader to \cref{table} for all notation used in this paper.
\begin{table}
\centering
\resizebox{\textwidth}{!}{
\begin{tabular}{llll}
$v_i$     & $i$-th vertex                   & $e_{ij}$                & An edge connecting vertices $v_i$ and $v_j$                     \\
$V$     & Set of all vertices                   & $E$                & Set of all edges                     \\
$V_{\rm top}$     & Set of vertices on the top membrane surface                  & $V_{\rm bot}$                & Set of vertices at the bottom membrane surface                    \\
$\mathbf{L_W}$     & (Weighted) Graph Laplacian of $G$                  & $\mathbf{W}$                & Weighted adjacency matrix of $G$        \\
$\mathbf{M}$     & Incidence matrix of $G$                      & $\mathbf{M}^{\rm T}$                  & Incidence matrix transpose                              \\
$\boldsymbol{P}\left(v_i\right)$     & Pressure at the vertex $v_i$                   & $\boldsymbol{C}_i$     & Concentration in vertex $v_i$ (unit of number per volume)              \\
${P}_0$     & Pressure at membrane top surface                   & 
${C}_0$     & Foulant concentration in the feed solution \\
$\mathbf{K}_{ij}$       & Conductance of edge $e_{ij}$                    & $\mathbf{Q}_{ij}$             & Flux through edge $e_{ij}$ \\ 
$R_{ij}$     & Radius of edge $e_{ij}$                    & $A_{ij}$                & Length of edge $e_{ij}$                               \\
$W$     & Side length of the square cross-section & $\chi\left(\cdot,\cdot\right)$ & Metric on $V$ \\
$D$     & Search radius                   & $N_{\rm total}$                & Total number of vertices                     \\
$D_{\rm min}$     & Minimal edge length                   & $\tau$                & Tortuosity                     \\
$\mu$     & Fluid viscosity & $\alpha$ & Parameter related to particle volume \\
$\Lambda$     & Affinity between the foulant and the pore wall (unit of velocity) & $N$ & Average number of neighbors for a graph $G$\\
$H$     & Throughput & $C_{\rm acm}$ & Accumulated foulant concentration at membrane outlet \\
 \hline
          &                                                &                      &                                                        
\end{tabular}
}
\caption{Key nomenclature used throughout this work.}
\label{table}
\end{table}

\subsection{Graph Theoretical Setup}\label{sec:2.1}
We consider a membrane filter as a slab of porous material that consists of an upstream (top) surface containing pore inlets, a downstream (bottom) surface containing pore outlets, and a porous interior comprising membrane material as the solid region and pore junctions and throats as the void region. We restrict attention to a cubic sub-region of the membrane slab in which the top and bottom surfaces are squares with side length $W$ separated by distance $W$ (the membrane depth). The pore throats in the interior of this sub-region are assumed to be circular cylinders connected at pore junctions, forming a network. We model this interior network of pores by a graph $G$ consisting of a set of vertices $V$ (the pore junctions) and edges $E$ (the pore throats; see \cref{fig:schematic}), more compactly written as $G=\left(V,E\right)$. The set of vertices $V$ consists of pore inlets on the membrane top (upstream) surface, pore junctions in the interior of the membrane material where edges (pores) meet, and pore outlets at the downstream surface of the filter. Each edge $e\in E$ represents a pore throat that connects two vertices. From a fluid dynamics perspective the graph forms a flow network,
with fluid flowing along the edges and through the vertices. Each edge of the graph is associated with a weight to be specified depending on the context; in our model the weights are associated with the fluid flux through the edge. In constructing $G$ we begin with an undirected graph, with directionality to be imparted by the direction of fluid flow (once defined). 

Each vertex $v\in V$ is associated with a Euclidean position vector $\boldsymbol{X}\left(v\right) = \left(X_1\left(v\right),X_2\left(v\right),X_3\left(v\right)\right)\in \mathbb{R}^3,$ where coordinates $X_1$, $X_2$ lie in the plane of the membrane, while $X_3$ is measured perpendicular to the membrane in the direction of flow.

\begin{definition} (Vertex set)
Let $X_3\left(v\right)$ be the depth of the vertex $v$, measured by the shortest distance from the vertex to the membrane inlet surface. Let $V_{\rm{top}},V_{\rm{bot}}\subset V$ be the set of vertices that lie in the top and bottom membrane surfaces respectively, 
\[
V_{{\rm top}} =\left\{ v\in V:X_{3}\left(v\right)=0\right\},\quad V_{{\rm bot}} =\left\{ v\in V:X_{3}\left(v\right)=W\right\} 
\]
where $W$ is the depth of the membrane. The set of vertices in the interior of the membrane is given by $V_{\rm int} = V\backslash \left(V_{\rm top}\cup V_{\rm bot}\right)$.
\end{definition}

We further assume that the graph is simple, i.e. no vertex is connected to itself, equivalent to the assumption that fluid always flows out of a pore junction.
\begin{definition}\label{def:edges} An edge $e_{ij}\in E \subseteq V\times V$ is a pair $\left(v_i,v_j\right)$, with $i\neq j$ whose existence is governed by a connection law (to be described). We do not distinguish between the order of the pair, namely, $e_{ij}=e_{ji}$.
\end{definition}

\subsection{Graph Generation for a Membrane Network}\label{sec:2.2}
Here we describe a random graph model in the specific context of membrane filtration by clarifying the constituents of $V_{\rm top}$, $V_{\rm bot}$, $V_{\rm int}$ and $E$, using a variant of the random geometric graph. 

In the interior of a rectangular box $\Omega$ with square cross-sections of side length $W$ and height $2W$, $\Omega := \left[0,W\right]\times\left[0,W\right]\times\left[-\frac{W}{2},\frac{3W}{2}\right]$, we generate independent uniformly distributed random points in each cartesian coordinate, from which our set of vertices will be taken. To set up a connection rule for vertices in order to define edges, we first introduce a physically-motivated assumption. Earlier work such as that of Iritani~\cite{iritani} and Siddiqui {\it et al.}~\cite{perm2016} suggests that pore throats in certain types of membrane filters are approximately cylindrical and slender, with $R_{ij}/A_{ij} \ll 1$, where $R_{ij}$ and $A_{ij}$ are the radius and length respectively of the edge (pore throat) $e_{ij}$ connecting vertices (pore junctions) $v_i$ and $v_j$. We assume our model membrane filter to be of this type, justifying the use of the Hagen-Poiseuille framework~\cite{probstein} to model steady laminar flow of fluid, assumed to be Newtonian with viscosity $\mu$, within each edge (details of the flow model are in \cref{sec:2.5} later). In all simulations we present, the initial edge radii are the same for all $i,j$ (though they will shrink at different rates under subsequent fouling depending on local particle concentration), thus validity of the Hagen-Poiseuille model requires that the length of each edge must exceed a certain threshold, $D_{\rm min}$. Meanwhile, to tune vertex connectivity, we prescribe a parameter $D$ as the maximum possible edge length. More precisely, we define the edge set,
\begin{equation}
    E=\left\{ \left(v,w\right)\in V\times V:D_{\min}<\chi\left(v,w\right)<D\right\},
    \label{edge}
\end{equation}
where $\chi\left(\cdot,\cdot\right)$ is the metric on our graph.
In this paper, we investigate and compare two distinct metrics: 1) $\Omega$ is treated as an isolated domain, with no pores entering or leaving through the four sides parallel to the $X_3$-direction, representing an isolated sub-unit of the whole membrane ({\it isolated} case) and 2) $\Omega$ is treated as a periodic domain such that any pore exiting through one of these four sides re-enters on the opposite side ({\it periodic} case). In each of these two cases the metric is defined by
\begin{align}
\chi\left(v,w\right) & =\begin{cases}
\left\Vert \boldsymbol{X}\left(v\right)-\boldsymbol{X}\left(w\right)\right\Vert _{2}, & {\rm isolated}\\
\min_{\boldsymbol{z}}\left\Vert \boldsymbol{X}\left(v\right)-\boldsymbol{X}\left(w\right)-\left(\boldsymbol{z},0\right)\mid\boldsymbol{z}=\left\{ -W,0,W\right\} ^{2}\right\Vert _{2}, \,\, D \leq \frac{W}{2}, & {\rm periodic}
\end{cases}
\label{metric}
\end{align}
where we note that the periodic metric wraps points around the boundaries in the $X_1$ and $X_2$ directions. This metric effectively determines the pore lengths $A_{ij} = \chi\left(v_i,v_j\right)$. All scenarios considered in this paper have search radius $D\leq \frac{W}{2}$, so we avoid the complication of having pores connect to each other twice (in the periodic case). 

\begin{figure}[!ht]
    \centering
    \subfloat[2D schematic of the 3D graph generation]{\label{2d_schematic}\includegraphics[scale = 0.3]{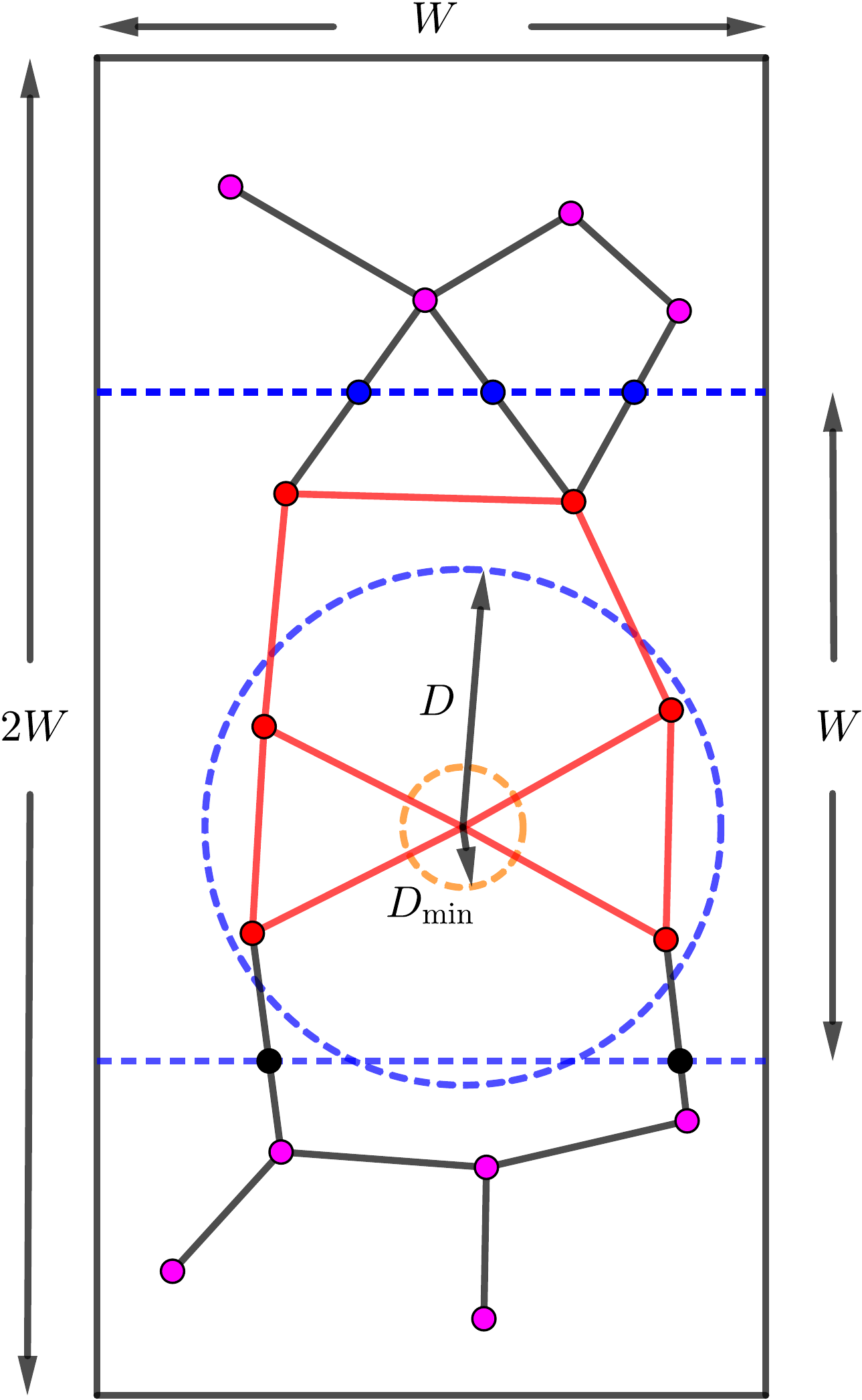}}
    \qquad
    \subfloat[3D realization with the periodic metric]{\label{3d_schematic}\includegraphics[scale = 0.55]{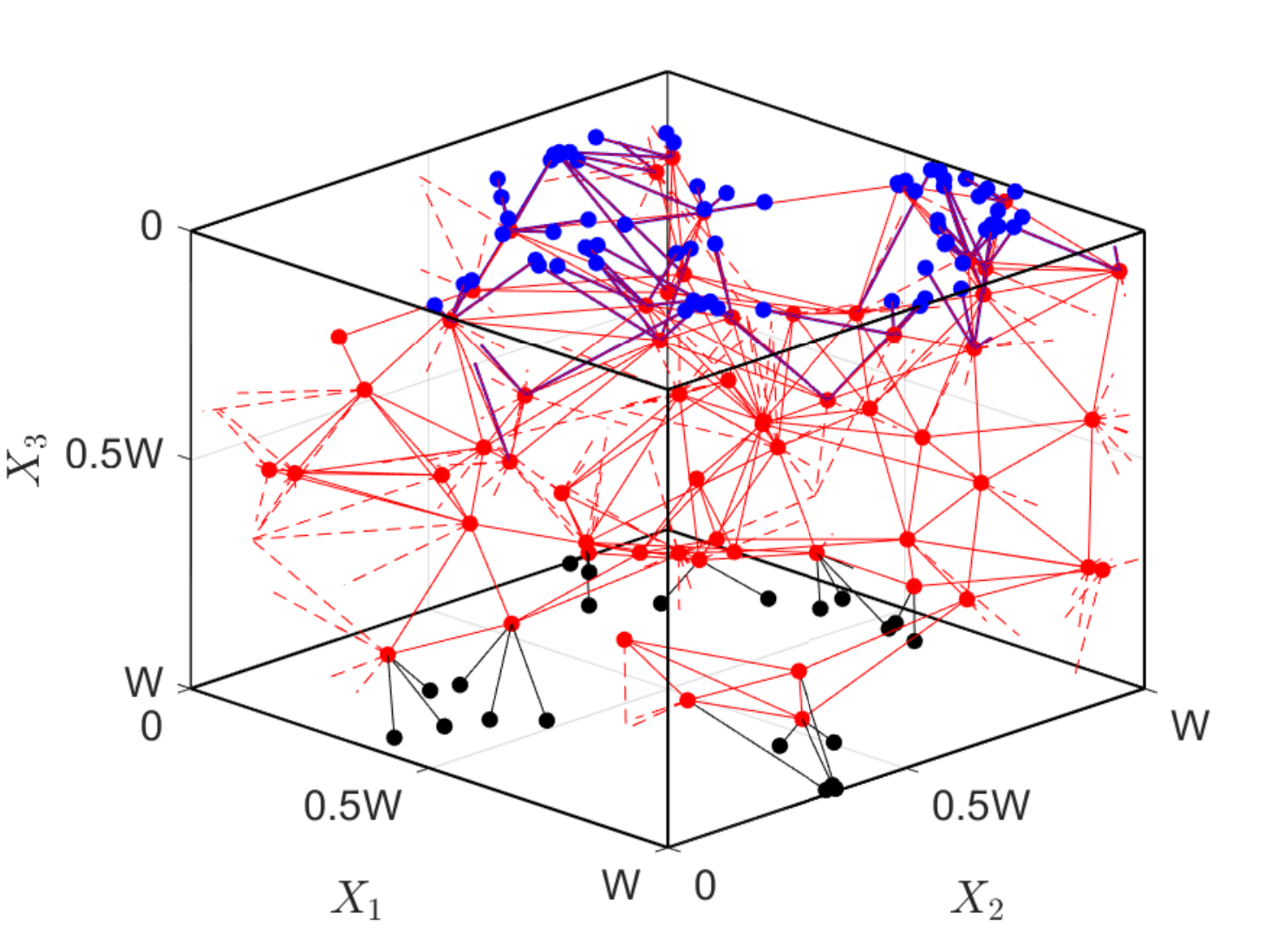}}
    \caption{Labels for (a) and (b): {\color{red}red filled circles} form $V_{\rm int}$, {\color{blue}blue filled circles} form induced inlets $V_{\rm top}$, and black filled circles form induced outlets $V_{\rm bot}$. In (a), {\color{blue}blue dotted lines} are cutting lines (planes in 3D), and {\color{magenta}magenta circles} are discarded points. The {\color{blue}blue} and {\color{orange}orange} dotted circles form the search annulus enforced by $D$ and $D_{\rm min}$. In (b), periodicity is enforced only through the four interior walls per \cref{metric}.}
    \label{fig:3d_bdy_gen}
\end{figure}
To generate the set of upstream (top) inlets $V_{\rm top}$ and downstream (bottom) outlets $V_{\rm bot}$, we cut $\Omega$ with two horizontal planes, $X_3=0$ (upstream cutting plane) and $X_3=W$ (downstream cutting plane). The points of intersection of these planes and the edges defined by \cref{edge} form the set of pore inlets $V_{\rm top}$ and outlets $V_{\rm bot}$ respectively (see \cref{fig:3d_bdy_gen}). All points between the planes $X_3=0,W$ form the set of interior vertices $V_{\rm int}$. This cutting procedure effectively eliminates two parameters (number of inlets and outlets) from the RGG setup proposed by Griffiths {\it et al.}~\cite{griffith_jms}. We also impose two geometric constraints: 1) total volume of the edges generated by the above scheme cannot exceed that of the box and 2) total cross-sectional area of all inlets (resp. outlets) cannot exceed the area of the top (resp. bottom) membrane surface.

Our setup ensures that vertices in the top and bottom membrane surfaces are not connected by any edges that lie in the same surface. Another nice feature of our graph generation is that top inlets and bottom outlets are always connected to interior points, a feature that the classical RGG setup does not guarantee (some inlets or outlets may be isolated with positive probability and must be regenerated). 

\subsection{Flow in an Edge}\label{sec:2.3}
We now briefly describe how Newtonian fluid of viscosity $\mu$ flows through a single edge (pore), before introducing the definitions of the relevant graph operators and associated function spaces (in \cref{sec:2.4}), necessary to describe flow through the whole network. Two key physical quantities specify the flow: pressure $\boldsymbol{P}\left(v_i\right)$ at each vertex $v_i \in V$ and flux $\mathbf{Q}_{ij}$ through each edge $\left(v_i,v_j\right)\in E$. For each such edge, with length $A_{ij}$ and radius $R_{ij}$ (assumed to depend only on $i,j$ and time $T$ in the following, but see the more detailed discussion in \cref{sec:2.6} below), we use the Hagen-Poiseuille equation to describe fluid flow through it. This law gives fluid flux $\mathbf{Q}_{ij}$ in $(v_i,v_j)$ as 
\begin{equation}
    \mathbf{Q}_{ij}=\mathbf{K}_{ij}\left(\boldsymbol{P}\left(v_i\right)-\boldsymbol{P}\left(v_j\right)\right), \quad \left(v_i,v_j\right)\in E,
    \label{intro_hp}
\end{equation}
where $\mathbf{K}_{ij}$ are the entries of a weight matrix $\mathbf{K}$ listing the {\it conductance} of each edge, \begin{equation}
\mathbf{K}_{ij}=\begin{cases}
\frac{\pi R_{ij}^{4}}{8\mu A_{ij}}, & \left(v_i,v_j\right)\in E,\\
0, & \text{otherwise}.
\end{cases}
\label{conductance}
\end{equation}
Note that $\mathbf{Q}_{ij}$ is spatially uniform within the edge under the Hagen-Poiseuille framework. Moreover, to ensure the Hagen-Poiseuille approximation remains valid we require that the aspect ratio of each edge, $R_{ij}/A_{ij}$, is small, enforced by choosing $D_{\rm min}$ (introduced in \cref{edge}) sufficiently large that 
\[
\frac{R_{ij}}{A_{ij}}\leq\frac{{\displaystyle \max_{\left(v_{i},v_{j}\right)\in E}R_{ij}}}{D_{\rm min}} \ll 1.
\]
To solve \cref{intro_hp} on a graph with many edges we must utilize the graph structure to express, for example, the pressure difference $\boldsymbol{P}\left(v_i\right)-\boldsymbol{P}\left(v_j\right)$ in \cref{intro_hp}, while maintaining flux conservation at each vertex. We introduce the necessary operators below.

\subsection{Operators and Function Spaces on Graphs}\label{sec:2.4}
We now introduce the function spaces and operators required to solve for the above-defined flow on our graph. First, define the set of {\it vertex functions} 
\begin{align}
\mathcal{V}=\left\{\boldsymbol{u}:V\rightarrow\mathbb{R}:\sum_{x\in V}\boldsymbol{u}\left(x\right)^{2}<\infty\right\},
\label{l2V}
\end{align}
endowed with inner product $\left(\boldsymbol{u\cdot v}\right)_{\mathcal{V}} = \sum_{x\in V}\boldsymbol{u}\left(x\right)\boldsymbol{v}\left(x\right)$, and the set of {\it edge functions}
\begin{align}
\mathcal{E}=\left\{ \boldsymbol{F}:E\rightarrow\mathbb{R}:\frac{1}{2}\sum_{x,y\in V}w_{xy}\boldsymbol{F}\left(x,y\right)^{2}<\infty\right\},
\label{l2E}
\end{align}
where $\boldsymbol{F}$ is a skew-symmetric function that satisfies
\begin{equation}
    \boldsymbol{F}\left(x,y\right) = -\boldsymbol{F}\left(y,x\right),
    \label{flow_condition}
\end{equation}
also known as the {\it flow condition}, and $w_{xy}$ are specified weights. This space is endowed with weighted inner product
\[
\left(\boldsymbol{F}_1\boldsymbol{\cdot}\boldsymbol{F}_2\right)_{\mathcal{E}}=\frac{1}{2}\sum_{x,y\in V}w_{xy}\boldsymbol{F}_1\left(x,y\right)\boldsymbol{F}_2\left(x,y\right).
\]
An example of a vertex function in the context of our problem is the pressure at each pore junction, while flux is an edge function when a well-defined weight is specified.

We next introduce the incidence matrix (and its transpose), a matrix operator that encodes the most fundamental information of a network represented by a graph --- an array specifying the vertices that each edge connects. Its operation on edge functions provides a finite difference of quantities prescribed at connected vertices.
\begin{definition}\label{definition_inc}
(Incidence matrix and transpose) Let $G=\left(V,E\right)$ be a graph. The incidence matrix $\mathbf{M}$ is a linear operator $\mathbf{M}:\mathcal{V}\rightarrow \mathcal{E}$ such that for $\boldsymbol{u}\in \mathcal{V}$, 
\begin{equation}
    \left(\mathbf{M}\boldsymbol{u}\right)\left(v_i,v_j\right) = \boldsymbol{u}\left(v_i\right)-\boldsymbol{u}\left(v_j\right), \quad v_i,v_j\in V.
    \label{disc_grad}
\end{equation}
The transpose $\mathbf{M}^{\rm T}:\mathcal{E}\rightarrow \mathcal{V}$ is also a linear operator that satisfies
\begin{equation}
\left(\mathbf{M}\boldsymbol{u}\boldsymbol{\cdot}\boldsymbol{F}\right)_{\mathcal{E}} =\left(\boldsymbol{u}\boldsymbol{\cdot}\mathbf{M}^{\rm T}\boldsymbol{F}\right)_{\mathcal{V}}, \quad \boldsymbol{u}\in \mathcal{V},\quad \boldsymbol{F}\in \mathcal{E}.
\label{transpose}
\end{equation}
\end{definition}
A direct calculation of \cref{transpose} using the above definitions yields the following identity
\begin{equation}
\left(\mathbf{M}^{\rm T}\boldsymbol{F}\right)\left(x\right)=\sum_{y:\left(x,y\right)\in E}w_{xy}\boldsymbol{F}\left(x,y\right),
\label{div}
\end{equation}
which computes a weighted sum of the edge functions whose edges connect to $x$.

Moreover, each graph $G$ can be associated with a graph Laplacian $\mathbf{L}$, an item central to graph analysis. 
\begin{definition} (Graph Laplacian)
	The $\mathbf{W}$-weighted graph Laplacian is given by 
	\begin{equation}
	\mathbf{L}_{\mathbf{W}}:=\mathbf{D}-\mathbf{W},
	\label{graph_laplacian}
	\end{equation}
    where $\mathbf{D}$ is the (diagonal) $\mathbf{W}$-weighted degree matrix for $G$, with entries
	\begin{equation}
	\mathbf{D}_{ij}=\begin{cases}
	\sum_{k=1}^{\left|V\right|}\mathbf{W}_{ik}, & j=i,\\
	0, & \text{otherwise},
	\end{cases}
	\label{degree}
	\end{equation}
	and $\mathbf{W}$ is a weighted adjacency matrix, with nonnegative entries $\mathbf{W}_{ij}$ when $\left(v_i,v_j\right)\in E$, to be specified according to context. We also say that $e_{ij}$ is an edge connecting $v_i$ and $v_j$ iff $\mathbf{W}_{ij}>0$ (or simply, $v_i$ and $v_j$ are connected or adjacent).
\end{definition}
$\mathbf{\mathbf{L}_W}$ is symmetric since $\mathbf{W}$ is, and the nonzero entries in the diagonal matrix $\mathbf{D}$ represent the respective row sums of $\mathbf{W}$. An unweighted version of $\mathbf{L_W}$ corresponds to the case where $\mathbf{W}_{ij} = 1$ for every $e_{ij}\in E$, and thus each entry of the adjacency matrix is an indicator that identifies an edge between a pair of vertices. In this case, each diagonal element of $\mathbf{D}$ is the number of edges connected to each vertex. 

Combining \cref{disc_grad}, \cref{div} and \cref{graph_laplacian} we obtain the classical result connecting the incidence matrix and the graph Laplacian. 
\begin{proposition}\label{lw}
For $\boldsymbol{u}\in \mathcal{V}$, $\mathbf{M}^{\rm T} \mathbf{M}\boldsymbol{u} = \mathbf{L_W}\boldsymbol{u}$.
\end{proposition}
This result (a proof of which is given by Grady~\cite{grady}) provides the basis for theoretical understanding of the structure of $\mathbf{M}$, and in turn, $\mathbf{L_W}$. The continuum analog of this result is the equivalence between the divergence of a gradient and the Laplacian. 

In many applications, the choice of weights $\mathbf{W}$ depends on the physical problem of interest. In the context of membrane filtration it is natural to consider weights such as conductance and flux, which we describe in \cref{sec:2.5} in more detail. We will apply \cref{lw} to set up the flow problem on our graph $G$. 

\subsection{Flow on a Graph}\label{sec:2.5}
In this section, we outline a general approach to describe fluid flow on a graph using the definitions given in \cref{sec:2.4} above. We set up the system of governing equations using the structural information of $G$, i.e. the graph Laplacian $\mathbf{L_W}$ with a properly chosen weight matrix $\mathbf{W}$. 

First, we rewrite \cref{intro_hp} using the incidence operator \cref{disc_grad},
\begin{equation}
    \mathbf{Q}_{ij}=\mathbf{K}_{ij}\left(\boldsymbol{P}\left(v_i\right)-\boldsymbol{P}\left(v_j\right)\right) = \mathbf{K}_{ij}\left(\mathbf{M}\boldsymbol{P}\right)\left(v_i,v_j\right), \quad \left(v_i,v_j\right)\in E,
    \label{hp}
\end{equation}
where $\mathbf{K}_{ij}$ is given by \cref{conductance}. The values of the pressure at the vertices, $\boldsymbol{P}\left(v_i\right)$ for each $v_i\in V$, form a vector of length $\left|V\right|$ (a {\it vertex function}, $\boldsymbol{P} \in \mathcal{V}$). The flux values $\mathbf{Q}_{ij}$ naturally form a matrix. Note that $\mathbf{Q}$ is an {\it edge function}, $\mathbf{Q}\in \mathcal{E}$, since it satisfies the flow condition \cref{flow_condition}. Next, we impose conservation of flux at each junction, \begin{equation}
    \sum_{v_j:\left(v_{i},v_{j}\right)\in E}\mathbf{Q}_{ij}=0, \quad v_i\in V_{\rm int}.
\end{equation}
Using \cref{div} with $\boldsymbol{F} = \mathbf{M}\boldsymbol{P}$, \cref{hp} and noting that $\mathbf{Q}\in \mathcal{E}$, we obtain for $v_i\in V_{\rm int}$,
\begin{equation}
    0=\sum_{v_j:\left(v_{i},v_{j}\right)\in E}\mathbf{Q}_{ij}=\sum_{v_j:\left(v_{i},v_{j}\right)\in E}\mathbf{K}_{ij}\left(\mathbf{M}\boldsymbol{P}\right)\left(v_i,v_j\right)=\left(\mathbf{M}^{{\rm T}}\mathbf{M}\boldsymbol{P}\right)\left(v_i\right) = \left(\mathbf{L_K}\boldsymbol{P}\right)\left(v_i\right),
    \label{dis_lp}
\end{equation}
to which we append the boundary conditions
\begin{equation}
    \boldsymbol{P}\left(v\right) = P_0, \quad v\in V_{\rm top};\quad \boldsymbol{P}\left(v\right) = 0, \quad v\in V_{\rm bot},
    \label{p:bc}
\end{equation}
modeling the transmembrane pressure difference that drives fluid flow.

Once we solve the linear system \cref{dis_lp}--\cref{p:bc} for pressure $\boldsymbol{P}$, we use \cref{hp} to compute the entries of the flux matrix $\mathbf{Q}$. 

\subsection{Foulant Advection and Adsorptive Fouling}\label{sec:2.6}
The previous section was concerned solely with the flow of Newtonian fluid through the graph. Here we address the fact that in filtration the fluid is a feed solution carrying particles, which are removed by the filter, leading to fouling. Fouling can occur via a number of distinct modes; we consider only adsorptive fouling by particles much smaller than pores, which are transported through the network by the flow and deposit on pore walls. 

We use a continuum model for the particle (foulant) concentration within the feed. To characterize particle transport on a network, we must describe the transport on each edge, accomplished via an advection equation with an adsorptive sink \cite{sanaei2018flow,gu2020}, then impose conservation of particle flux at each vertex. For each edge $e_{ij}=(v_i,v_j)$ of length $A_{ij}$, with $Y$ a local coordinate measuring distance along the edge from $v_i$ (positive in the direction of flux $\mathbf{Q}_{ij}$), let $C_{ij}\left(Y,T\right)$ be the particle concentration at any point of the edge at time $T$, then
\begin{subequations}
\begin{align}
\mathbf{Q}_{ij}\frac{\partial C_{ij}}{\partial Y} &= -\Lambda R_{ij}C_{ij}, \quad 0\leq Y \leq A_{ij},\label{cont_transport}\\
C_{ij}\left(0,T\right) &= \boldsymbol{C}_{i}\left(T\right), \quad  \left(v_i,v_j\right)\in E, \label{bc_conc}
\end{align}
\end{subequations}
where $\Lambda$ is a parameter (with dimensions of velocity) that captures the affinity between the foulant and the pore wall. In practice this could change depending on whether the pore wall is clean, or already fouled by particle deposits; in this work we assume $\Lambda$ is constant throughout the filtration. We denote by $\boldsymbol{C}_{i}\left(T\right)$ the foulant concentration at the vertex $v_i$, which acts as a boundary condition for the foulant concentration $C_{ij}$ flowing to downstream edges. The array of $\boldsymbol{C}_{i}\left(T\right)$ values forms a vertex function $\boldsymbol{C}$. For $v_i \in V_{\rm top}$, $\boldsymbol{C}_{i}\left(T\right)$ is prescribed by imposing a constant boundary condition (the concentration $C_0$ in the feed solution), and conservation of particle flux is imposed to determine the $\boldsymbol{C}_{i}\left(T\right)$ for $v_i \in V\backslash V_{\rm top}$. 

The system \cref{cont_transport}, \cref{bc_conc} is simple, but once coupled to an evolution equation for pore radii $R_{ij}(Y,T)$ and solved on many different (large) graph realizations to obtain reliable statistics, is time-consuming to solve numerically. We therefore adopt a convenient approximation (discussed further below), assuming that pore radius does not change appreciably along individual pores, and may thus be considered approximately independent of the coordinate $Y$. With this assumption, we observe that the system \cref{cont_transport}, \cref{bc_conc} has analytical solution 
\begin{equation}
    C_{ij}\left(Y,T\right) = \boldsymbol{C}_{i}\left(T\right)\exp\left(-\frac{\Lambda R_{ij}\left(T\right)}{\mathbf{Q}_{ij}}Y\right), \quad 0\leq Y \leq A_{ij}.
    \label{adv_single}
\end{equation}
We write $\tilde{C}_{ij}\left(T\right) := C_{ij}\left(A_{ij},T\right)$ to represent the foulant concentration flowing from $v_i$ into an adjacent vertex $v_j$. Define 
\begin{equation}
\mathbf{B}_{ij}\left(T\right) := \exp\left(-\frac{\Lambda R_{ij}\left(T\right)A_{ij}}{\mathbf{Q}_{ij}}\right),
\label{eq:conc_factor}
\end{equation}
as the multiplicative factor by which foulant concentration changes over the length $A_{ij}$ of the edge. By conservation of particle flux at each vertex $v_i$ (suppressing the temporal and spatial dependence for simplicity of notation),
\begin{align}
0 &= -\overset{\text{outgoing}}{\overbrace{\sum_{v_{k}:\left(v_{i},v_{k}\right)\in E}\mathbf{Q}_{ik}\boldsymbol{C}_{i}}}+\overset{\text{incoming}}{\overbrace{\sum_{v_{k}:\left(v_{i},v_{k}\right)\in E}\mathbf{Q}_{ki}\tilde{C}_{ki}}},\quad v_{i}\in V\backslash V_{\rm top},\label{cons_part_flux}\\
&=-\boldsymbol{C}_{i}\sum_{v_{k}:\left(v_{i},v_{k}\right)\in E}\mathbf{Q}_{ik}+\sum_{v_{k}:\left(v_{i},v_{k}\right)\in E}\mathbf{Q}_{ki}\tilde{C}_{ki}, \nonumber \\
&=-\boldsymbol{C}_{i}\sum_{v_{k}:\left(v_{i},v_{k}\right)\in E}\mathbf{Q}_{ki}+\sum_{v_{k}:\left(v_{i},v_{k}\right)\in E}\mathbf{Q}_{ki}\mathbf{B}_{ki}\boldsymbol{C}_{k}, \label{def_tilde_C_applied}\\
\boldsymbol{C}_{i}\left(T\right) &= C_{0},\quad v_{i}\in V_{{\rm top}},\quad\forall T\geq 0,
\label{Conc_BC}
\end{align}
where $C_0$ is the constant foulant concentration in the feed solution. Note the conservation of flux used here to derive \cref{def_tilde_C_applied} follows a revised definition of the flux matrix $\mathbf{Q}$ (see \cref{eq:mod_Q} in \cref{app:B}). The first term in \cref{cons_part_flux} accounts for all outgoing particle flux from vertex $v_i$, the second for all incoming flux. Conservation of flux at each $v_i\in V\backslash V_{\rm top}$ and the definition of $\tilde{C}_{ij}$ are applied in \cref{def_tilde_C_applied}, which can be written more compactly using $\mathbf{Q}$ as a weight matrix and the in-degree weighted graph Laplacian $\mathbf{L}_\mathbf{Q}^{\rm in}:l^2\left(V\backslash V_{\rm top}\right)\rightarrow l^2\left(V\backslash V_{\rm top}\right)$,
\begin{equation}
\mathbf{L}_\mathbf{Q}^{\rm in}:=\mathbf{D}_{\mathbf{Q}^{\rm T}}-\left(\mathbf{Q}\circ {\mathbf B}\right)^{\rm T},
\label{advection_laplacian}
\end{equation}
where $\circ$ is the Hadamard product (element-wise matrix product). This operator is similar to \cref{graph_laplacian} but with a notion of directionality imparted by $\mathbf{Q}$ (also discussed by Chapman \& Mesbahi~\cite{chapman_mesbahi}). The term $\mathbf{D}_{\mathbf{Q}^{\rm T}}$ is a $\mathbf{Q}^{\rm T}$-weighted degree matrix (per \cref{degree}) that accounts for total incoming (upstream) flux of $v_i$ (thereby ``in-degree''). The last term, which contributes to the off-diagonal entries of $\mathbf{L}_\mathbf{Q}^{\rm in}$, displays the incoming flux reduced (multiplicatively) by the foulant deposition effect reflected in ${\mathbf B}$. Altogether, the transport equation for foulant concentration can be concisely written as
\begin{align}\label{adv_graph}
\mathbf{L}_\mathbf{Q}^{\rm in}\boldsymbol{C} &=\left(\mathbf{Q}\circ {\mathbf B}\right)^{\rm T}\boldsymbol{C}_0, \quad \boldsymbol{C} \in l^2\left(V\backslash V_{\rm top}\right), \quad \forall T\geq 0, \\
\boldsymbol{C}_0\left(T\right) &= \left(C_0,\ldots,C_0,0,\ldots,0\right)^{\rm T},
\end{align}
with boundary condition $C_0$ specified for each $v\in V_{\rm top}$ (the number of nonzero entries $C_0$ in $\boldsymbol{C}_0$ is equal to $\left|V_{\rm top}\right|$, per \cref{Conc_BC}). By solving this linear system, we obtain the particle concentration $\boldsymbol{C}_{i}$ at each vertex $v_i \in V\backslash V_{\rm top}$.

We assume that the pore radius shrinkage due to adsorption follows a simple model (used by Sanaei \& Cummings~\cite{sanaei2018flow}, among others) such that its rate of change is proportional to local particle concentration. With the approximation introduced above, that pore radius $R_{ij}$ varies only modestly with distance $Y$ along the pore, we assume that its rate of shrinkage is determined by the foulant concentration at the upstream\footnote{This concentration is chosen because it will dictate, in practice, the fastest shrinkage rate at the upstream end of the pore, and pore resistance is dominated by the narrowest pore radius.} vertex $v_i$,
\begin{equation}
    \frac{dR_{ij}}{dT}=-\Lambda\alpha \boldsymbol{C}_{i}, \quad R_{ij}\left(0\right) = R_{ij,0}, \quad \left(v_i,v_j\right)\in E,
    \label{adsorption}
\end{equation}
where $\alpha$ relates to foulant particle volume, see \cref{app:A}. We also assume that all radii initially take the same value, i.e. $R_{ij,0} = R_0$. This final equation \cref{adsorption} closes the membrane filtration model with adsorption.

Our assumption that pore radius is spatially uniform is motivated by significant computational cost savings; we have verified that the simulations reported later in the manuscript are more than 100 times faster compared to those where this simplification is not used. The price to pay for this saving is a slight overestimate of the membrane resistance, and therefore underestimate of the total throughput. To confirm that the approximation introduced by the model has only a minor influence, we have directly compared a subset of results obtained using this model with results obtained using the full model that allows $R_{ij}(Y,T)$, finding that the difference in throughput is on average around $10\%$ at maximum (and much smaller than 10\% where the maximum pore radius $D$ is small compared to the membrane thickness $W$). The differences in computed values of the particle concentration are smaller still.  We consider that this is reasonable, in particular since the approximation allows us to perform efficiently a large number of simulations and explore the influence of variability of results with respect to the model parameters of interest.

In simulating filtration through a network we impose a {\it stopping criterion} that membrane filtration ends when there exist no flow paths between any vertices in $V_{\rm top}$ and $V_{\rm bot}$, due to pore closures (individually characterized by $R_{ij} = 0$). The criterion is checked by a pathfinding algorithm. We terminate the filtration at the earliest time $T_{\rm final}$ that the criterion is satisfied. Note that in physical membranes, even when adsorptive foulants have accumulated to the extent that the pore is essentially closed, leakage through the pore may still take place; we have not included such effects in the present model.

\subsection{Measures of Performance\label{sec:2.7}}

Volumetric throughput of a membrane filter over its operational lifetime is a commonly-used measure of overall efficiency. 
\begin{definition}The volumetric throughput $V(T)$ through the filter is defined by
\[
    H\left(T\right)=\int_{0}^{T}Q_{\rm out}\left(T^{\prime}\right)dT^{\prime}, \quad Q_{{\rm out}}\left(T\right)=\sum_{v_{j}\in V_{{\rm bot}}}\sum_{v_{i}:\left(v_{i},v_{j}\right)\in E}\mathbf{Q}_{ij}\left(T\right)
\]
where $Q_{{\rm out}}\left(T\right)$ is the total flux exiting the filter.
\end{definition}
In particular, we are interested in $H_{\rm final}:=H\left(T_{\rm final}\right)$, the total volume of filtrate processed by the filter over its lifetime. 

Another performance measure is the accumulated foulant concentration in the filtrate, which captures the aggregate particle capture efficiency of the filter.
\begin{definition} The accumulated foulant concentration is defined by 
\[
    C_{{\rm acm}}\left(T\right)=\frac{\int_{0}^{T}C_{{\rm out}}\left(T^{\prime}\right)Q_{{\rm out}}\left(T^{\prime}\right)dT^{\prime}}{\int_{0}^{T}Q_{{\rm out}}\left(T^{\prime}\right)dT^{\prime}},
\]
where 
\[
C_{{\rm out}}\left(T\right)=\frac{{\displaystyle \sum_{v_{j}\in V_{{\rm bot}}}\sum_{v_{i}:\left(v_{i},v_{j}\right)\in E}}\boldsymbol{C}_{j}\left(T\right)\mathbf{Q}_{ij}\left(T\right)}{Q_{{\rm out}}\left(T\right)}.
\]
\end{definition}
Of particular interest is $C_{\rm acm}\left(T_{\rm final}\right)$, which provides a measure of the aggregate particle capture efficiency of the filter over its lifetime. 

\section{Nondimensionalisation}\label{sec:3}

We nondimensionalise the model presented in \cref{sec:2} with the following scales,
\begin{equation}\label{eq:scaling}
    \begin{gathered}
    \boldsymbol{P}=P_{0}\boldsymbol{p},\qquad\boldsymbol{X}=W\boldsymbol{x},\qquad A_{ij}=Wa_{ij},\\
    \left(D,D_{{\rm min}}\right)=W\left(d,d_{{\rm min}}\right),\qquad R_{ij}=Wr_{ij},\qquad R_{0}=Wr_{0},\\
    \mathbf{Q}_{ij}=\frac{\pi W^{3}P_{0}}{8\mu}\mathbf{q}_{ij},\qquad\mathbf{K}_{ij}=\frac{\pi W^{3}}{8\mu}\mathbf{k}_{ij},\qquad\mathbf{k}_{ij}=\frac{r_{ij}^{4}}{a_{ij}},\\
    \boldsymbol{C}=C_{0}\boldsymbol{c},\qquad Y=Wy,\qquad\Lambda=\frac{\pi WP_{0}}{8\mu}\lambda,\qquad T=\frac{W}{\Lambda\alpha C_{0}}t,\qquad V=\frac{W^{3}}{\alpha C_{0}}v.
    \end{gathered}
\end{equation}
Under these scalings we derive dimensionless equations for pressure $\boldsymbol{p}$ and flux $\mathbf{q}$,
\begin{subequations}
\begin{align}
\mathbf{L_k} \boldsymbol{p}&= 0, \label{dimless_laplace}\\
\boldsymbol{p}\left(v\right) &= 1, \quad \forall v\in V_{\rm top}; \quad \boldsymbol{p}\left(v\right) = 0, \quad \forall v\in V_{\rm bot}, \label{dimless_laplace_bc}\\
\mathbf{q}_{ij}&=\mathbf{k}_{ij}\left(\boldsymbol{p}\left(v_i\right)-\boldsymbol{p}\left(v_j\right)\right), \quad \forall \left(v_i,v_j\right) \in E, \label{dimless_hp}
\end{align}
\end{subequations}
where $\mathbf{L_k}$ is defined in \cref{graph_laplacian}; for foulant concentration $\boldsymbol{c}$,
\begin{subequations}
\begin{align}
\mathbf{L}_\mathbf{q}^{\rm in}\boldsymbol{c} &=\left(\mathbf{q}\circ {\mathbf b}\right)^{\rm T}\boldsymbol{c}_0, \quad \mathbf{L}_\mathbf{q}^{\rm in}=\mathbf{D}_{\mathbf{q}^{\rm T}}-\left(\mathbf{q}\circ {\mathbf b}\right)^{\rm T}, \quad  \boldsymbol{c} \in l^2\left(V\backslash V_{\rm top}\right),\label{dimless_conc}\\
\boldsymbol{c}_0 &= \left(1,\ldots,1,0,\ldots,0\right)^{\rm T}, \quad \mathbf{b}_{ij} = \exp\left(\frac{-\lambda r_{ij} a_{ij}}{\mathbf{q}_{ij}}\right),\label{dimless_conc_bc}
\end{align}
\end{subequations}
where $\mathbf{L}_\mathbf{q}^{\rm in}$ is given by \cref{advection_laplacian}; and for pore radius $r_{ij}$ (for the pore $e_{ij}=(v_i,v_j)$),
\begin{equation}
    \frac{dr_{ij}}{dt}=-\boldsymbol{c}_{i}, \quad r_{ij}\left(0\right) = r_0, \quad \forall \left(v_i,v_j\right)\in E.
\label{dimless-adsorption}
\end{equation}
The dimensionless throughput is given by 
\begin{equation}
 h\left(t\right)=\frac{1}{\lambda}\int_{0}^{t}q_{\rm out}\left(t^{\prime}\right)dt^{\prime}, \quad q_{\rm out}\left(t\right) = \sum_{v_{j}\in V_{{\rm bot}}}\sum_{v_{i}:\left(v_{i},v_{j}\right)\in E}\mathbf{q}_{ij}\left(t\right),
    \label{dimless_thruput}
\end{equation}
and $h_{\rm final}:= h\left(t_{\rm final}\right)$. Dimensionless accumulated foulant concentration is written as
\begin{equation}
    c_{{\rm acm}}\left(t\right)=\frac{\int_{0}^{t}c_{{\rm out}}\left(t^{\prime}\right)q_{{\rm out}}\left(t^{\prime}\right)dt^{\prime}}{\int_{0}^{t}q_{{\rm out}}\left(t^{\prime}\right)dt^{\prime}},
    \label{c_acm}
\end{equation}
where 
\[
c_{{\rm out}}\left(t\right)=\frac{{\displaystyle \sum_{v_{j}\in V_{{\rm bot}}}\sum_{v_{i}:\left(v_{i},v_{j}\right)\in E}}\boldsymbol{c}_{j}\left(t\right)\mathbf{q}_{ij}\left(t\right)}{q_{{\rm out}}\left(t\right)}.
\]
\section{Algorithm}
\label{sec:4}
We summarize the network generation protocol and solution technique for the proposed model equations in \cref{algo}. We refer the reader to a worked example in \cref{app:C} for a better visualisation of how the governing equations evolve on a very simple network and how our chosen performance metrics depend on the model parameters.
\begin{algorithm}[ht]
\caption{\bf Filtration with adsorption.}
\label{algo}
\begin{enumerate}
\item Initialization
    \begin{enumerate}
    \item Generate uniformly $N_{\rm total}$ random points in box $\left[0,1\right]^{2}\times \left[-0.5,1.5\right]$. 
    \item Connect all points separated by a distance smaller than $d$ but larger than $d_{\rm min}$ using the metric $\chi$ (periodic or isolated). Discard isolated points.
    \item Truncate the box with cutting planes at $x_{3} =0$ and $x_{3} =1$; record points with $0 < x_{3} < 1$ as the set of interior points $V_{\rm int}$; label intersections between cutting planes and edges as inlets $V_{\rm int}$ for intersections at $x_{3}=0$ and outlets $V_{\rm out}$ for those at $x_{3}=1$. Set $V = V_{\rm int} \cup V_{\rm top} \cup V_{\rm bot}$.
    \item Initialise radii $r_{ij,0}$ for $\left(r_{ij}\right)_{\left(v_i,v_j\right)\in E}$.
    \end{enumerate}

\item Fluid Flow
    \begin{enumerate}
        \item Find pressures $\boldsymbol{p}$ and fluxes $\mathbf{q}$ by solving \cref{dimless_laplace}-\cref{dimless_hp}.
    \end{enumerate}
\item Foulant Concentration
    \begin{enumerate}
        \item Initialise concentrations $\boldsymbol{c}_0$ (per \cref{dimless_conc_bc}) for $v_i \in V_{\rm top}$.
        \item Find foulant concentration $\boldsymbol{c}$ by solving \cref{dimless_conc}-\cref{dimless_conc_bc}.
    \end{enumerate}

\item Adsorption 
    \begin{enumerate}
        \item Evolve each pore radius $r_{ij}$ via \cref{dimless-adsorption} until it decreases to $0$ and stays at $0$ for the rest of the filtration.
    \end{enumerate}
\item Check {\it stopping criterion}. Stop if satisfied.
\item Compute throughput using \cref{dimless_thruput}.
\item Update conductance $\mathbf{k}$ and weighted graph Laplacian $\mathbf{L_k}$.
\item Increment time $t\to t+dt$ and return to (2). 
\end{enumerate}

\end{algorithm}

\section{Results}
\label{sec:5}
In this section, we present numerical results on how filter performance metrics, such as total throughput and accumulative foulant concentration (introduced in \cref{sec:2.7}), depend on filter geometry, characterized by three main geometric network parameters detailed below. For each metric/parameter pair considered we first describe the observed trends and compare the results between the isolated and periodic network configurations. Then, we discuss how the findings are related, with a focus on identifying universal parameters that describe filter performance.

\subsection{Geometric Network Parameters}\label{sec:5.1}
To investigate different network architectures, we introduce three geometric parameters. The first is the average number of neighbors for interior vertices/junctions, 
\[
N := \frac{\left|E\right|}{\left|V_{\rm int}\right|}.
\]
This parameter provides one way to characterize the connectivity strength of a network. We compute $N$ only over interior vertices, as any inlet or outlet has exactly one neighbor due to our network generation protocol. Second, the initial void volume, 
\begin{equation}
{\rm Vol}_0 = \frac{\pi}{2} \sum_{\left(v_i,v_j\right)\in E} a_{ij}r_{ij}^2\left(t=0\right) = \frac{\pi r_0^2}{2} \sum_{\left(v_i,v_j\right)\in E} a_{ij},
\label{ini_vol}
\end{equation}
provides an estimate of how much fluid the filter can process initially, and how much surface area is available for foulant adsorption (since the initial pore radius is the same for all pores, the initial total pore surface area is given by $2 {\rm Vol}_0/r_0$). Since the representative membrane volume we consider is a unit cube in our dimensionless framework, ${\rm Vol}_0$ is exactly the initial membrane {\it porosity.} Lastly, we consider tortuosity $\tau$, defined by the {\it average length of paths} (defined with respect to flow paths) that connect the top and bottom membrane surface through the network, normalized by the membrane thickness (the shortest possible path). A rigorous definition of $\tau$ is given in the appendix (\cref{tor_def} and \cref{tort_formula}).

Our input parameters $\left(d,N_{\rm total}\right)$ (see \cref{algo}) are chosen so that the volume and area constraints (initial total pore volume should not exceed that of the representative unit membrane cube, and total pore cross-sectional area on membrane top and bottom surfaces should not exceed the area of the unit square, see Section 2.2) are not violated, and so that the generated graphs are nontrivially connected from top to bottom surfaces. The ranges of $d$ and $N_{\rm total}$ are $\left[0.1,0.45\right]$ and $\left[100,5000\right]$ respectively.  In all simulations we fix $r_0 = 0.01$, $d_{\rm min} = 0.06$, $\lambda = 5\times 10^{-7}$ and produce data curves for each chosen $d$-value (represented by the same marker) by varying $N_{\rm total}$. 

Due to the random nature of our network generation, we compute the average of quantities of interest over a number of realizations. For simplicity, we do not use additional notation to indicate such averaged quantities. For all results below, each data point is obtained by averaging over $500$ simulations, i.e. $500$ realizations of a random graph with parameter pair $\left(d,N_{\rm total}\right)$. We histogram these $500$ realizations in the supplement for each pair $\left(d,N_{\rm total}\right)$, to demonstrate that their variance is influenced only by geometric parameters, justifying the sufficiency of this number of realizations. All results are shown for the two cases in which the underlying random graphs are generated with isolated and periodic boundary conditions, with results (data points in blue and red respectively) compared side-by-side.

\subsection{Initial Void Volume and Average Number of Neighbors}
In this section, we present results showing the dependence of performance metrics on two geometric parameters: initial void volume ${\rm Vol}_0$ and average number of neighbors $N$. 

\begin{figure}[!ht]
    \centering
    \subfloat[]{\label{fig:tt_vol_np}\includegraphics[width=.47\textwidth,height=.4\textwidth]{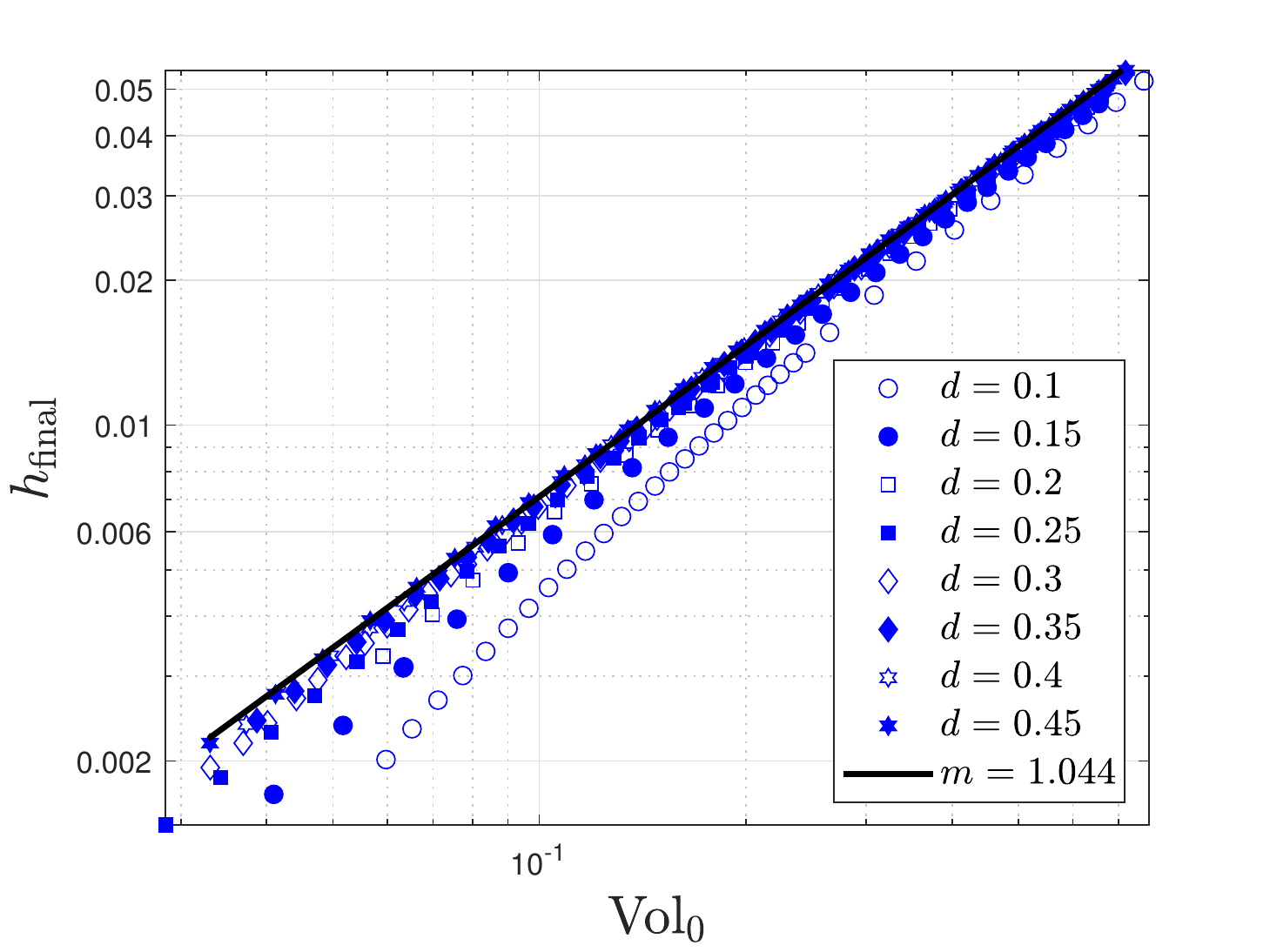}}
    \subfloat[]{\label{fig:tt_vol_p}\includegraphics[width=.47\textwidth,height=.4\textwidth]{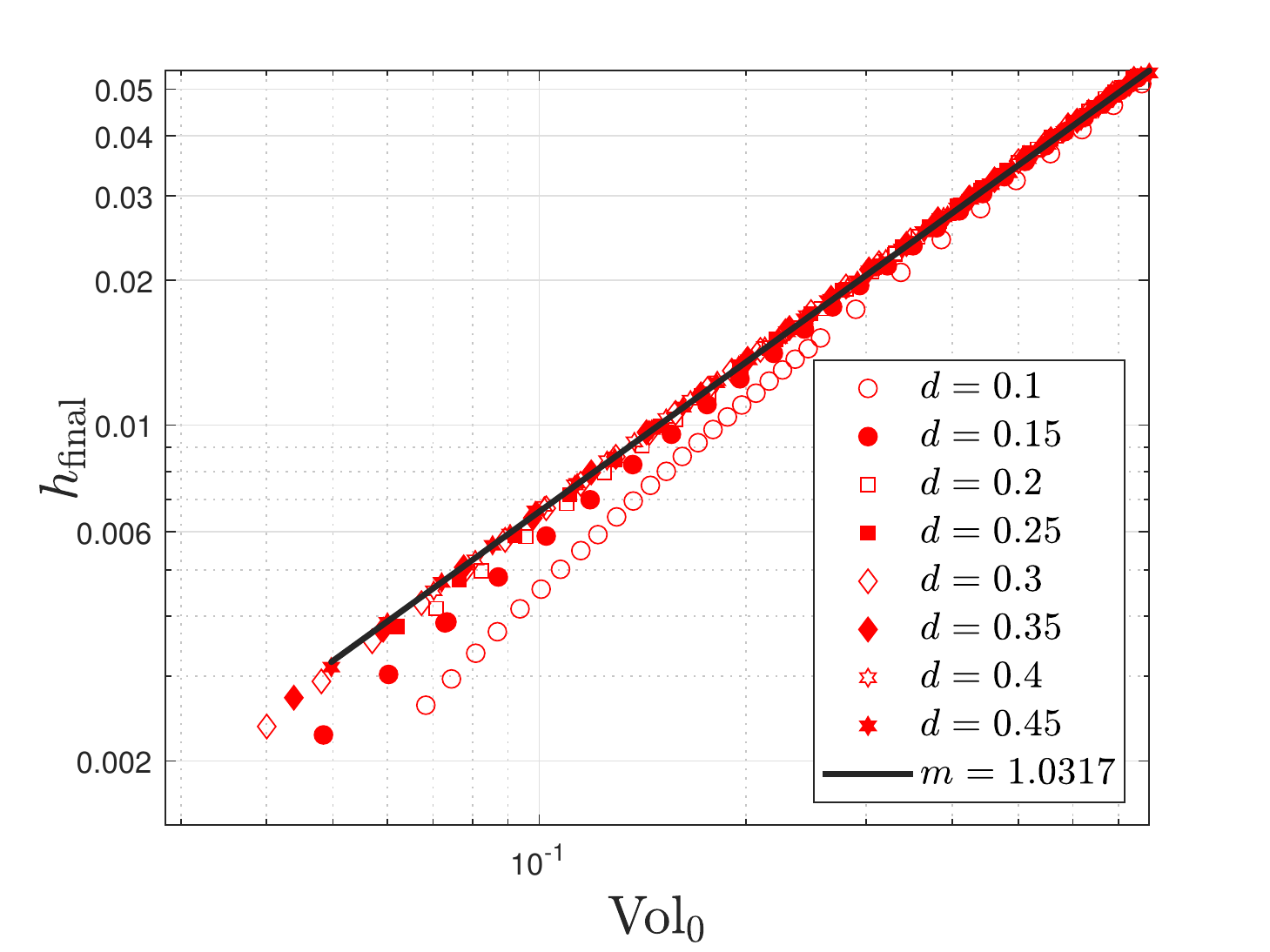}}
    \caption{Total throughput $h_{\rm final}$ vs initial void volume Vol$_0$ (loglog scales). (a) {\color{blue}Isolated network setup}; (b) {\color{red}periodic setup}. Line of best fit for $d=0.45$ is in black, with gradient $m$ given in legend (with $R^2 = 0.99989$ and $0.99993$ respectively). Distribution of error for each data point is given in the histograms in the supplement.}
    \label{fig:tt_vol}
\end{figure}
\cref{fig:tt_vol} shows total throughput $h_{\rm final}$ against initial void volume ${\rm Vol}_0$ on a log-log scale, for various $d$-values. Both quantities are increasing functions of the initially-specified number of random points, $N_{\rm total}$, for each $d$. In both isolated and periodic cases we see a power law relationship emerge for sufficiently large ${\rm Vol}_0$, with good collapse of the data for ${\rm Vol}_0\gtrsim 0.2$, $d\gtrsim 0.2$ onto a single line, hinting at a universal law. The fact that $h_{\rm final}$ is an increasing function of ${\rm Vol}_0$ makes sense because the larger the initial pore volume, the more filtrate can be processed. In both setups, the data curves for smaller $d$-values begin to deviate from the universal power law for ${\rm Vol}_0\lesssim 0.3$, which enables us to observe some hierarchy. For similar ${\rm Vol}_0$-values, corresponding to similar total length of edges per \cref{ini_vol}, we see that larger $d$-values lead to larger total throughput $h_{\rm final}$. We attribute this to the fact that networks with larger $d$-values induce more inlets on the upstream (and downstream) membrane surfaces, allowing more filtrate to pass through. The differences between the isolated and periodic setups in \cref{fig:tt_vol} are minor.

\begin{figure}[!ht]
    \centering
    \subfloat[]{\label{fig:tt_nbar_np}\includegraphics[width=.47\textwidth,height=.4\textwidth]{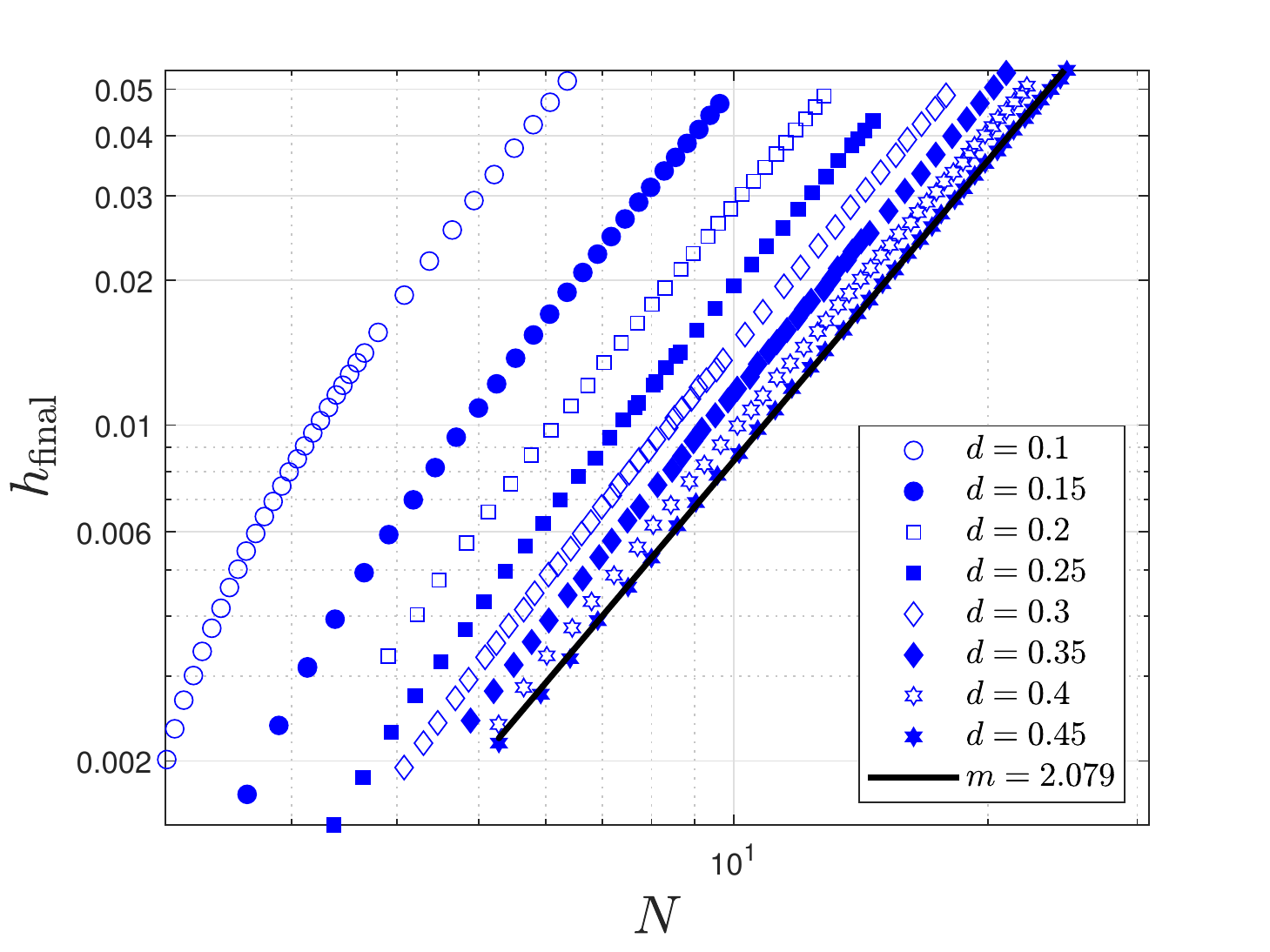}}
    \subfloat[]{\label{fig:tt_nbar_p}\includegraphics[width=.47\textwidth,height=.4\textwidth]{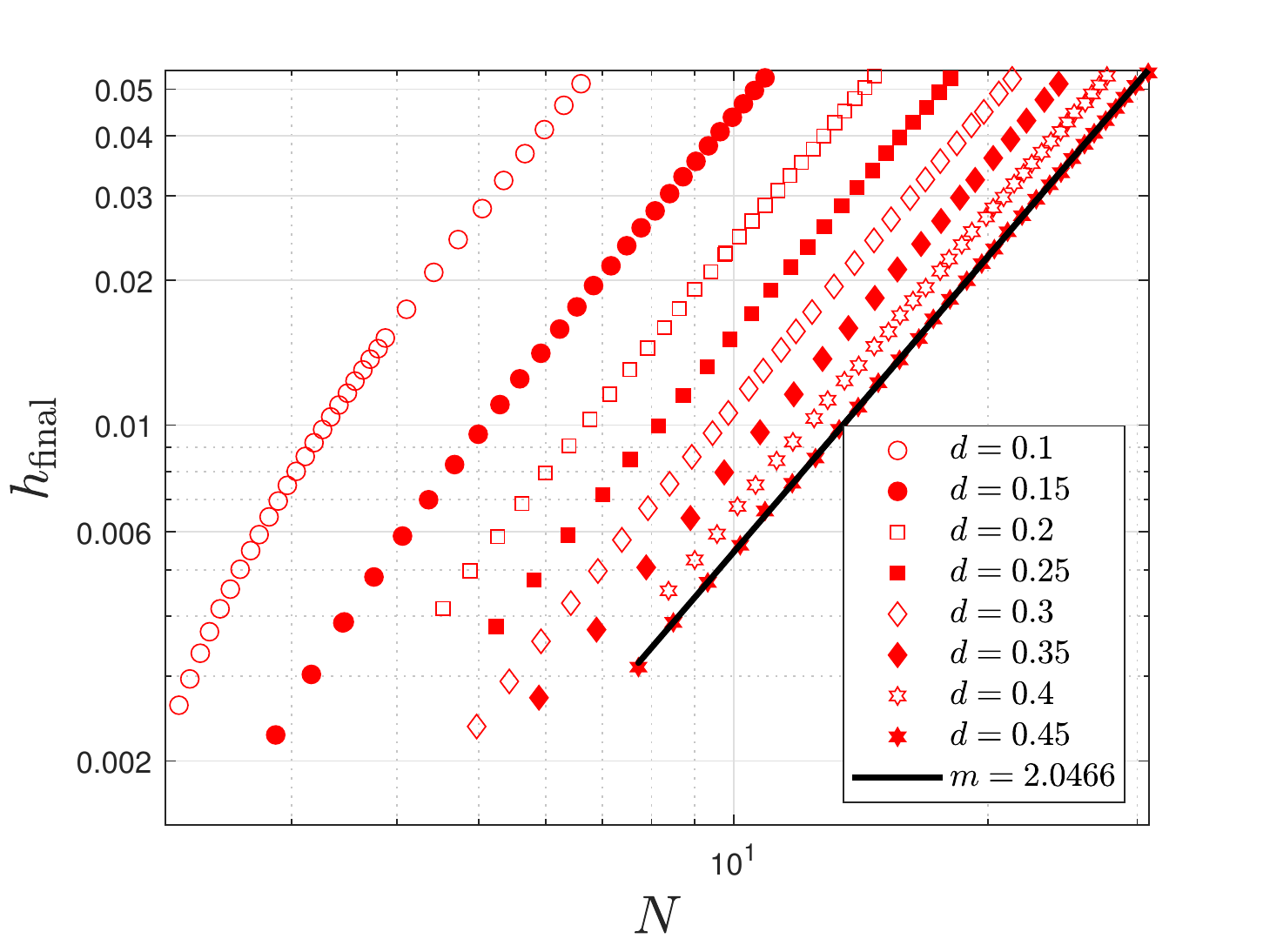}}\\
    \subfloat[]{\label{fig:tt_volume_scaled_nbar_np}\includegraphics[width=.47\textwidth,height=.4\textwidth]{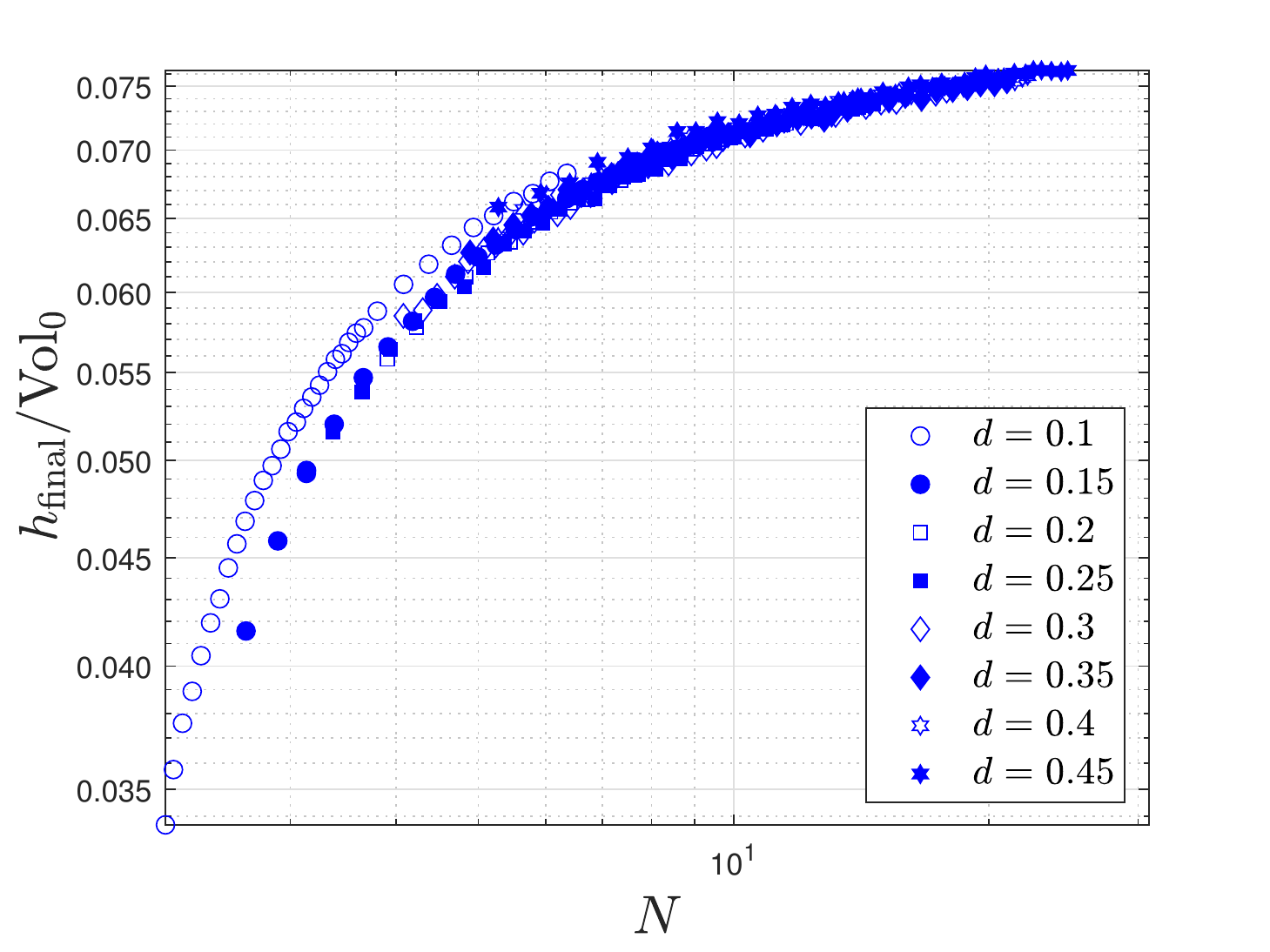}}
    \subfloat[]{\label{fig:tt_volume_scaled_nbar_p}\includegraphics[width=.47\textwidth,height=.4\textwidth]{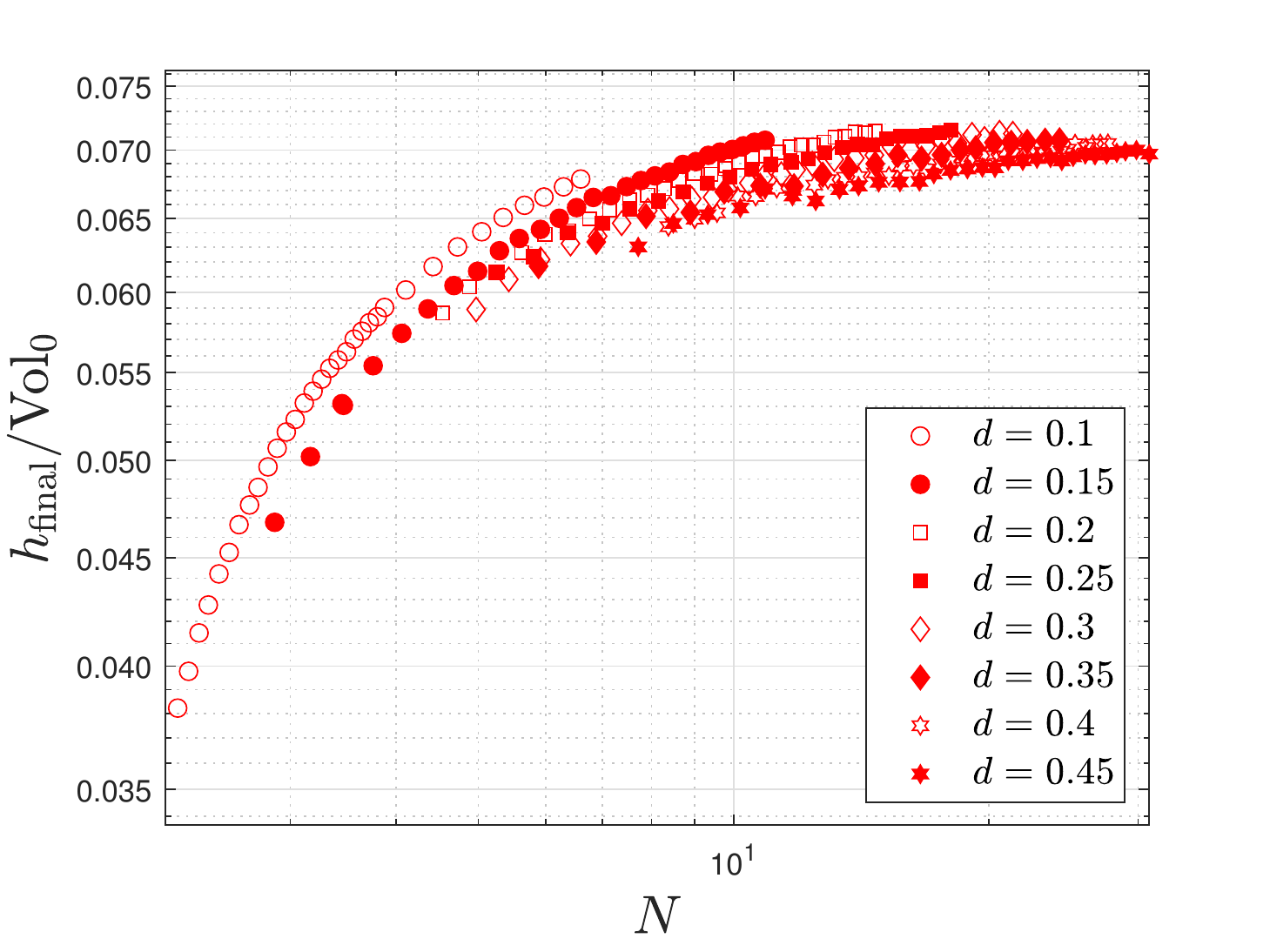}}
    \caption{Total throughput $h_{\rm final}$ vs average number of neighbors $N$ (loglog scales). (a) {\color{blue}Isolated network setup}; (b) {\color{red}periodic setup}.  Line of best fit for $d=0.45$ is in black, with gradient $m$ given in legend. Distribution of error for each data point is given in the histograms in the supplement. Same setup for (c) and (d) with $h_{\rm final}/{\rm Vol}_0$ as vertical axis (loglog scales).}
    \label{fig:tt_nbar}
\end{figure}
In \cref{fig:tt_nbar_np} and \ref{fig:tt_nbar_p}, we plot total throughput $h_{\rm final}$ against average number of neighbors $N$ (the average number of pores entering or leaving each vertex) for various values of the search radius $d$, using a log scale for both axes. Each data point shown corresponds to a different choice of $N_{\rm total}$; and, as discussed above, represents results averaged over 500 individual random graph realizations. Several features are common to both plots: first, $h_{\rm final}$ is an increasing function of $N$ for each $d$ value and in particular, obeys a power law for sufficiently large $N$. For fixed $N$, the smaller the search radius $d$, the larger the average throughput $h_{\rm final}$. This suggests that the number of pore junctions plays an important role in controlling and predicting total throughput. Second, the slopes for each $d$-value at large $N$ are very similar for both isolated and periodic setups, suggesting a common power law relating $h_{\rm final}$ and $N$ for each $d$ (see $d=0.45$ in \cref{fig:tt_nbar}(a,b) with slope $\approx 2$ in both cases). Moreover, as $d$ increases the curves in \cref{fig:tt_nbar}(a,b) become more closely-spaced in both setups, demonstrating that filters with sufficiently large search radius perform similarly in terms of total throughput. In both cases the quantities $N$ and $h_{\rm final}$ are increasing functions of $N_{\rm total}$ for each $d$.

We also highlight the differences between the two compared network connection metrics in \cref{fig:tt_nbar_np} and \ref{fig:tt_nbar_p}. First, for small $d$ (e.g. $d = 0.1,0.15$), the isolated and periodic networks give similar throughputs since random graph generations under the two metrics do not differ greatly; however, as $d$ increases, the periodic metric generates more edges than the isolated one (for the same value of $N_{\rm total}$), hence the periodic setup yields higher throughput. By the same reasoning, for fixed $h_{\rm final}$, $N$ is much larger in the periodic setup than in the isolated one for each $d \gtrsim 0.3$ (the data curves in the isolated setup are seen to be more closely-spaced than those in the periodic setup as $d$ increases). In other words, though $N$ is an increasing function of $d$ in each setup, it increases at a higher rate under the periodic metric.

Motivated by the strong relationship between total throughput and volume found in \cref{fig:tt_vol}, we also plot volume-scaled total throughput, $h_{\rm final}/{\rm Vol}_0$, against average number of neighbors $N$ for both connection metrics in \cref{fig:tt_volume_scaled_nbar_np} and \ref{fig:tt_volume_scaled_nbar_p}. This volume-scaled total throughput can be understood as a measure of efficiency of the membrane filter in terms of filtrate production capability. In \cref{fig:tt_volume_scaled_nbar_np}, we observe a good collapse of data in the non-periodic setup, which suggests that the average number of neighbors is a good predictor for total throughput scaled by volume. We see a similar trend in \cref{fig:tt_volume_scaled_nbar_p} for the periodic setup, though the collapse is less strong. We further note that in both non-periodic and periodic setups, data for $d=0.1$ are slightly separated from the rest. This is consistent with the deviation of data for $d=0.1$ in \cref{fig:tt_vol} from the universal power law relating total throughput and volume.

\begin{figure}[!ht]
    \centering
    \subfloat[]{\label{fig:conc_vol_np}\includegraphics[width=.47\textwidth,height=.4\textwidth]{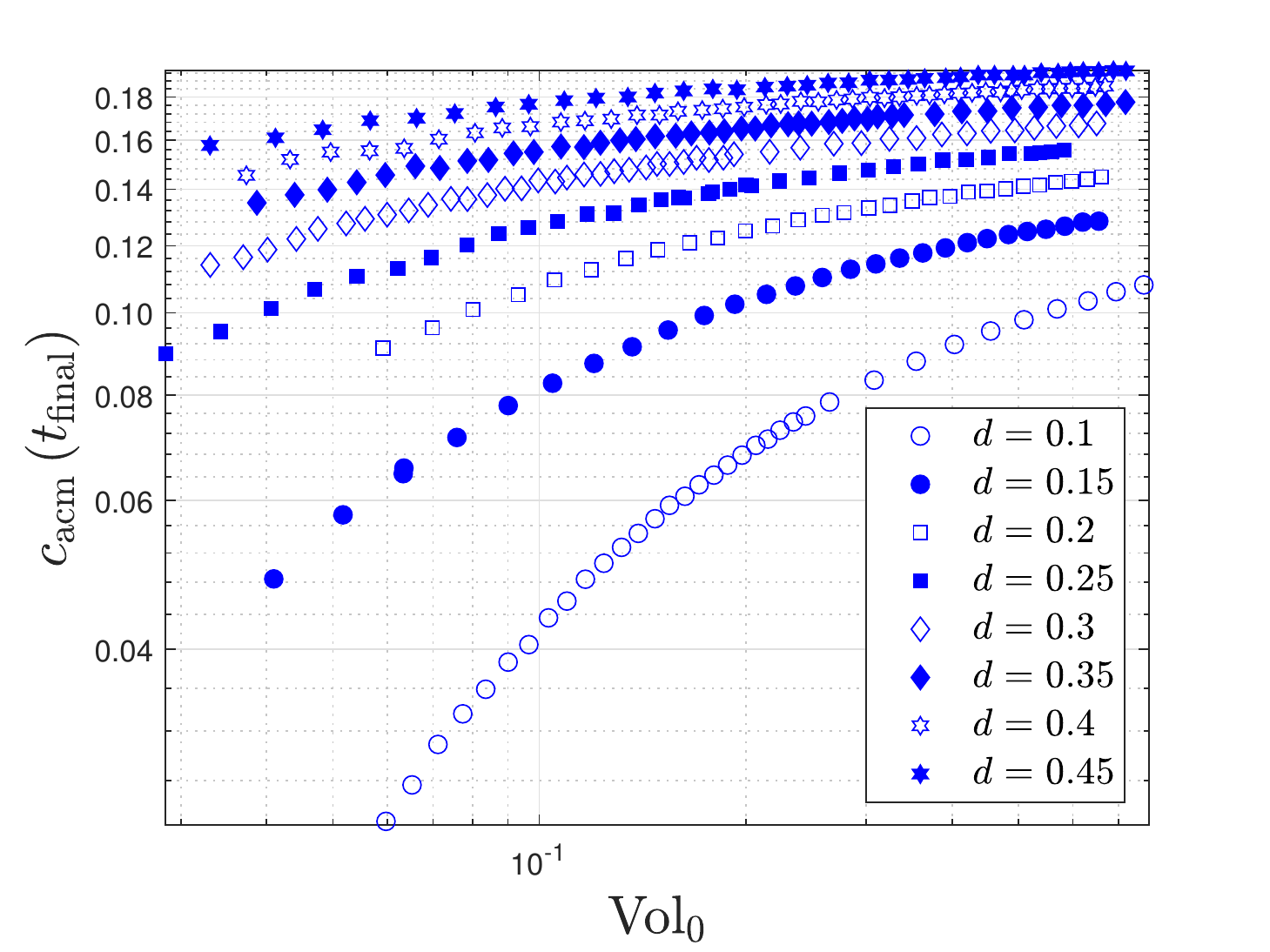}}
    \subfloat[]{\label{fig:conc_vol_p}\includegraphics[width=.47\textwidth,height=.4\textwidth]{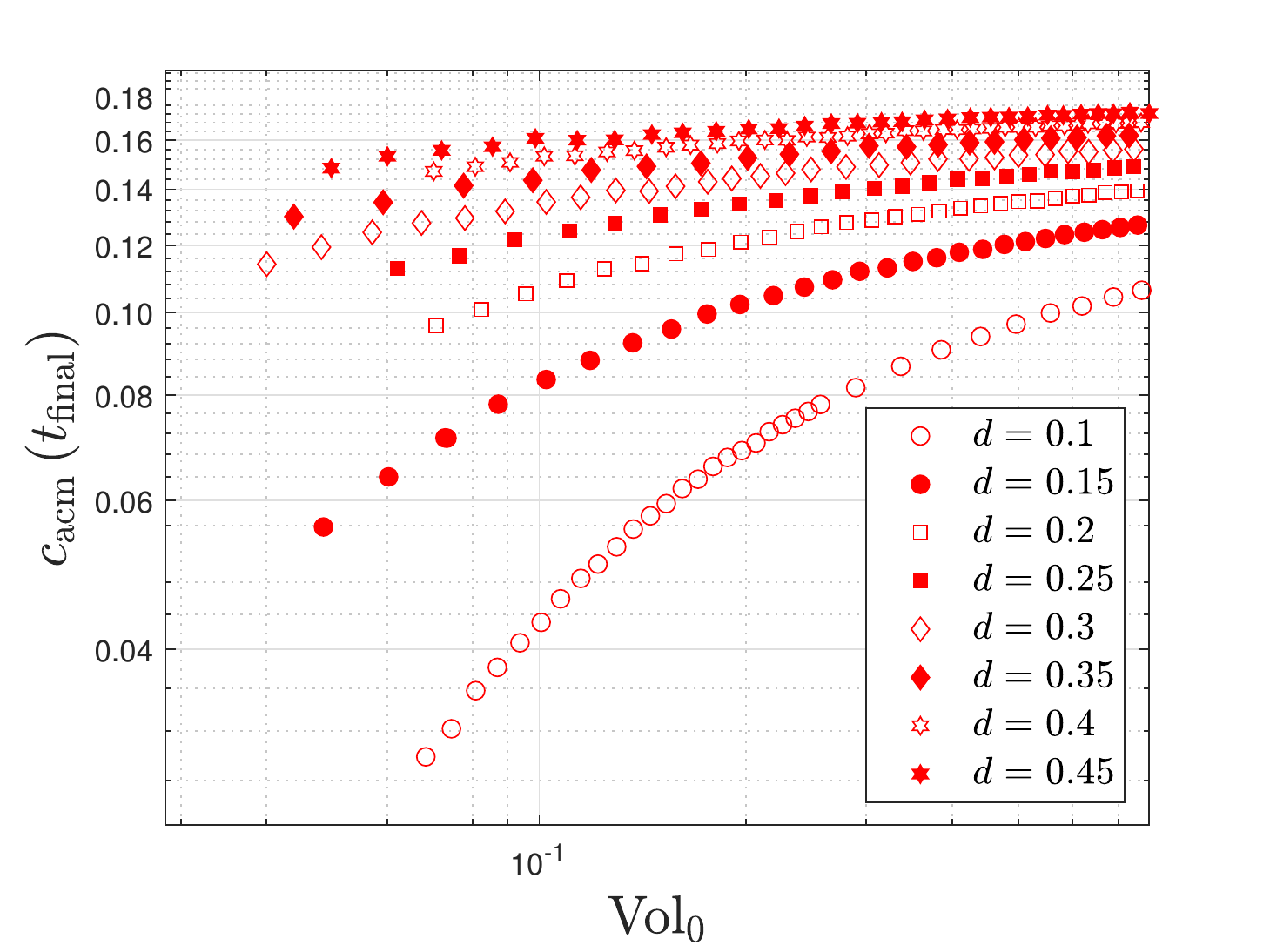}}
    \caption{Final accumulated foulant concentration $c_{\rm acm}(t_{\rm final})$ vs initial void volume Vol$_0$ (loglog scales). (a) {\color{blue}Isolated network setup}; (b) {\color{red}periodic setup}. Distribution of error for each data point is given in the histograms in the supplement.}
    \label{fig:conc_vol}
\end{figure}
\cref{fig:conc_vol} shows final accumulated foulant concentration in the filtrate, $c_{\rm acm}\left(t_{\rm final}\right)$, vs initial void volume, ${\rm Vol}_0$, on a log-log scale, as $d$ is varied. Both quantities are increasing functions of $N_{\rm total}$ for each $d$. The plots for both isolated and periodic network metrics show that $c_{\rm acm}\left(t_{\rm final}\right)$ increases as search radius $d$ increases for fixed ${\rm Vol}_0$. This is because with similar initial pore volume and thus similar initial total edge length, networks with larger $d$ have fewer edges to which foulants can adhere. Combining this finding with the result in \cref{fig:tt_vol} for a fixed volume further confirms the intuition that networks producing larger throughput (larger $d$) also have worse foulant control. A second trait shared by both network setups is that for larger $d$ (for example $d=0.4,0.45$), $c_{\rm acm}\left(t_{\rm final}\right)$ is relatively insensitive to changes in Vol$_0$, while for smaller $d$ (say $d=0.1,0.15$), the change in $c_{\rm acm}\left(t_{\rm final}\right)$ with Vol$_0$ is much more dramatic. This says that in membrane networks with longer pores, final accumulated foulant concentration depends less on initial void volume than in networks with shorter pores. We further observe that for a given ${\rm Vol}_0$, for small $d$, the values of $c_{\rm acm}\left(t_{\rm final}\right)$ are similar for both isolated and periodic configurations. However, as $d$ increases, the isolated case incurs larger $c_{\rm acm}\left(t_{\rm final}\right)$ and thereby exhibits worse foulant control than the periodic case. We defer further discussion of this observation to \cref{sec:5.3} below. Lastly, we find that in practice, when a tolerance level of contaminant concentration in the filtrate is specified in \cref{fig:conc_vol}, networks with smaller $d$ values have larger initial pore volume, which corresponds to larger total throughput via the strong relationship observed in \cref{fig:tt_vol} (e.g. if a threshold of $c_{\rm acm}\left(t_{\rm final}\right) = 0.06$ is set, filters with $d=0.1$ are preferable as they have larger initial pore volume and thus larger total throughput).

\subsection{Tortuosity\label{sec:5.3}}
We now examine the dependence of our performance metrics on initial network tortuosity $\tau$, the average distance travelled by a fluid particle from the top membrane surface to the bottom, before any fouling has occurred. A full characterization of $\tau$ is given by \cref{tor_def} and \cref{tort_formula} in the appendix.

\begin{figure}[!ht]
    \centering
    \subfloat[]{\label{fig:conc_tau_np}\includegraphics[width=.47\textwidth,height=.4\textwidth]{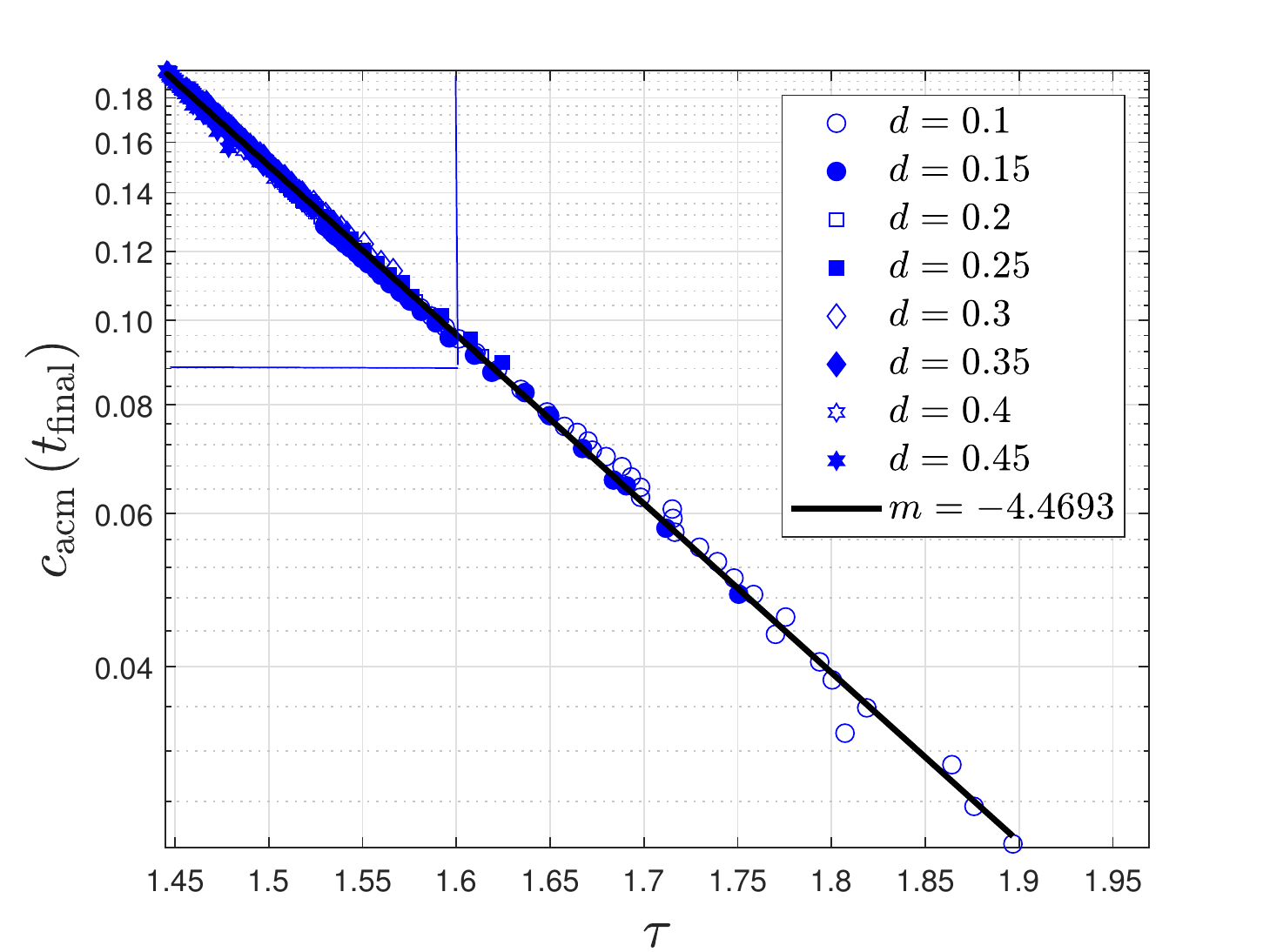}}
    \subfloat[]{\label{fig:conc_tau_p}\includegraphics[width=.47\textwidth,height=.4\textwidth]{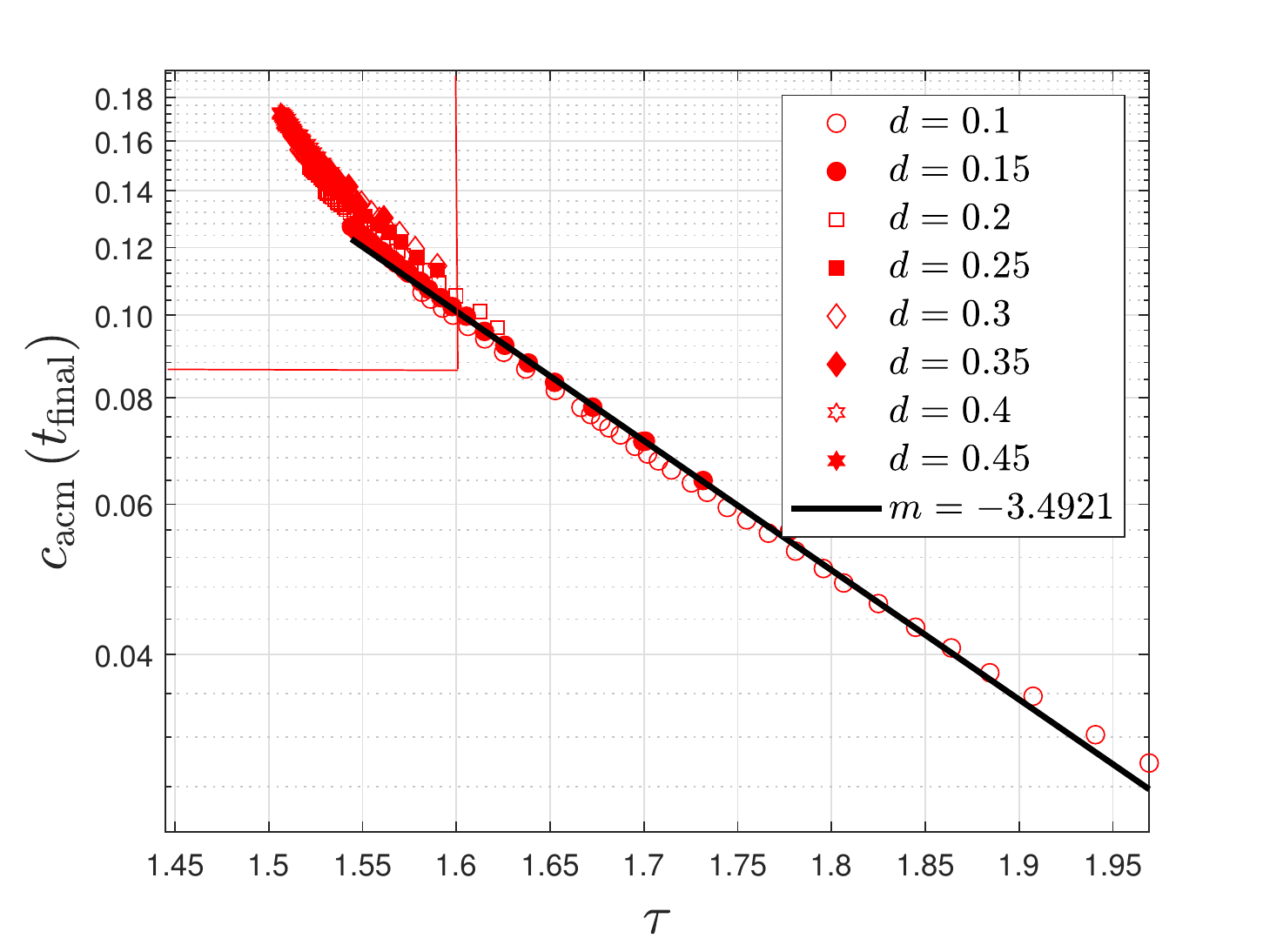}}\\
    \subfloat[]{\label{fig:conc_tau_np_zoom}\includegraphics[width=.47\textwidth,height=.4\textwidth]{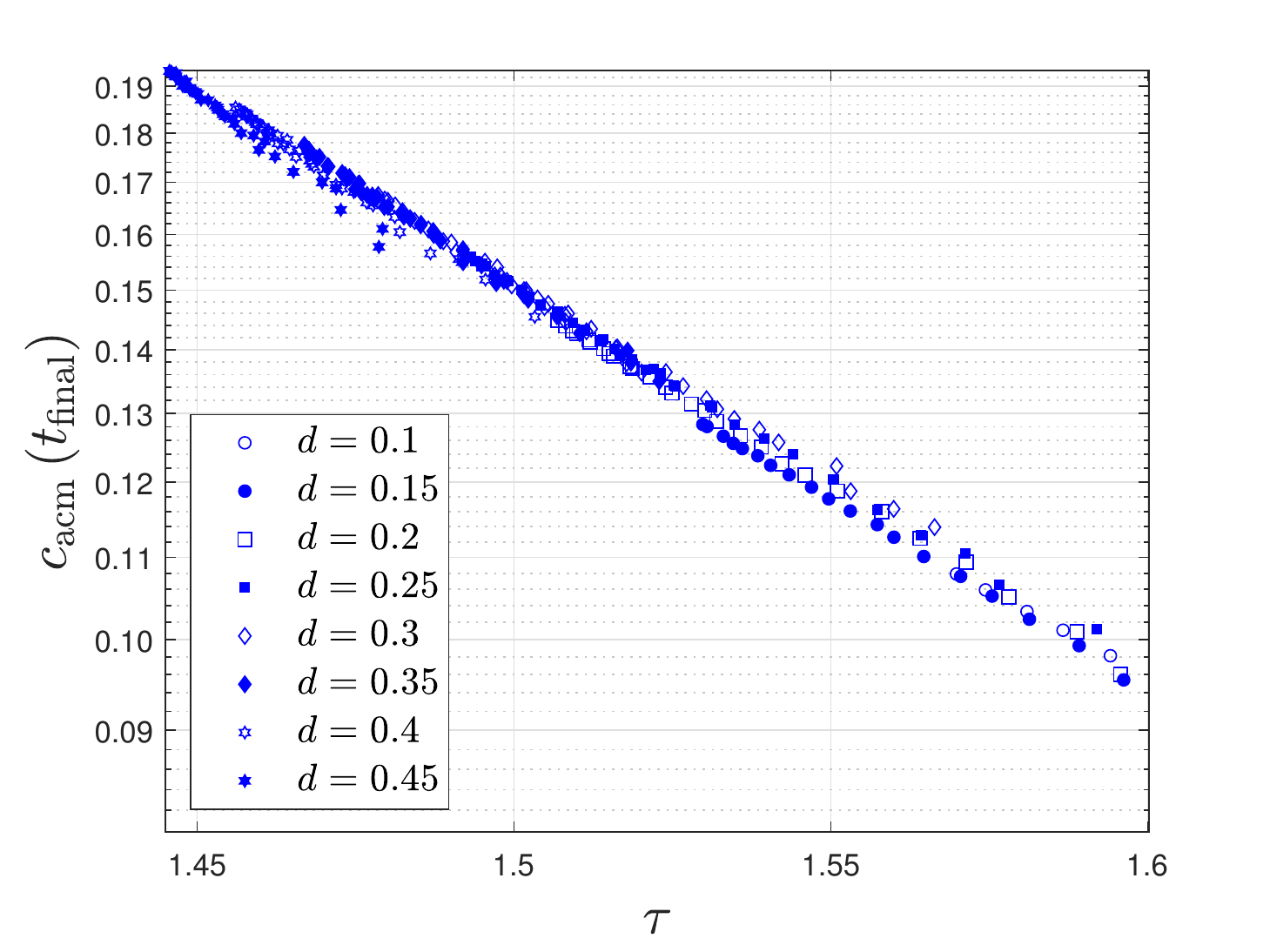}}
    \subfloat[]{\label{fig:conc_tau_p_zoom}\includegraphics[width=.47\textwidth,height=.4\textwidth]{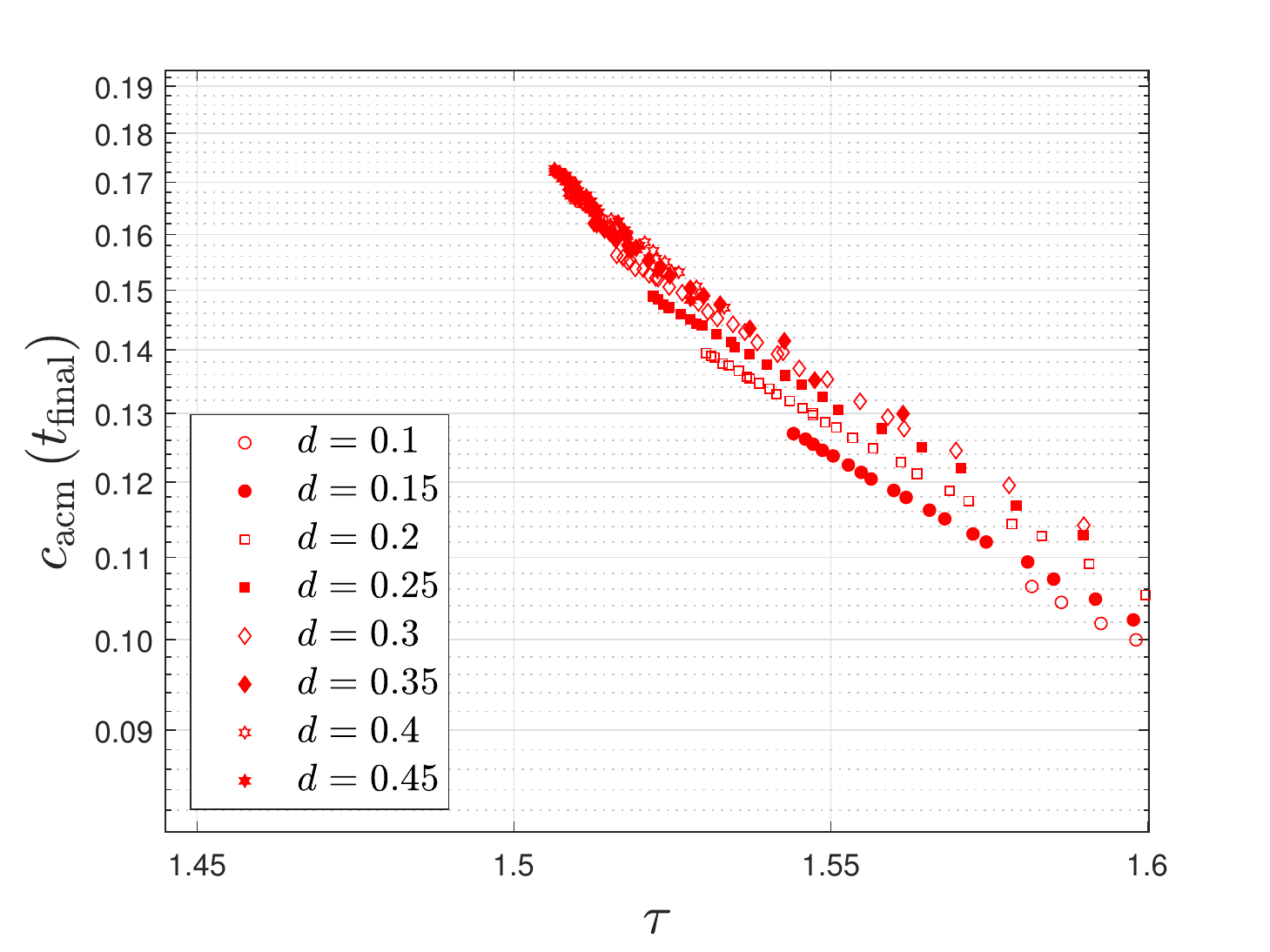}}
    \caption{Final accumulated foulant concentration $c_{\rm acm}\left(t_{\rm final}\right)$ vs tortuosity $\tau$ (semilog plot). (a) {\color{blue}Isolated network setup}; (b) {\color{red}periodic setup}. The line of best fit is in black in each plot, with gradient $m$ given in the legend (with $R^2 = 0.99838$ and $0.9961$ respectively). The {\color{blue}blue} and {\color{red}red} boxes at top left are shown as zooms in (c) and (d) respectively, for small tortuosity values (same data as (a) and (b) respectively). Distribution of error for each data point is given in the histograms in the supplement.}
    \label{fig:conc_tau}
\end{figure}
Our main result is that tortuosity is an important universal parameter that predicts accumulated foulant concentration independently of input parameters $\left(d,N_{\rm total}\right)$. \cref{fig:conc_tau_np} and \ref{fig:conc_tau_p} plot accumulated foulant concentration in the filtrate, $c_{\rm acm}\left(t_{\rm final}\right)$, vs tortuosity, $\tau$, for various values of the search radius $d$. We readily observe the similar collapse of data points in both cases, particularly strongly in the isolated case. We find that $c_{\rm acm}\left(t_{\rm final}\right)$ decays exponentially with $\tau$, with a negative exponent, which implies that networks with more winding flow-paths make much better filters in terms of foulant control. This is because the extent of deposition (per \cref{adv_single}) increases as distance travelled by the fluid increases. We also observe two further details: within the regime where the fit (solid black line) is strongest in the plots, the data for the isolated setup (\cref{fig:conc_tau_np}) has a larger (negative) slope than that in the periodic setup (\cref{fig:conc_tau_p}); and, for fixed $\tau$, filters in the periodic setup have slightly higher accumulated foulant concentration. 

\cref{fig:conc_tau_np_zoom} and \ref{fig:conc_tau_p_zoom} show zoomed plots of the data from \cref{fig:conc_tau_np} and \ref{fig:conc_tau_p} respectively, to probe the details of the low tortuosity regime, where accumulated foulant concentration is highest. To facilitate the discussion below, note that for each $d$, a larger value of $N_{\rm total}$ corresponds to a smaller tortuosity $\tau$ (i.e. adding more vertices to the network decreases average path length). \cref{fig:conc_tau_p_zoom} shows that for the periodic pore network we see an emergent nonlinear trend, breaking the exponential relation between $c_{\rm acm}(t_{\rm final})$ and $\tau$ as $N_{\rm total}$ is increased. Although the isolated setup in \cref{fig:conc_tau_np_zoom} appears to persist in its linear (exponential relationship) trend, we note that for yet larger values of $N_{\rm total}$ this trend must break down due to the existence of a lower bound, $\tau_{\rm min}$, on tortuosity. While there is a trivial lower bound, $\tau_{\rm min} \geq  1$ for each setup, in \cref{app:B} we show that bounds can be tightened to approximately $1.128$ for the isolated setup and $1.304$ for the periodic one. The bound is larger for the periodic case because, for a fixed set of vertices, the additional paths obtained from the periodic setup are on average longer, as they can penetrate through the boundaries; see \cref{app:B}. Our simulations cannot access the breakdown in the power-law for the isolated case because, for the large values of $N_{\rm total}$ required to access this regime, we increasingly often violate the volume constraint for individual realizations of the random network.

\begin{figure}[!ht]
    \centering
    \subfloat[]{\label{fig:tt_tau_np}\includegraphics[width=.47\textwidth,height=.4\textwidth]{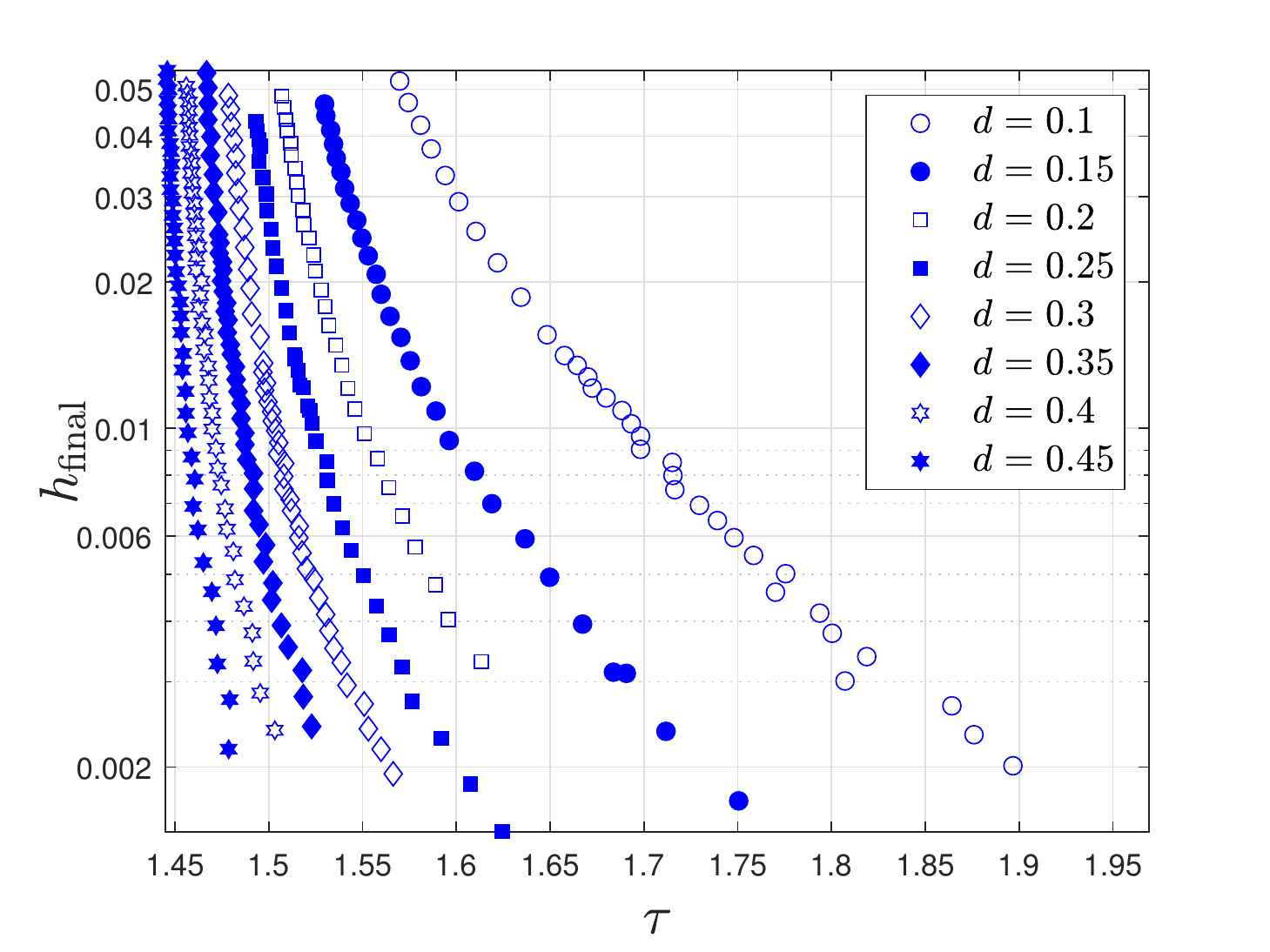}}
    \subfloat[]{\label{fig:tt_tau_p}\includegraphics[width=.47\textwidth,height=.4\textwidth]{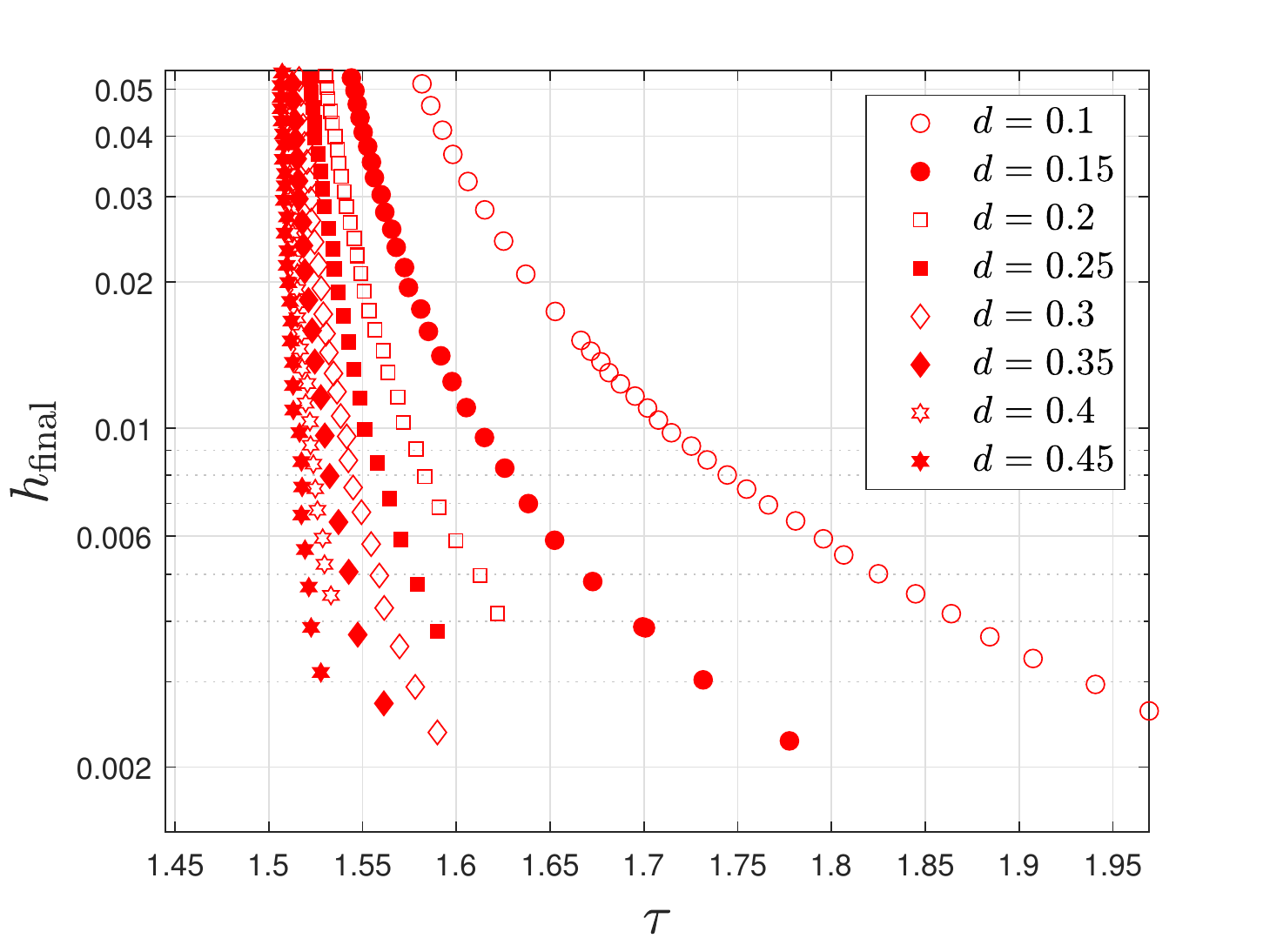}}
    \caption{Total throughput $h_{\rm final}$ vs tortuosity $\tau$ (semilog plot). (a) {\color{blue}Isolated setup}; (b) {\color{red}periodic setup}. Error distribution for each data point is provided in the supplement. Same scale for $h_{\rm final}$ as in \cref{fig:tt_vol}.}
    \label{fig:tt_tau}
\end{figure}

\cref{fig:tt_tau} shows total throughput $h_{\rm final}$ vs tortuosity $\tau$ for various $d$ values, on a semilog scale. We note that $h_{\rm final}$ decreases as $\tau$ increases for each $d$ because, if feed solution traverses longer paths (on average) through the filter, then it deposits more foulant and thus pores close faster. Second, for fixed $h_{\rm final}$, larger $d$ corresponds to smaller $\tau$. Combining this finding with the results in \cref{fig:conc_tau}, we reach the following conclusion: between two membrane networks that produce the same total throughput, the one with shorter characteristic pore length (smaller $d$) is favoured since it also has larger tortuosity and thus better foulant control. Lastly, in both \cref{fig:tt_tau_np} and \cref{fig:tt_tau_p}, we see again the clear evidence of a limiting tortuosity $\tau_{\rm min}>1$: in the data curves for each $d \geq 0.35$, $\tau$ does not vary greatly as we vary $N_{\rm total}$.

Now, we highlight the differences in \cref{fig:tt_tau}: first, for a fixed tortuosity $\tau$ (at a given $d$ value), $h_{\rm final}$ in the isolated setup (\cref{fig:tt_tau_np}) is lower than in the periodic one (\cref{fig:tt_tau_p}), because the periodic boundary conditions give rise to more edges and hence a larger initial void volume ${\rm Vol}_0$. We confirm this reasoning by plotting ${\rm Vol}_0$ against $\tau$ in \cref{fig:vol_tau}. There, when we fix $\tau$ for each value of $d$, we observe a higher initial volume for the periodic case than the isolated one. This observation also explains why accumulated foulant concentration is higher in the periodic case for a fixed $\tau$ (per \cref{fig:conc_tau}) because networks with periodic boundary conditions, by having larger initial volume, process more filtrate (larger throughput) which leads to a greater quantity of foulant escaping the filter. This reasoning also explains the difference between $c_{\rm acm}\left(t_{\rm final}\right)$ in \cref{fig:conc_vol_np} and \cref{fig:conc_vol_p} by fixing ${\rm Vol}_0$ for each $d$.
\begin{figure}[!ht]
    \centering
    \subfloat[]{\label{fig:vol_tau_np}\includegraphics[width=.47\textwidth,height=.4\textwidth]{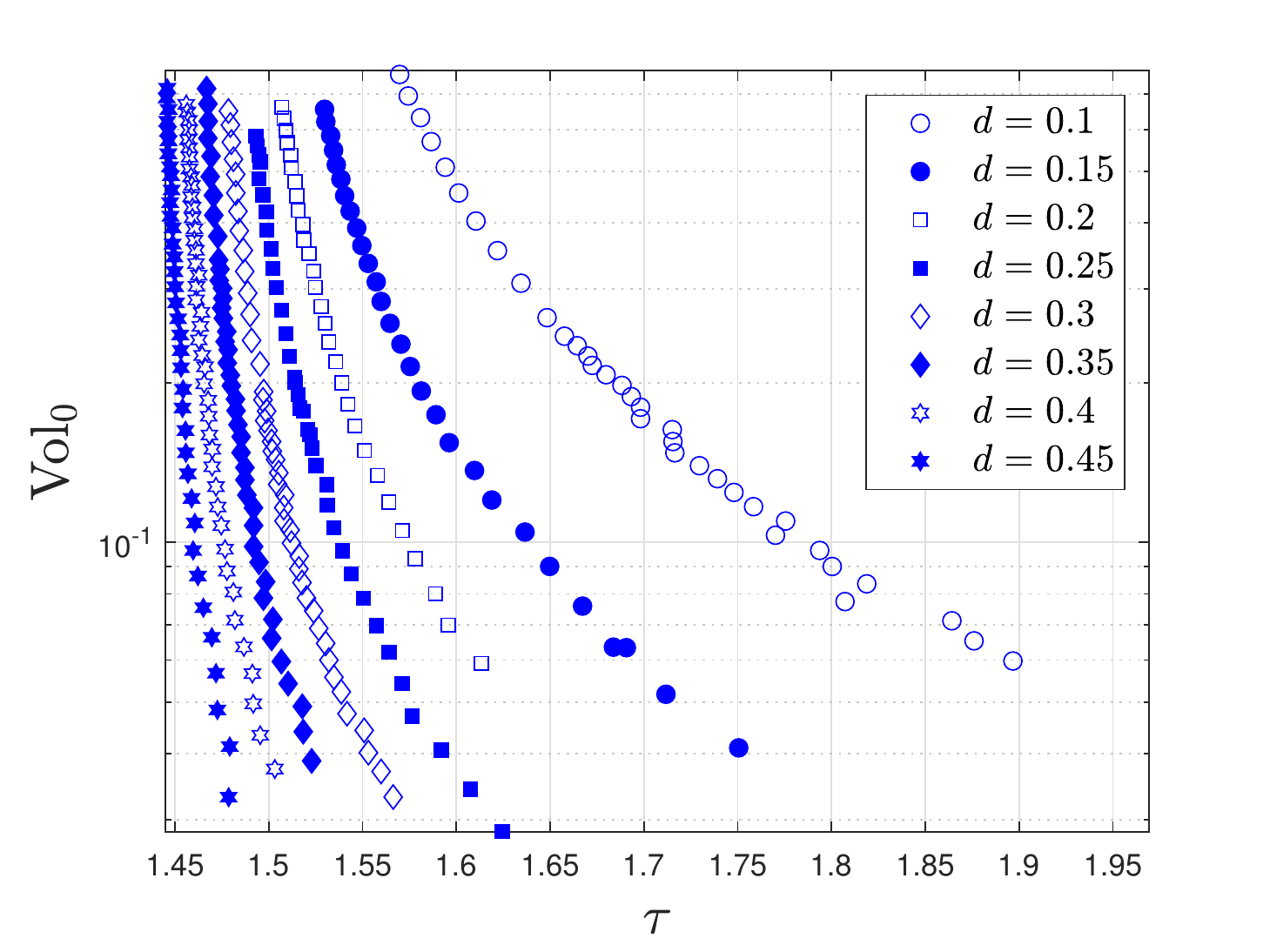}}
    \subfloat[]{\label{fig:vol_tau_p}\includegraphics[width=.47\textwidth,height=.4\textwidth]{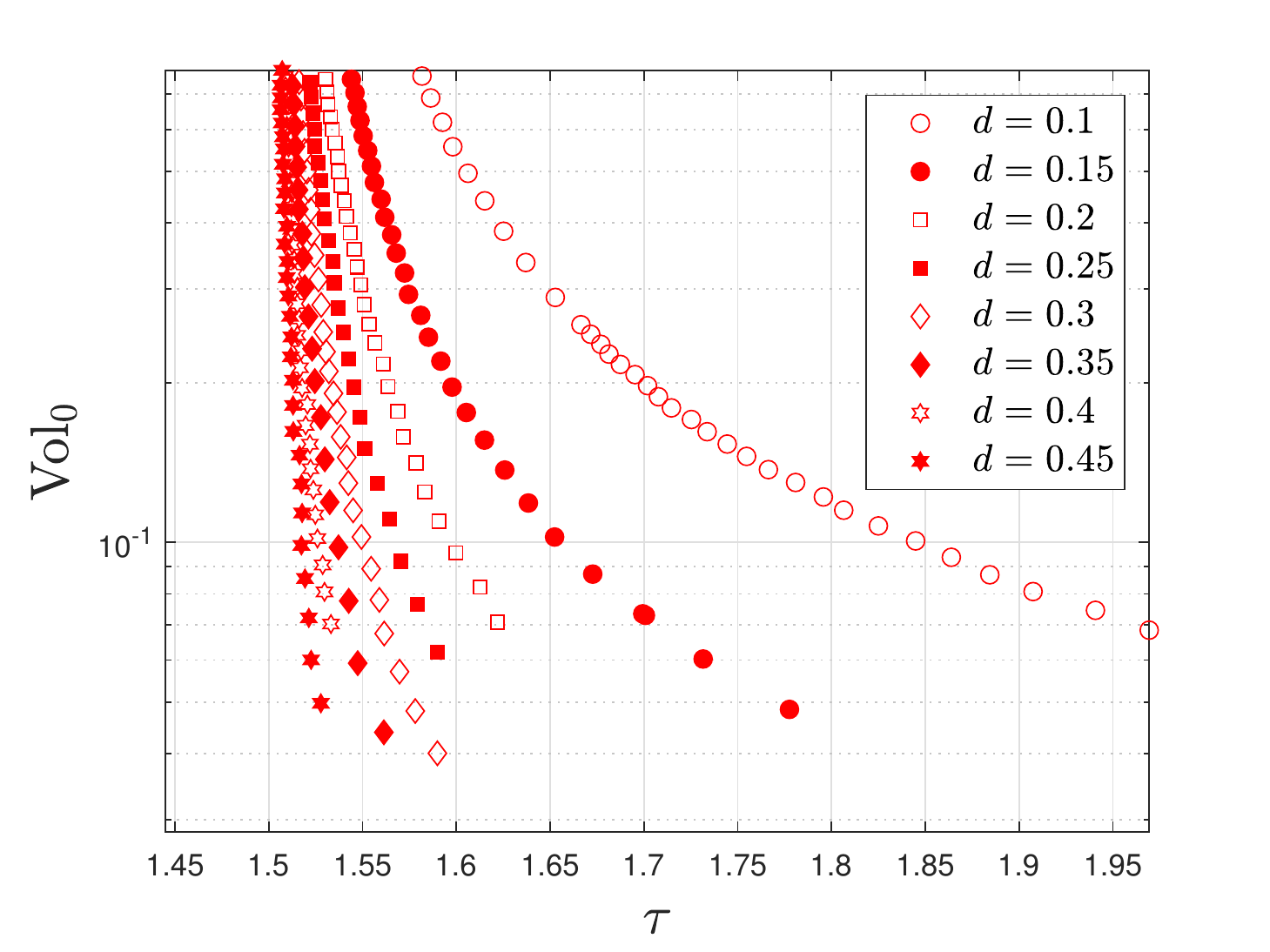}}
    \caption{Initial void volume vs tortuosity (semilog plot). (a) {\color{blue}Isolated setup}; (b) {\color{red}periodic setup}. Error distribution for each data point is provided in the supplement.}
    \label{fig:vol_tau}
\end{figure}

\section{Conclusions}
\label{sec:6}
We here summarize our findings on the connections between performance metrics of the network/filter and each characteristic geometric network parameter, and highlight areas for future work. We first collect our main findings:
\begin{enumerate}
    \item Initial void volume ${\rm Vol}_0$ is a good predictor of total throughput $h_{\rm final}$, particularly when ${\rm Vol}_0> 0.5$, when there appears to be a power-law relation between the two quantities (per \cref{fig:tt_vol}). Average number of neighbors $N$ is also correlated with $h_{\rm final}$, but is a weaker predictor in this respect than Vol$_0$ (per \cref{fig:tt_nbar}).
    \item Tortuosity $\tau$ fully characterizes final accumulated foulant concentration \newline $c_{\rm acm}\left(t_{\rm final}\right)$ in the filtrate by a negative exponential relationship. See \cref{fig:conc_tau}.
    \item When a minimum concentration requirement is imposed, membrane filters with small characteristic pore length should be considered since they have larger void volume and thus process more filtrate. See \cref{fig:conc_vol}.
    \item When two membrane networks produce the same final throughput $h_{\rm final}$, the one with shorter characteristic pore length should be favored, since it will have higher tortuosity and thus better foulant control. See \cref{fig:vol_tau}. 
\end{enumerate}
Griffiths {\it et al.} \cite{griffith_jms} showed that the particle removal efficiency, corresponding equivalently to $1-c_{\rm acm}\left(t_{\rm final}\right)$ in our setup, increases as tortuosity increases, a trend consistent with our conclusion here regarding the negative exponential relationship between accumulated foulant concentration and tortuosity, though the exact relationship in both approaches is different (negative sub-exponential in that work \cite{griffith_jms}). Such differences may originate from the different graph generation protocols used in the two works. 

We also find that the choice of network connection metric (isolated versus periodic) may significantly affect how the network functions as a filter. With identical input parameters $N_{\rm total}$ (representing the number of initial random points generated) and $d$ (the search radius, or maximum edge length), the periodic setup yields larger average number of neighbors and thus higher connectivity. This is important, since the periodic setup provides a more 
realistic representation of a (large) membrane. 

In recent years, many modern imaging techniques to study material properties have been developed which, in turn, provide fertile ground for graphical modeling approaches to membrane filtration and analysis of membrane performance. For example, Martinez {\it et al.} \cite{martinez} have outlined a topological algorithm (Figure 1 in their paper) that translates a
two-dimensional image of a membrane cross-section into a pore network. The idea is to take an original binary image where black and white represent membrane void and materials (respectively), and then construct a distance-to-nearest-object function (the DNO function) that computes the distance between a pixel to its nearest object (membrane material). Then, the location of pore junctions is determined by the local maxima of the DNO function, while the width of the edges connecting the junctions have widths determined by the minimal distance normal to the edge. For other techniques, see Sun {\it et al.} \cite{image_2007} and Sundaramoorthi {\it et al.}~\cite{3d_image}.

In future work, as well as utilizing some of the image analysis methods outlined above to generate more realistic networks, we plan to consider network models of membrane filters that include pore size variations, providing a more accurate representation of real membrane filters, due to inevitable manufacturing defects/inhomogeneities. We plan also to incorporate other fouling modes, such as blocking or sieving by large particles that may occlude pores, unaccounted for in this work. Including such large-particle blocking will add new dimensions of stochastic complexity to our modeling. In particular, a model that includes both extensions listed above may exhibit internal blockage at an earlier stage than the closure of top inlets (as always observed in our work here) when pore-size variations are sufficiently large, leading to much more complex flow and fouling behaviour. Our graphical description of the problem will allow us to investigate problems of this type in a systematic and computationally-efficient manner. 

\appendix
\section{Justification for pore radius evolution model}
\label{app:A}

In this section, we justify the form of the pore radius evolution equation \cref{adsorption} using the (exact) solution for the foulant concentration model, \cref{adv_single}, to relate the rate of change of particle volume accretion inside a single pore to pore radius shrinkage. In the following derivation, we drop the indices $ij$ for notational simplicity and assume that the radius $R$ is spatially constant at each time $T$ (though the arguments can be adapted to the variable radius case with a suitable bound on $\partial R/\partial Y$).

Let $V_p$ be the particle volume. In a single pore (assumed circularly cylindrical), we consider an infinitesimally thin circular disk at distance $Y$ from the pore inlet, with thickness $dY$. The particle flux difference across this disk is 
\begin{equation}
    \left(C\left(Y+dY,T\right)-C\left(Y,T\right)\right)\mathbf{Q}\left(T\right)\approx\frac{\partial C}{\partial Y}\left(Y,T\right)\mathbf{Q}\left(T\right)dY,
    \label{eq:adsorped}
\end{equation}
i.e., the number of particles per time deposited in the disk. The total particle volume adsorped in this pore is $V_p$ times this quantity. We integrate \cref{eq:adsorped} over the length of the pore to obtain the total volume of deposited particles per time in the pore,
\begin{subequations}
    \begin{align}
    V_{p}\mathbf{Q}\left(T\right)\int_{0}^{A}\frac{\partial C}{\partial Y}\left(Y,T\right)dY
     & =V_{p}\mathbf{Q}\left(T\right)\left[C\left(A,T\right)-C\left(0,T\right)\right]\\
     & =V_{p}\mathbf{Q}\left(T\right)\boldsymbol{C}\left(T\right)\left(\exp\left(-\frac{\Lambda R\left(T\right)A}{\mathbf{Q}\left(T\right)}\right)-1\right) \\
     & \approx -V_{p}\boldsymbol{C}\left(T\right)\Lambda R\left(T\right)A
     \label{eq:part_accret}
    \end{align}
\end{subequations}
where we used the analytical expression \cref{adv_single} for $C\left(Y,T\right)$ in the second equality. The final approximate Taylor expansion \cref{eq:part_accret} is justified for sufficiently small values of the exponent, that is (using the scales in \cref{sec:3}), provided
\begin{equation}
    \mathbf{q}\left(t\right)\gg\lambda r_{0}d,
\end{equation}
where $d$ is the largest pore length (see \cref{edge}).

As foulant particle volume accumulates at the rate given in \cref{eq:part_accret}, pore volume ${\rm Vol}\left(T\right)$ also changes at this rate via
\begin{equation}
    \frac{d{\rm Vol}\left(T\right)}{dT} =-V_{p}\boldsymbol{C}\left(T\right)\Lambda R\left(T\right)A.
    \label{eq:particle_volume_rate}
\end{equation}
At the same time, we relate the volume of the pore to its radius ${\rm Vol}\left(T\right) = \pi R^2\left(T\right) A$ and obtain
\begin{equation}
    \frac{d{\rm Vol}\left(T\right)}{dT}=2\pi R\left(T\right)\frac{dR}{dT}A.   
    \label{eq:pore_volume_rate}
\end{equation}
Equating \cref{eq:particle_volume_rate} and \cref{eq:pore_volume_rate}, we arrive at the form of \cref{adsorption} with $\alpha = \frac{V_p}{2\pi}$.

\section{Tortuosity} \label{app:B}

In this section, we define tortuosity of a graph representing a membrane filter pore network, and provide an explicit formula using the geometric and initial flow information from the network (found in \cref{sec:2}). In all of our investigations we restrict attention to the tortuosity of the initial pore network, with no regard for how it subsequently evolves under fouling, hence in the following discussion it should be understood that we consider properties of the network at time $t=0$.

We define tortuosity as the average distance travelled by a fluid particle through the membrane via the network, relative to the thickness of membrane $W$. Now, given a path from any inlet to any outlet, we can associate it with its total initial flux, which we use as a weight for the path. This is equivalent to having the fluid particle perform a discrete random walk on the graph $G$ directed by fluid flux at each junction. More precisely, the transition matrix $\mathbf{P}$ of this random walk is defined as follows. We omit the argument of $t=0$ for notational simplicity.

We first enforce non-negativity of the flux matrix via the following modification. Consider $\mathbf{Q}_{+}$ and $\mathbf{Q}_{-}$, the positive and negative parts of $\mathbf{Q}$ respectively, such that $\mathbf{Q} = \mathbf{Q}_{+}+\mathbf{Q}_{-}$. Owing to the skew-symmetry of $\mathbf{Q}$, we construct 
\begin{equation}
    \mathbf{\tilde{Q}} = \mathbf{Q}_{+}-\mathbf{Q}_{-}^{T},
    \label{eq:mod_Q}
\end{equation}
which preserves the flow information (direction and magnitude) while enforcing non-negativity.
\begin{definition}(Transition Matrix)
Given modified flux matrix $\mathbf{\tilde{Q}}$ in \cref{eq:mod_Q}, the transition matrix $\mathbf{P}$ is determined by rescaling $\mathbf{\tilde{Q}}$ by its row sum: 
\begin{align}
\mathbf{P}_{ij}=\begin{cases}
\frac{{\displaystyle \mathbf{\tilde{Q}}_{ij}}}{{\displaystyle \sum_{v_{j}:\left(v_{i},v_{j}\right)\in E}\mathbf{\tilde{Q}}_{ij}}}, & \left(v_{i},v_{j}\right)\in E,\quad v_{i}\in V_{{\rm top}}\cup V_{{\rm int}},\\
1, & v_{i}\in V_{{\rm bot}},\quad j=i,\\
0, & \text{otherwise.}
\end{cases}
\end{align}
For vertices in the bottom surface, fluid particles are absorbed, i.e. once they reach any $v\in V_{\rm bot}$, they stay there with probability one. 
\end{definition}
Let $X_n$ be a random walk with transition $\mathbf{P}$, i.e. $X_n$ is the vertex after the fluid particle has taken $n$ steps on $V$. Let $\mathbb{P}$ be the induced probability measure. This random walk has a natural initial distribution (a column vector of length $\left|V\right|$), 
\begin{equation}
    \pi_{0i}=\mathbb{P}\left\{X_0 = v_i\right\}:=
    \begin{cases}
\frac{{\displaystyle \sum_{v_{j}:\left(v_{i},v_{j}\right)\in E}\mathbf{\tilde{Q}}_{ij}}}{{\displaystyle \sum_{v_{i}\in V_{{\rm top}}}\sum_{v_{j}:\left(v_{i},v_{j}\right)\in E}\mathbf{\tilde{Q}}_{ij}}}, & v_{i}\in V_{{\rm top}},\\
0, & \text{otherwise,}
    \end{cases}
    \label{eq:initial_dist}
\end{equation}
namely, the probability of going to each inlet on the top surface is determined by the proportion of flux going through that inlet, relative to total flux. 

\begin{definition}(Tortuosity)\label{tor_def}
Let $l_n$ be the total path length after the random walk has taken $n$ steps. Tortuosity of the graph $G$ is defined by $\mathbb{E}\left[\tau\right]$ where $\tau := \frac{l_m}{W}$
and $m$ is the (deterministic) number of steps of the longest path from any inlet on the top surface to any outlet in the bottom surface.
\end{definition}
The integer $m$ can be understood as a graph diameter where the notion of graph distance for $m$ is encoded in the unweighted adjacency matrix $\mathbf{W}$ (entries are indicators of the existence of an edge). $m$ is trivially bounded above by $\left|V_{\rm int}\right|+2$ (one step from $V_{\rm top}$ to $V_{\rm int}$ and $V_{\rm int}$ to $V_{\rm bot}$ respectively, and traverse all of $V_{\rm int}$ at worst), which we use here. Although this bound can be tightened by connectivity measures such as the smallest number of vertices that must be removed to disconnect a graph (see Coll {\it et al.}~\cite{Coll2020ThePD}, for example), our algorithm (discussed after proving the formula \cref{tort_formula}) for $\mathbb{E}\left[\tau\right]$ does not incur significant computational cost from the size of $m$.

One may estimate this expected value by sending a large number of particles through the network and computing the average of path lengths. We here provide an explicit formula for $\mathbb{E}\left[\tau \right]$ that depends on the transition matrix $\mathbf{P}$ and a distance weight matrix $\mathbf{W}_E$, whose entries are
\[
\mathbf{W}_{E,ij}=\begin{cases}
\chi\left(v_{i},v_{j}\right), & \left(v_{i},v_{j}\right)\in E,\\
0, & \text{otherwise,}
\end{cases}
\]
where $\chi\left(v_{i},v_{j}\right)$ is the distance between vertices $v_i$ and $v_j$ via the metric $\chi$ defined in \cref{metric}. Using this formula directly obviates the use of large-number-of-particle simulations and thus reduces computational load significantly.

\begin{proposition}(Tortuosity formula)
\begin{equation}
    \mathbb{E}\left[\tau\right]=\frac{1}{W}\mathbb{E}\left[l_m\right]=\frac{\pi_{0}^{{\rm T}}}{W}\left(\sum_{n=1}^{m}\mathbf{P}^{n-1}\right){\rm diag}\left(\mathbf{P}\mathbf{W}_{E}\right),
    \label{tort_formula}
\end{equation}
where ${\rm diag}\left(A\right)$ is a column vector that lists the diagonal elements of a matrix $A$.
\end{proposition}
\begin{proof}
We compute $\mathbb{E}\left[l_m\right]$. Denote the conditional probability and expectation
\[
\mathbb{P}\left[\cdot\mid X_{0}=i\right]=\mathbb{P}_{i}\left[\cdot\right],\quad\mathbb{E}\left[\cdot\mid X_{0}=i\right]=\mathbb{E}_{i}\left[\cdot\right].
\]
First, we observe by law of total expectation that
\[
\mathbb{E}\left[l_m\right]=\sum_{i=1}^{\left|V\right|}\mathbb{E}\left[l_m\mid X_{0}=i\right]\mathbb{P}\left\{ X_{0}=i\right\} =\sum_{i=1}^{\left|V\right|}\mathbb{E}_{i}\left[l_m\right]\pi_{0i}:=\pi_0^{\rm T} \boldsymbol{U}, 
\]
where $\pi_0$ is given by \cref{eq:initial_dist} and $\boldsymbol{U} := \left(\mathbb{E}_{i}\left[l_m\right]\right)_{v_i\in V}$. We now focus on an arbitrary element of $\boldsymbol{U}$. Noting that $\mathbb{E}_{i}\left[L_{0}\right]=0$, and using linearity of expectation, law of total expectation, the Markov property of the random walk and symmetry of $\mathbf{W}_{E}$, we have
\begin{gather*}
\mathbb{E}_{i}\left[l_{m}\right]=\mathbb{E}_{i}\left[\sum_{n=1}^{m}\left(l_{n}-l_{n-1}\right)\right] =\sum_{n=1}^{m}\mathbb{E}_{i}\left(l_{n}-l_{n-1}\right)\\
 =\sum_{n=1}^{m}\sum_{j\in V}\mathbb{E}_{i}\left(l_{n}-l_{n-1}\mid X_{n-1}=j\right)\mathbb{P}_{i}\left(X_{n-1}=j\right)\\
 =\sum_{n=1}^{m}\sum_{j\in V}\left(\sum_{k\in V}\mathbf{W}_{E,jk}\mathbf{P}_{jk}\right)\mathbf{P}_{ij}^{\left(n-1\right)}
 =\sum_{n=1}^{m}\sum_{j\in V}\mathbf{P}_{ij}^{\left(n-1\right)}\left(\mathbf{P}\mathbf{W}_{E}\right)_{jj}
\end{gather*}
where $\mathbf{P}_{ij}^{\left(k\right)}$ is the $k$-th iterate of $\mathbf{P}$. Thus, in matrix form, 
\begin{equation*}
\boldsymbol{U} =\sum_{n=1}^{m}\sum_{j\in V}\left(\begin{array}{c}
\mathbf{P}_{1j}^{\left(n-1\right)}\\
.\\
.\\
.\\
\mathbf{P}_{\left|V\right|j}^{\left(n-1\right)}
\end{array}\right)\left(\mathbf{P}\mathbf{W}_{E}\right)_{jj}
 =\left(\sum_{n=1}^{m}\mathbf{P}^{n-1}\right)\text{diag}\left(\mathbf{P}\mathbf{W}_{E}\right),
\end{equation*}
using the fact that $\mathbf{P}^{\left(n-1\right)} =\mathbf{P}^{n-1}$, completing the proof.
\end{proof}
In practice, to avoid taking large matrix powers when evaluating \cref{tort_formula}, we have devised a fast algorithm in evaluating the geometric series $\sum_{n=1}^{m}\mathbf{P}^{n-1}$, by appealing to a geometric sum formula on matrices involving matrix inversions. We utilize the block upper triangular structure of $\mathbf{P}$ by block partitioning into components (known as a {\it divide and conquer}-type scheme) including the identity block in its southeast corner, to ease the computational load of inversions. Naive evaluation of the series is of complexity $O\left(m\left|V\right|^3\right)$ (cubic term from matrix multiplication and $m$ additions) while our algorithm is $O\left(\left|V\right|^3\right)$. 

We argue that the constant initial radius assumption on the pores deems $\tau$ a geometric parameter independent of fluid flow, even though its definition requires initial flow information and geometric information such as a distance weighted adjacency (see \cref{tort_formula}). In essence, we claim that $\tau$ does not vary too much until the filtration stopping criterion. Firstly, foulant concentration is monotonically decreasing along each edge, and thus the radii of all inlets, under the influence of foulant concentration in the feed solution (see \cref{dimless_conc_bc}), go to zero earlier than all other downstream channels. Thus, adjacencies of the network do not change until the filter top surface is clogged. Secondly, though outflowing flux from an arbitrary junction changes over time as the filter fouls, the relative contribution from each outgoing edge does not vary greatly. Altogether, we believe that tortuosity does not depend heavily on the time of filtration but only on the initial geometry of the network. This feature makes tortuosity a universal parameter for foulant concentration.

Lastly, we note that a theoretical limit exists for tortuosity for both setups. As $N_{\rm total}\rightarrow \infty$ (so that the membrane interior, top and bottom surface are uniformly densely packed with pore junctions), we can provide a simple lower bound for the limit in the following sense. Consider an arbitrary inlet-outlet pair. They will be connected by a path with vertical component $1$. For the isolated setup, the horizontal component can be estimated as the average distance between two uniformly random points in 2D in a unit square, and is about $0.521$ (length of {\color{blue}$YK$} in \cref{fig:np_tort}). Together, these numbers provide a lower bound for the limiting tortuosity $\tau_{\rm min}$ of around $1.128 = \sqrt{1^2+0.521^2}$ for the isolated case. The argument for the periodic case is slightly more elaborate -- while the vertical component of an average path is still $1$, the horizontal component now is the average distance between two random points uniformly sampled from the squares $\left[0.5,1.5\right]^2$ and $\left[0,2\right]^2$ respectively (found to be $0.838$ (length of {\color{blue}$YK$} in \cref{fig:p_tort})); thus an average path length is about $1.304$. The difference arises because with the periodic metric, any inlet uniformly sampled from a unit square can potentially connect to outlets that are outside the unit square (bottom surface) but within $0.5$ distance to the boundary (due to the constraint that search radius $d<0.5$). These arguments provide some justification for the observations that the tortuosity in the periodic case may be larger than the isolated case. The numerical values given above are obtained by standard probabilistic calculations and numerical integration. We look forward to a more theoretical study of how tortuosity is affected by initial number of points, search radius and the underlying metric. 
\begin{figure}
    \centering
    \subfloat[]{\label{fig:np_tort}\includegraphics[scale = 0.25]{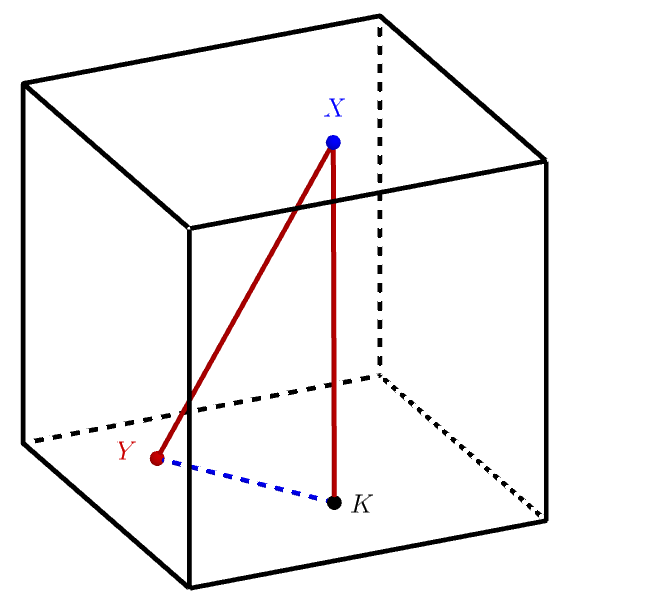}}
    \subfloat[]{\label{fig:p_tort}\includegraphics[scale = 0.25]{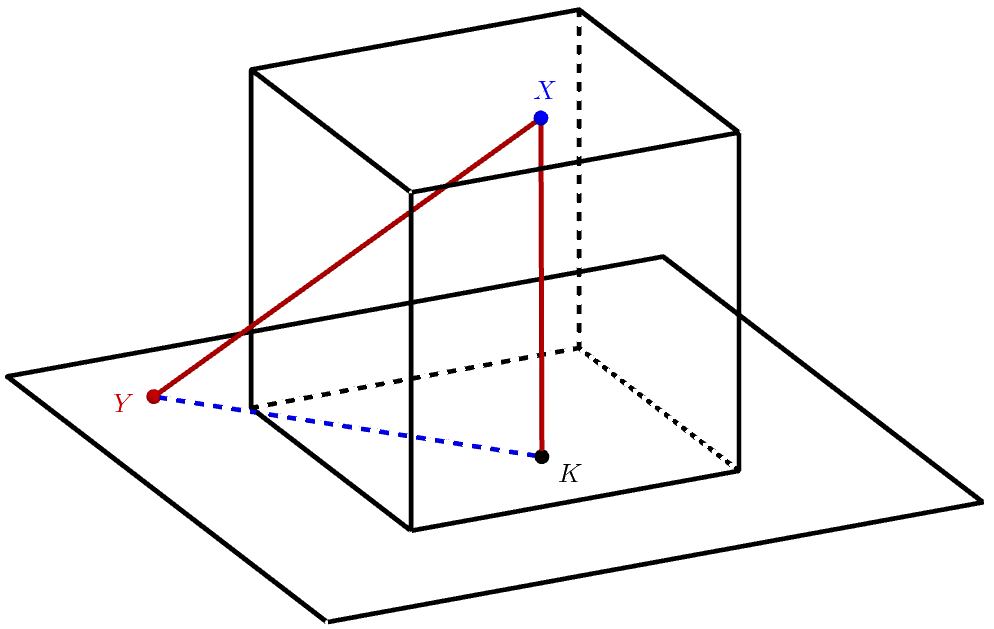}}
    \caption{Tortuosity limit: (a) Isolated setup (unit cube); (b) periodic setup (unit cube and a square of side length $2$ containing the bottom surface of the unit cube). {\color{blue}$X$ (blue)} and {\color{red}$Y$ (red)} are uniformly sampled from the top and bottom membrane surface respectively. $K$ (black) in both figures is the projection of $X$ onto the bottom surface; {\color{red}$XK$} in both figures has length $1$.}
    \label{fig:tort_schematic}
\end{figure}

\section{Worked Example} \label{app:C}
In this section, we provide an example of a simple network (see \cref{fig:inverted_Y}) to illustrate how the governing equations described in \S\S 2.3--2.6 depend on the model parameters. The network is a reflected Y-shape, with two inlets and outlets of length ${1/2}$ and one interior edge of length ${1/3}$. 

Each edge has conductance $\mathbf{k}_{ij}\left(t\right)$ per \cref{eq:scaling}. For this network, the conductance-weighted graph Laplacian, acting on the pressures at interior vertices, yields the graph Laplace equation \cref{dis_lp},
\begin{multline}
L_{\mathbf{k}}\boldsymbol{p}\left(t\right)=\left[\begin{array}{cc}
\mathbf{k}_{13}\left(t\right)+\mathbf{k}_{23}\left(t\right)+\mathbf{k}_{34}\left(t\right) & -\mathbf{k}_{34}\left(t\right)\\
-\mathbf{k}_{34}\left(t\right) & \mathbf{k}_{34}\left(t\right)+\mathbf{k}_{45}\left(t\right)+\mathbf{k}_{46}\left(t\right)
\end{array}\right]\left[\begin{array}{c}
p_{2}\left(t\right)\\
p_{3}\left(t\right)
\end{array}\right]\\
=\left[\begin{array}{c}
\mathbf{k}_{13}\left(t\right)+\mathbf{k}_{23}\left(t\right)\\
0
\end{array}\right]
\label{eq:lp_example}
\end{multline}
where the final equality incorporates the boundary conditions \cref{dimless_laplace_bc}. Then, the fluxes $\mathbf{q}_{ij}$ satisfy 
\begin{align*}
\mathbf{q}_{13}\left(t\right) & =\mathbf{k}_{13}\left(t\right)\left(1-p_{3}\left(t\right)\right),\\
\mathbf{q}_{23}\left(t\right) & =\mathbf{k}_{23}\left(t\right)\left(1-p_{3}\left(t\right)\right),\\
\mathbf{q}_{34}\left(t\right) & =\mathbf{k}_{34}\left(t\right)\left(p_{3}\left(t\right)-p_{4}\left(t\right)\right),\\
\mathbf{q}_{45}\left(t\right) & =\mathbf{k}_{45}\left(t\right)p_{4}\left(t\right),\\
\mathbf{q}_{46}\left(t\right) & =\mathbf{k}_{46}\left(t\right)p_{4}\left(t\right).
\end{align*}
To proceed, we then solve \cref{dimless_conc} (the advection graph Laplace equation) to find the foulant concentration at each vertex,
\begin{multline}
L_{\mathbf{q}}^{{\rm in}}\boldsymbol{c}=\left[\begin{array}{cccc}
\mathbf{q}_{13}+\mathbf{q}_{23} & 0 & 0 & 0\\
-\mathbf{q}_{34}\mathbf{b}_{34} & \mathbf{q}_{34} & 0 & 0\\
0 & -\mathbf{q}_{45}\mathbf{b}_{45} & \mathbf{q}_{45} & 0\\
0 & -\mathbf{q}_{46}\mathbf{b}_{46} &  & \mathbf{q}_{46}
\end{array}\right]\left[\begin{array}{c}
c_{3}\left(t\right)\\
c_{4}\left(t\right)\\
c_{5}\left(t\right)\\
c_{6}\left(t\right)
\end{array}\right]=\\
\left(\mathbf{q}\circ\mathbf{b}\right)^{{\rm T}}\boldsymbol{c}_{0}=\left[\begin{array}{cc}
\mathbf{q}_{13}\mathbf{b}_{13} & \mathbf{q}_{23}\mathbf{b}_{23}\\
0 & 0\\
0 & 0\\
0 & 0
\end{array}\right]\left[\begin{array}{c}
1\\
1
\end{array}\right]=\left[\begin{array}{c}
\begin{array}{c}
\mathbf{q}_{13}\mathbf{b}_{13}+\mathbf{q}_{23}\mathbf{b}_{23}\end{array}\\
0\\
0\\
0
\end{array}\right]
\label{eq:trans_example}
\end{multline}
where the 3rd equality uses the boundary condition \cref{dimless_conc_bc}.

After obtaining the concentrations $c_i\left(t\right)$ for each vertex, we evolve edge radius via \cref{dimless-adsorption}.

By appealing to the symmetry of this network, one may reduce the system for the dynamics of network evolution to a system of two non-autonomous nonlinear ordinary differential equations (ODEs), describing the radius evolution $r_{34}\left(t\right)$ and $r_{45}\left(t\right)$ (which by symmetry is equal to $r_{46}(t)$), with a closed solution for the linear shrinkage rate of $r_{13}(t)=r_{23}(t)$. Since our principal concern here is the end state of the network, we do not present or study this system here.

\begin{figure}
    \centering
    \includegraphics[scale = 0.26]{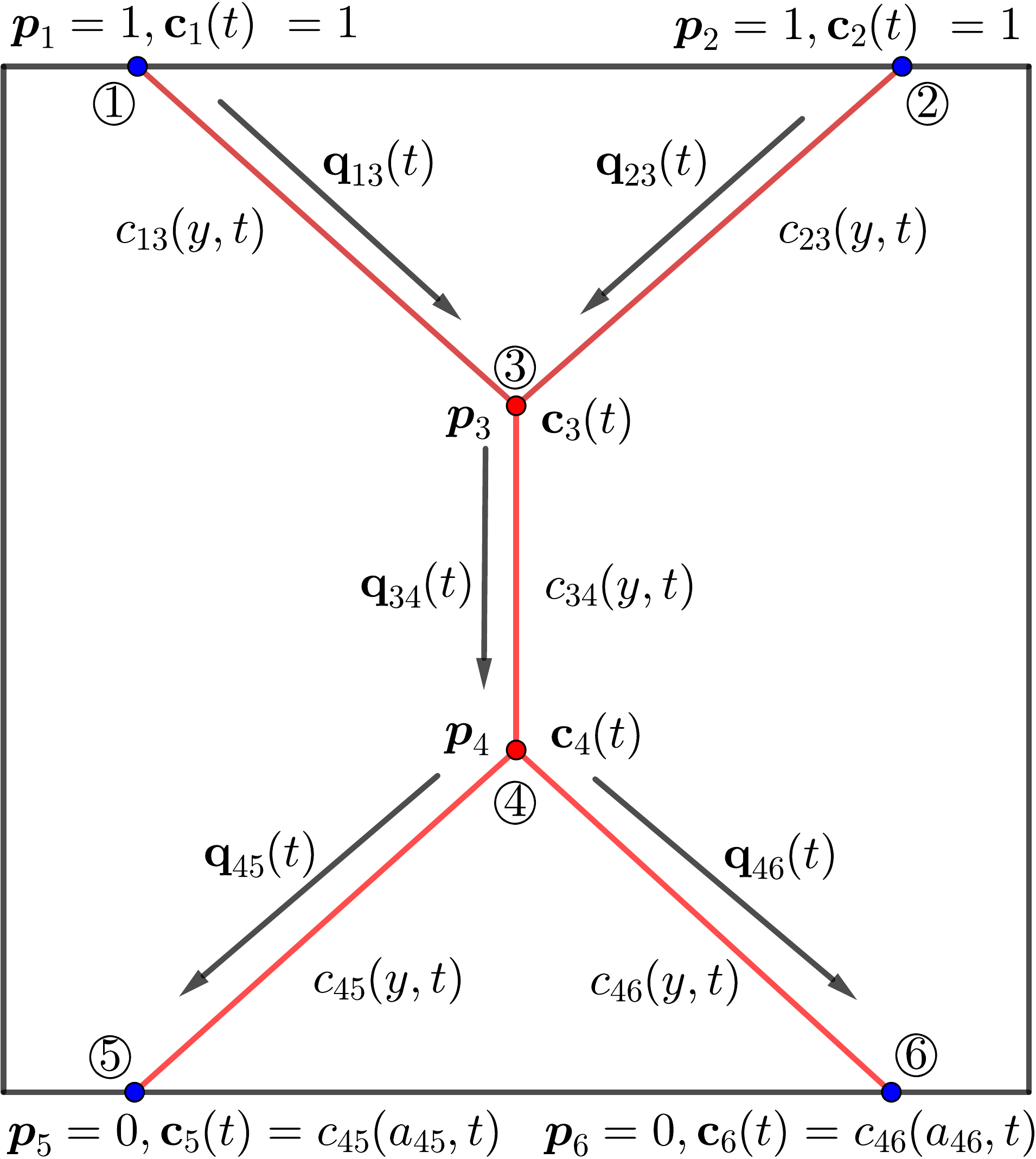}
    \caption{2D schematic of a reflected-Y network with labelled physical quantities presented in \S3.}
    \label{fig:inverted_Y}
\end{figure}

\section*{Acknowledgments}
This work was supported by NSF Grant No. DMS-1615719.

\bibliographystyle{siamplain}
\bibliography{references}
\end{document}


\maketitle

\section{Variability over random graphs}
Here we show the distribution of initial void volume Vol$_0$, tortuosity $\tau$, and total throughput $h_{\rm final}$, over $500$ realizations of our random graph setup, scaled by the mean of each quantity respectively, for both isolated and periodic setups. Distributions are shown in \cref{fig:hist_d_01} for $d=0.1$ and \cref{fig:hist_d_045} for $d=0.45$. We see that the distribution of the mean-scaled data of each quantity is approximately Gaussian, with most of the data points concentrated near the mean (especially well for initial tortuosity when $d=0.1$, see \cref{fig:hist_d_01_tau_np} and \cref{fig:hist_d_01_tau_p}). These histograms also help us determine if a parameter pair $\left(d,N_{\rm total}\right)$ will produce networks that violate the volume and area constraints.
\begin{figure}[!ht]
    \centering
    \subfloat[]{\label{fig:hist_d_01_vol_np}\includegraphics[width=.47\textwidth,height=.4\textwidth]{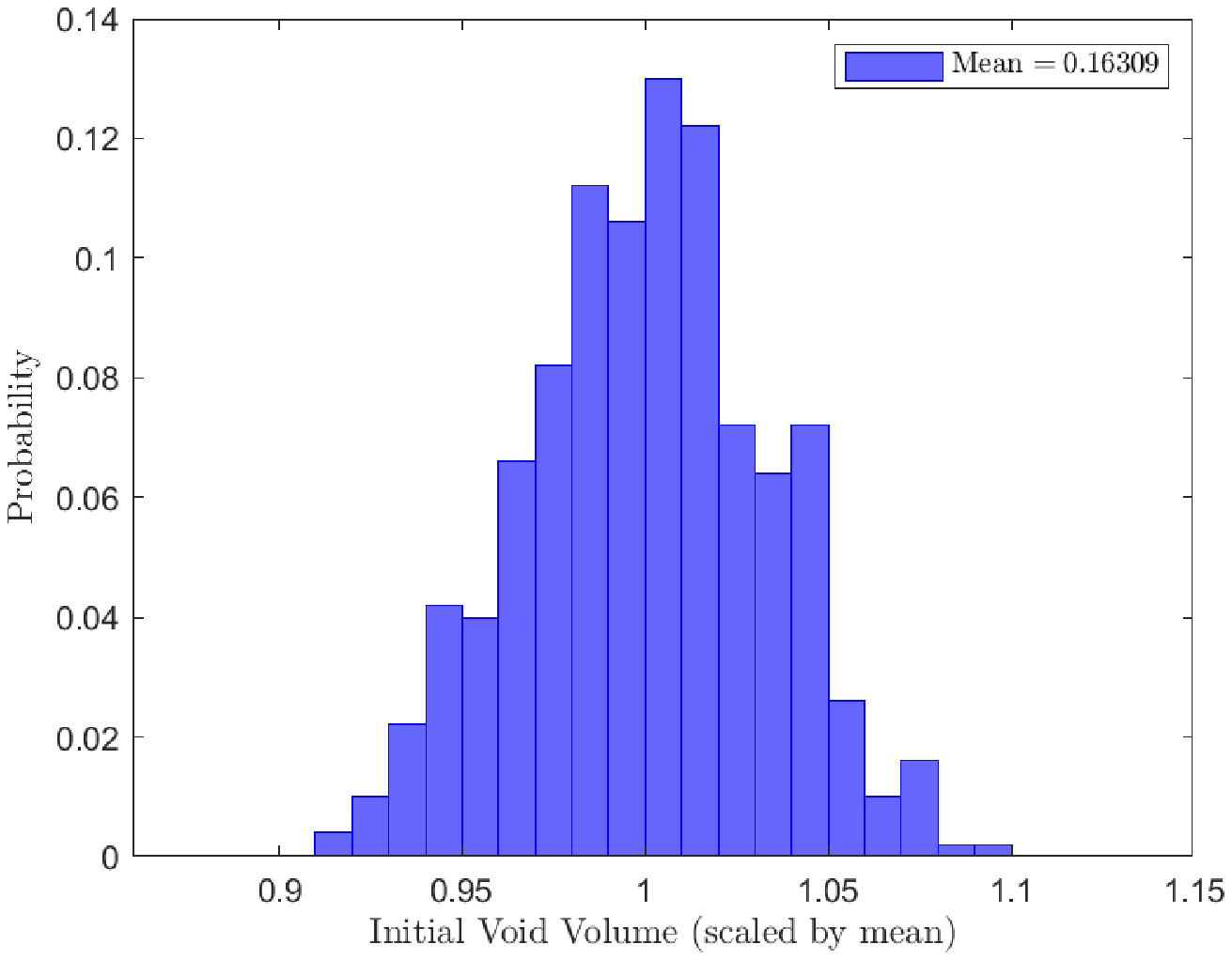}}
    \subfloat[]{\label{fig:hist_d_01_vol_p}\includegraphics[width=.47\textwidth,height=.4\textwidth]{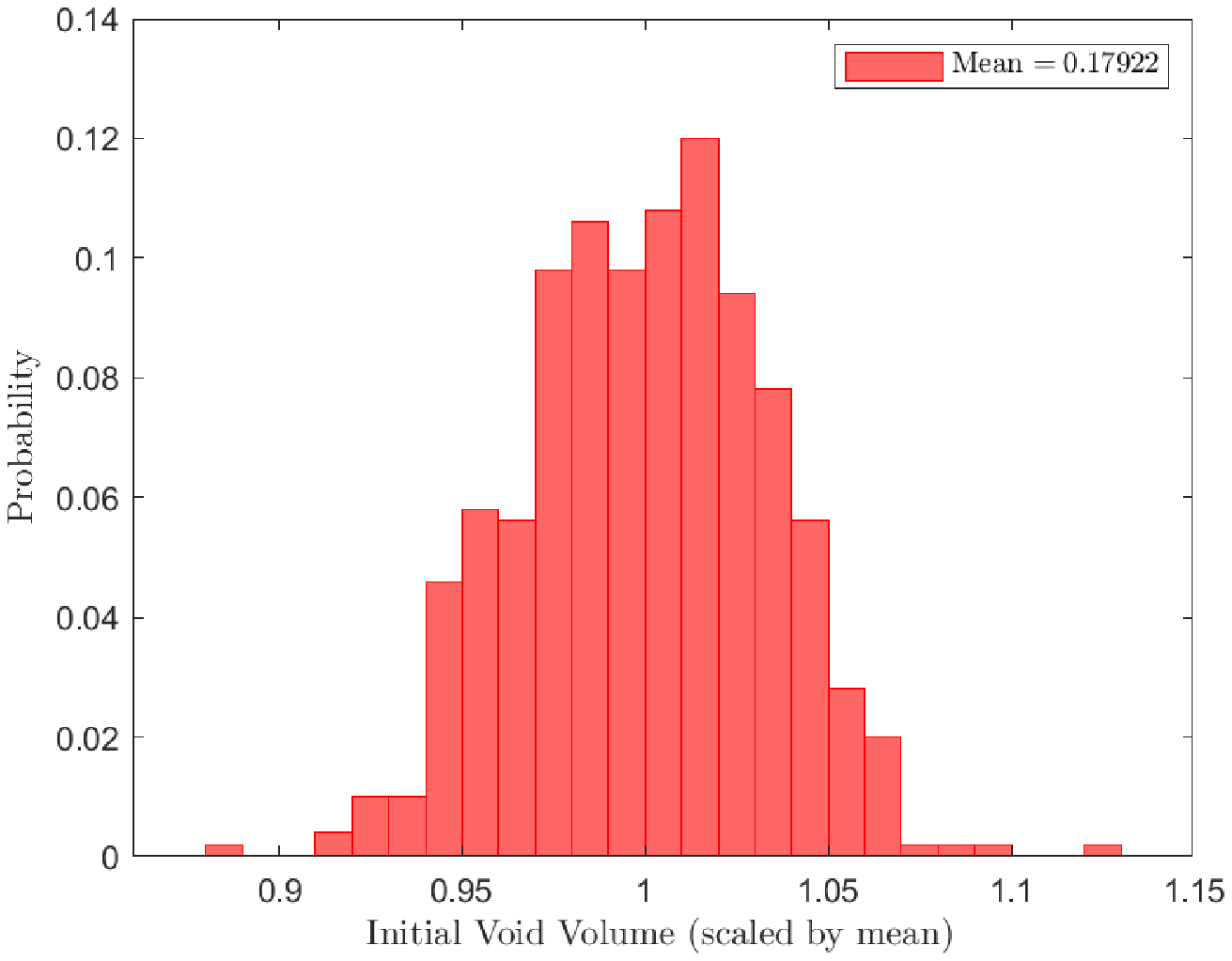}} \\
    \subfloat[]{\label{fig:hist_d_01_tau_np}\includegraphics[width=.47\textwidth,height=.4\textwidth]{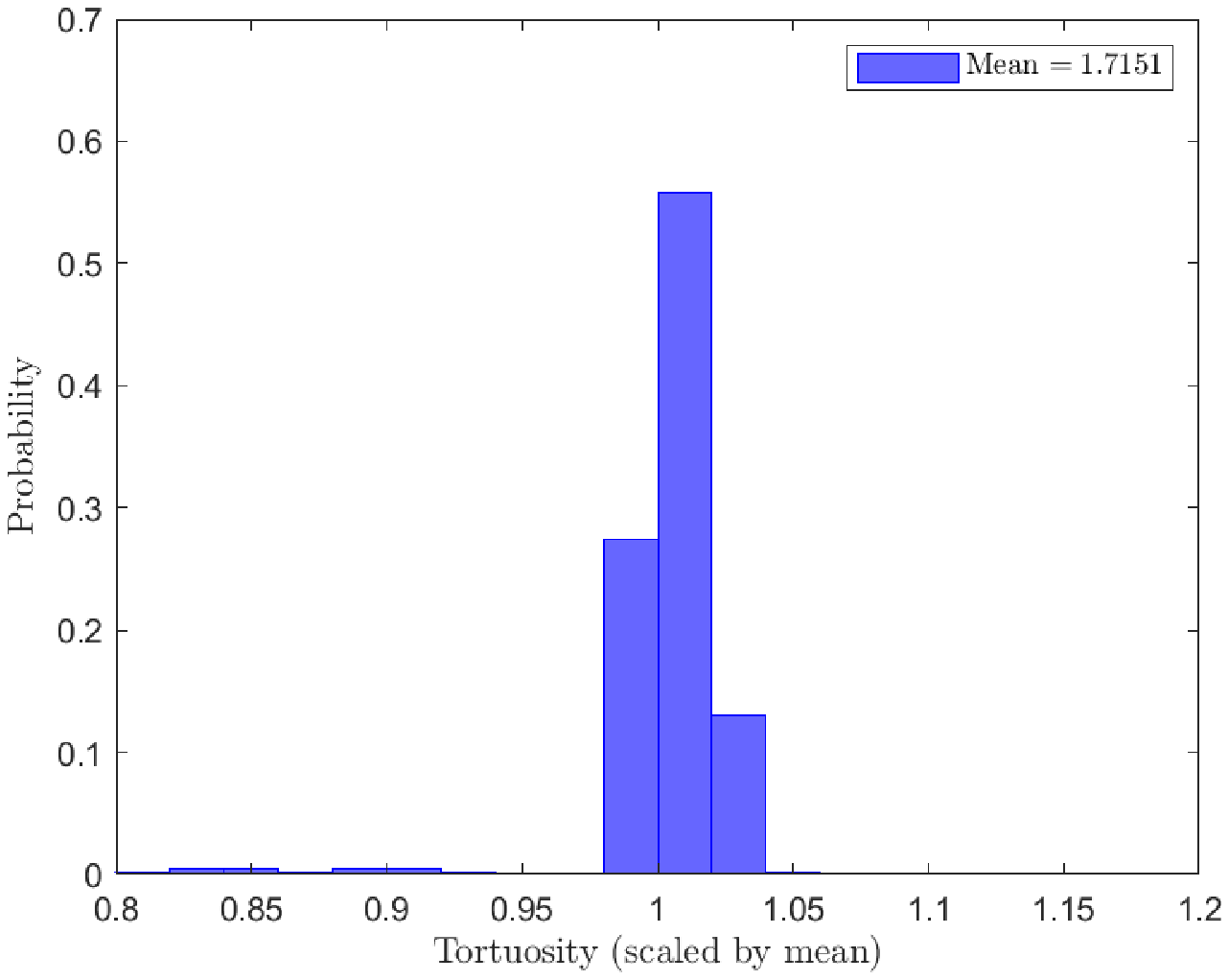}}
    \subfloat[]{\label{fig:hist_d_01_tau_p}\includegraphics[width=.47\textwidth,height=.4\textwidth]{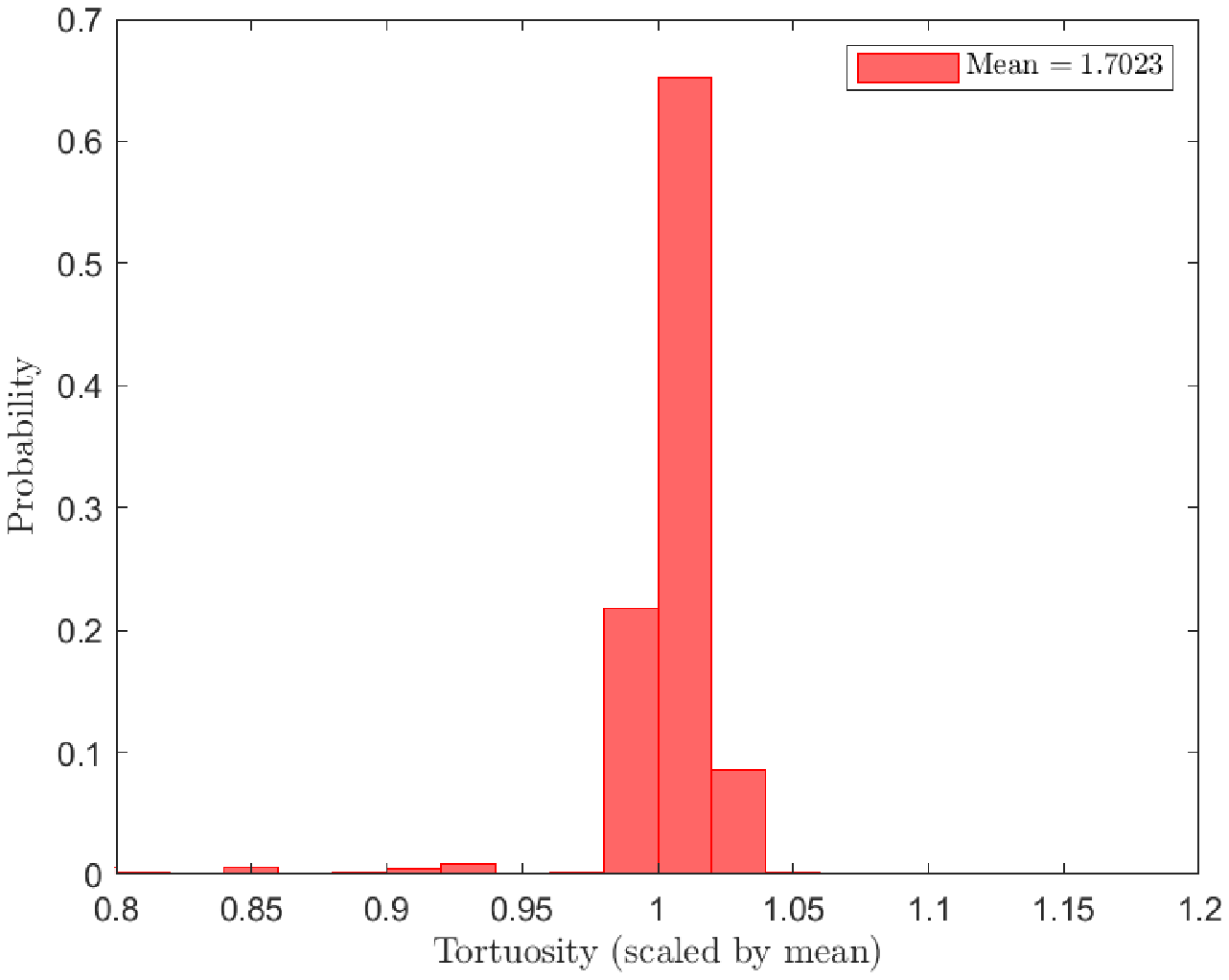}} \\
    \subfloat[]{\label{fig:hist_d_01_tt_np}\includegraphics[width=.47\textwidth,height=.4\textwidth]{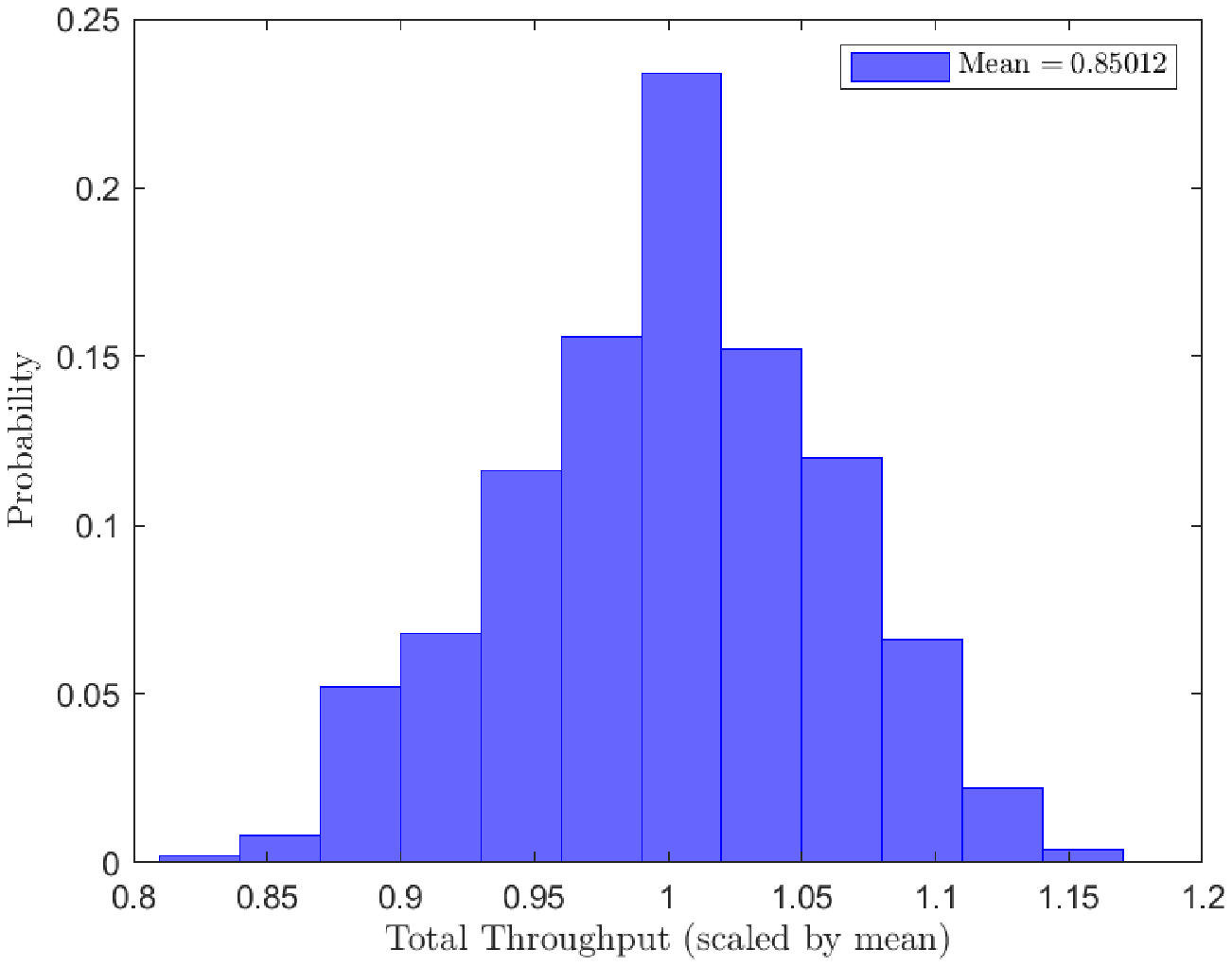}}
    \subfloat[]{\label{fig:hist_d_01_tt_p}\includegraphics[width=.47\textwidth,height=.4\textwidth]{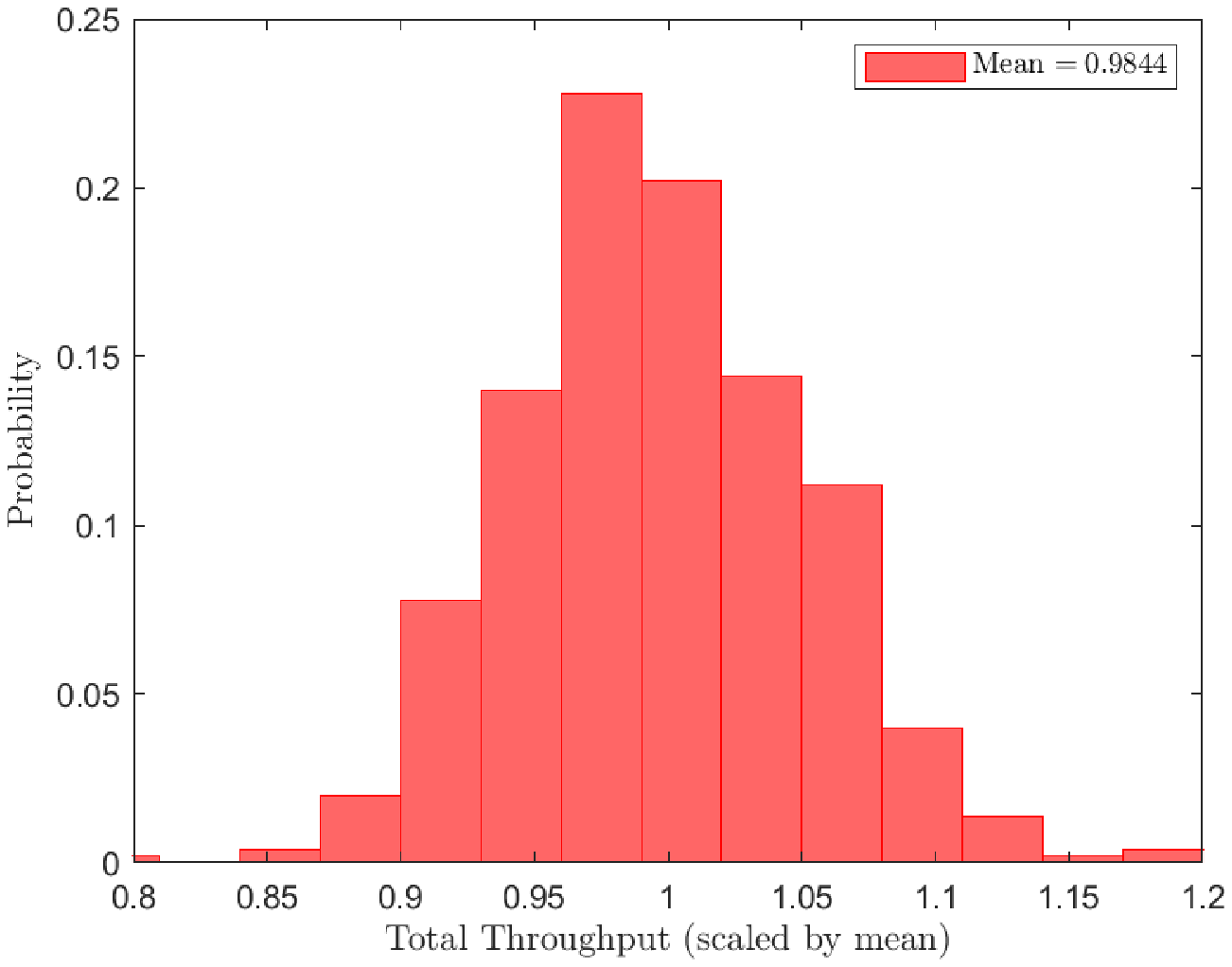}}
    \caption{Histogram of initial void volume, tortuosity and total throughput with $500$ simulations, scaled by the mean of each quantity respectively, under (a,c,e) {\color{blue}Isolated setup}; (b,d,f) {\color{red}periodic setup}. For all figures, $d=0.1$ and $N_{\rm total} = 4100$.}
    \label{fig:hist_d_01}
\end{figure}
\begin{figure}[!ht]
    \centering
    \subfloat[]{\label{fig:hist_d_045_vol_np}\includegraphics[width=.47\textwidth,height=.4\textwidth]{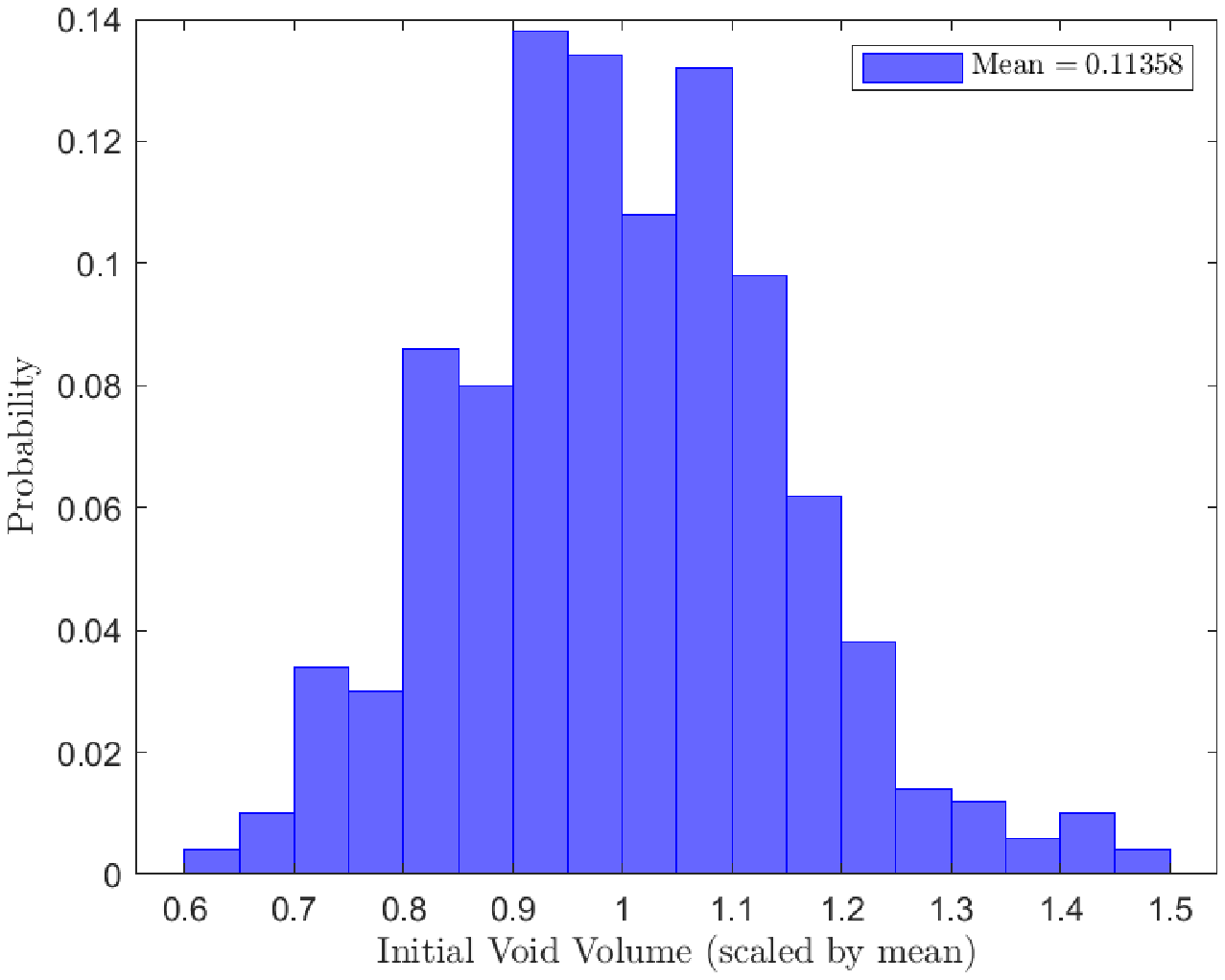}}
    \subfloat[]{\label{fig:hist_d_045_vol_p}\includegraphics[width=.47\textwidth,height=.4\textwidth]{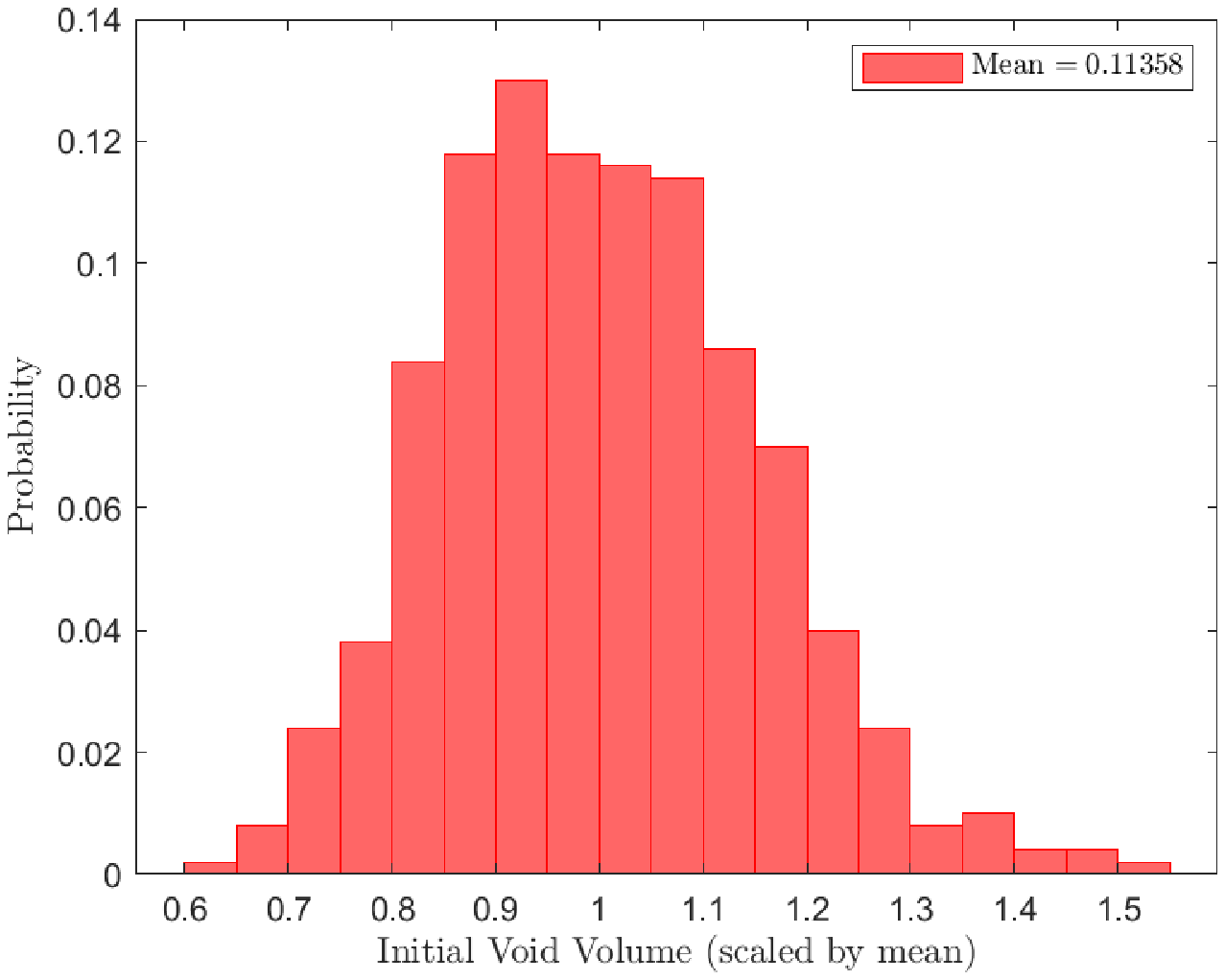}}\\
    \subfloat[]{\label{fig:hist_d_045_tau_np}\includegraphics[width=.47\textwidth,height=.4\textwidth]{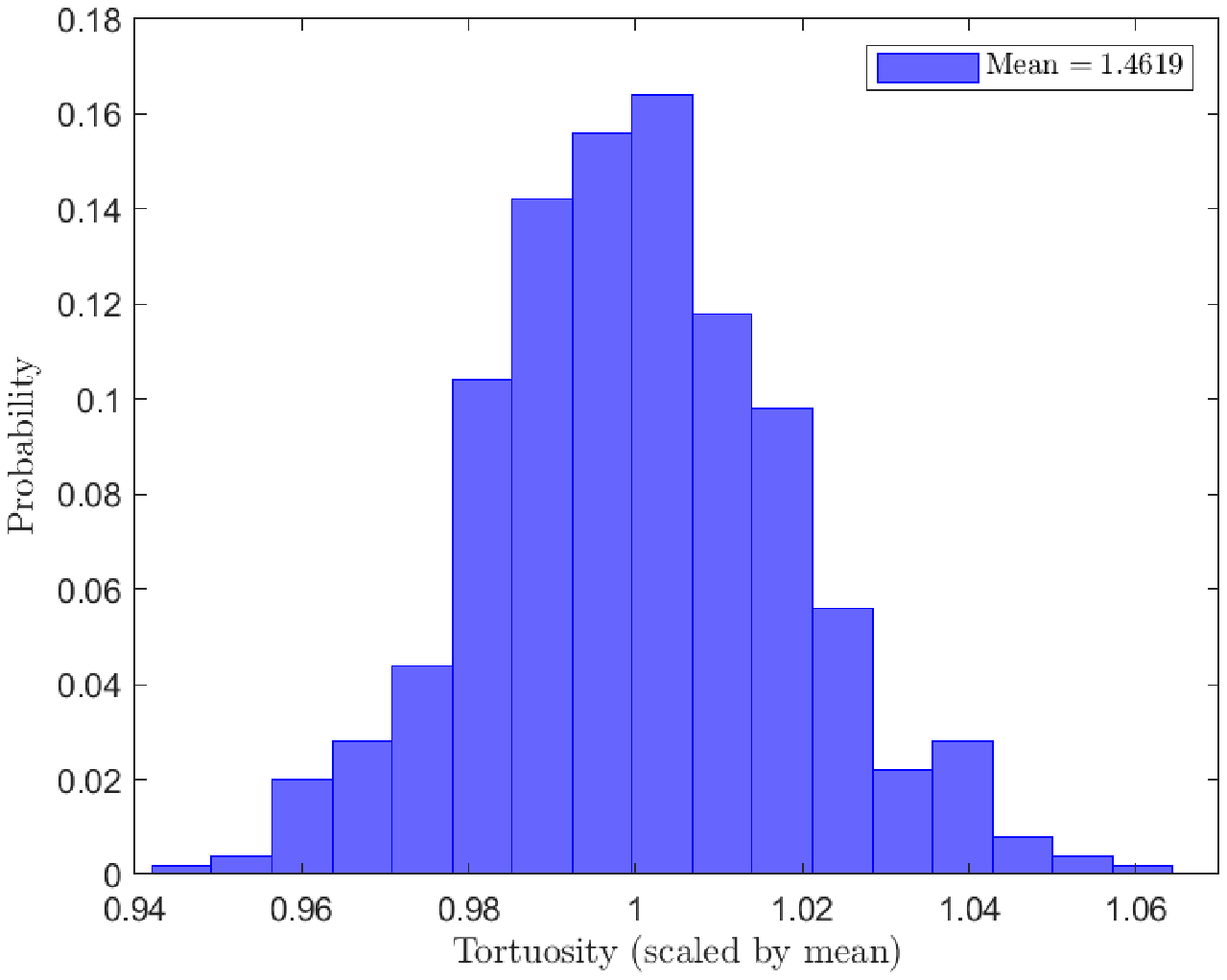}}
    \subfloat[]{\label{fig:hist_d_045_tau_p}\includegraphics[width=.47\textwidth,height=.4\textwidth]{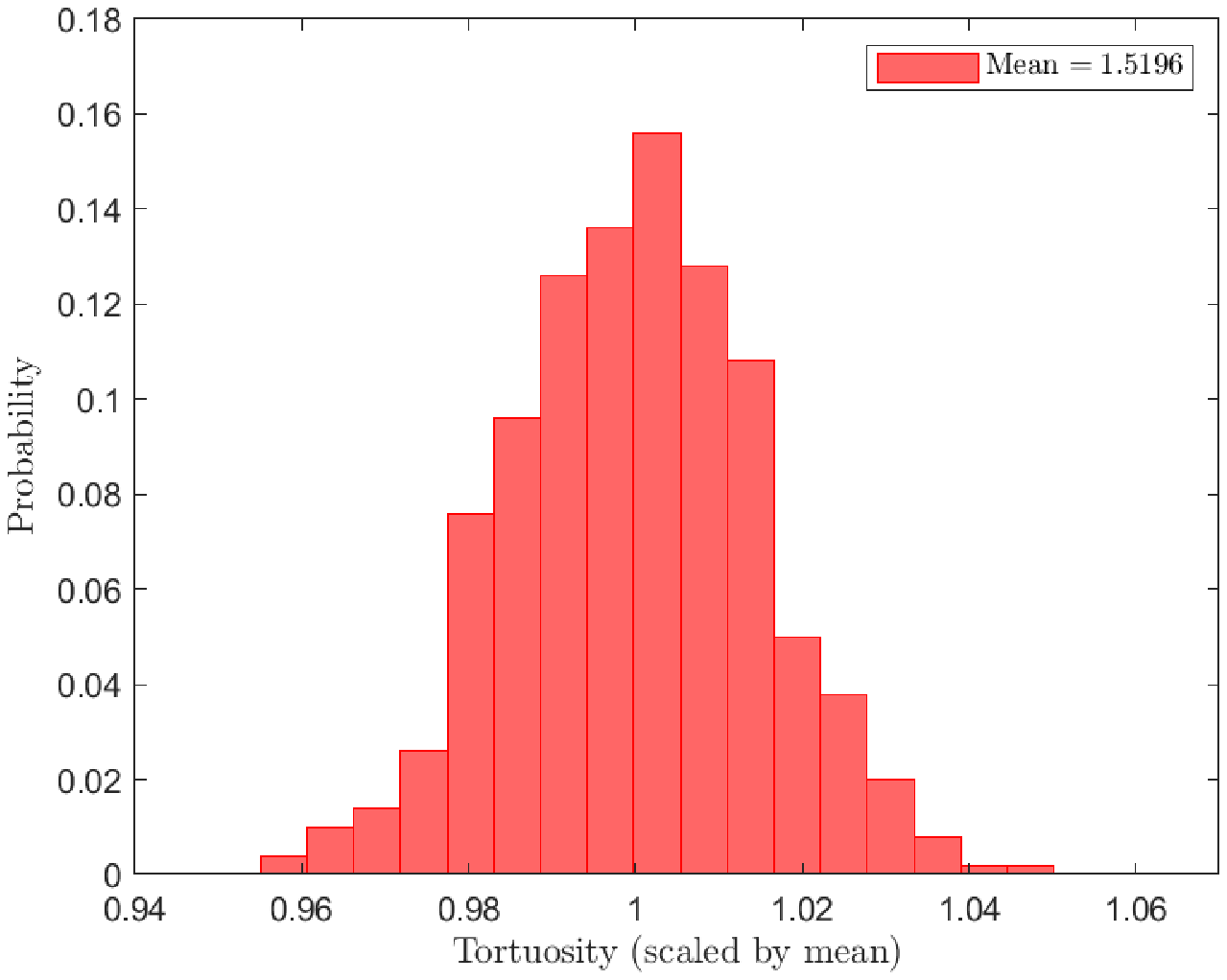}}\\
    \subfloat[]{\label{fig:hist_d_045_tt_np}\includegraphics[width=.47\textwidth,height=.4\textwidth]{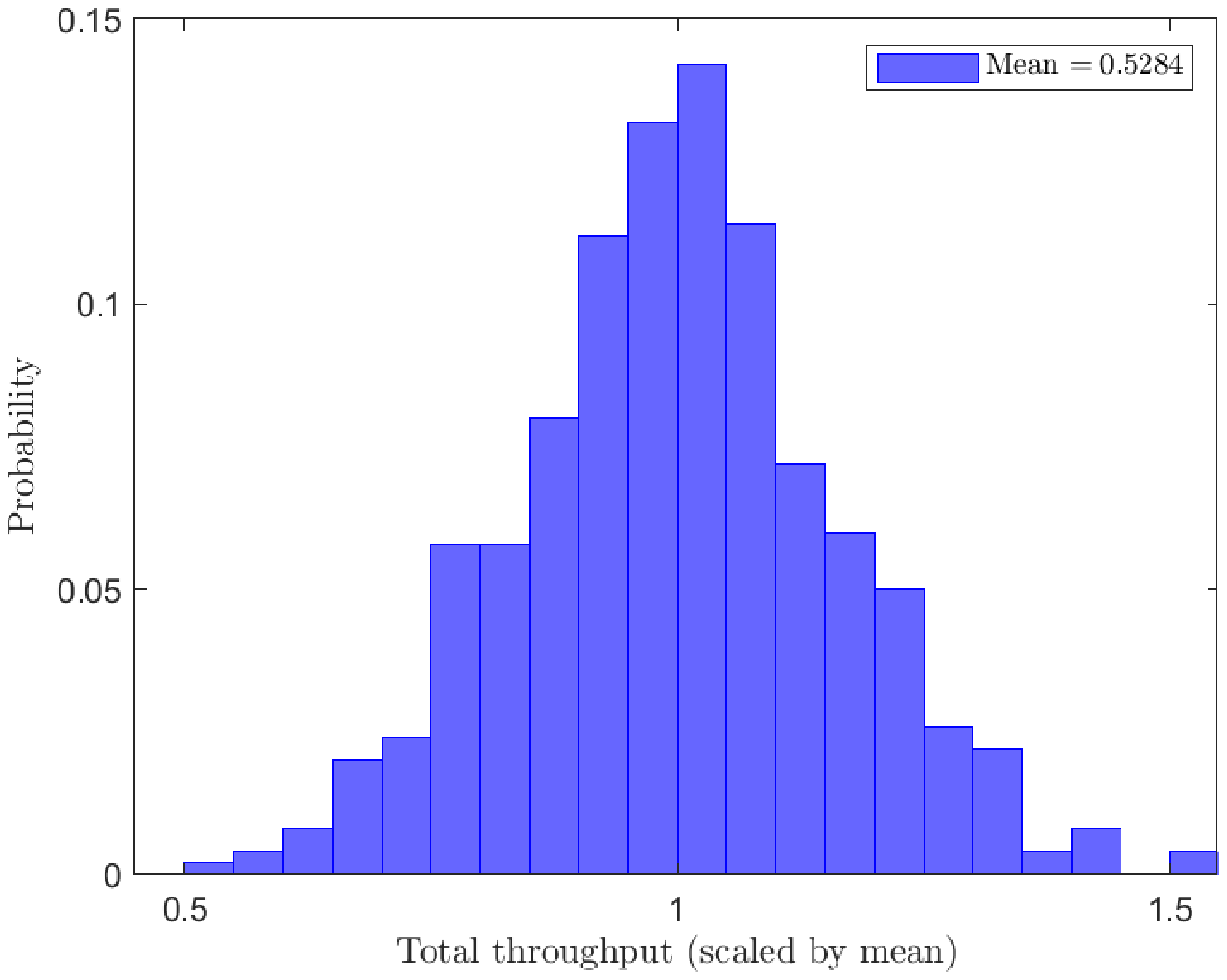}}
    \subfloat[]{\label{fig:hist_d_045_tt_p}\includegraphics[width=.47\textwidth,height=.4\textwidth]{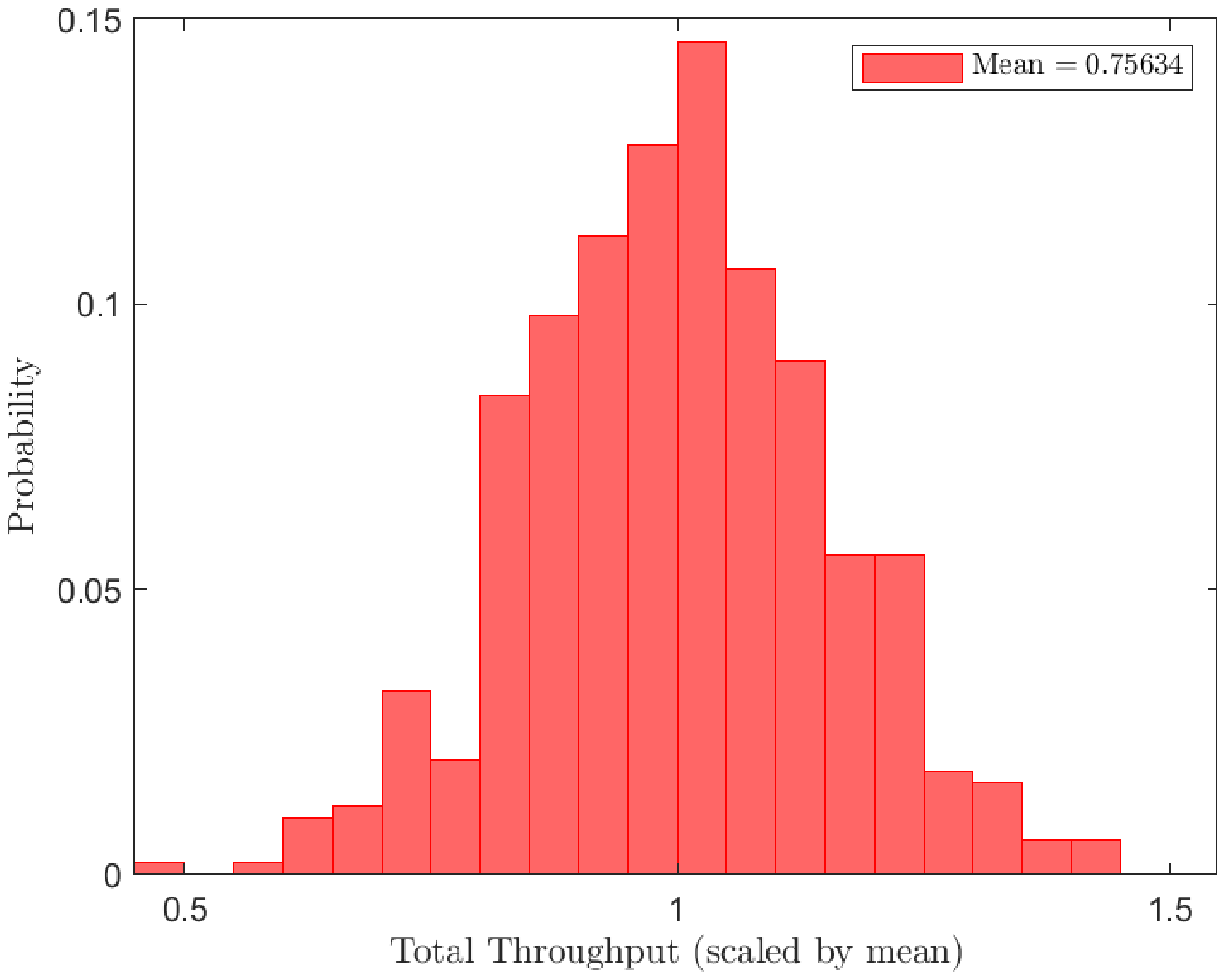}}
    \caption{Histogram of initial void volume, tortuosity and total throughput with $500$ simulations, scaled by the mean of each quantity respectively, under (a,c,e) {\color{blue}Isolated setup}; (b,d,f) {\color{red}periodic setup}. For all figures, $d=0.45$ and $N_{\rm total} = 150$.}
    \label{fig:hist_d_045}
\end{figure}

%% file: ex_shared.tex
\usepackage{amsfonts}
\usepackage{graphicx}
\usepackage{epstopdf}
\usepackage{algorithmic}
\usepackage{graphicx,color}
\usepackage[section]{placeins}
\usepackage{subfig}
\usepackage{epsfig}
\usepackage{amsmath}
\usepackage{hyperref}
\usepackage{verbatim}
\usepackage{chngcntr}

\usepackage{tikz}
\usetikzlibrary{intersections}
\usetikzlibrary{calc}
\usetikzlibrary{positioning}
\usetikzlibrary{backgrounds}
\usetikzlibrary{decorations.markings}
\pgfdeclarelayer{nodelayer}
\pgfdeclarelayer{edgelayer}
\pgfsetlayers{edgelayer,nodelayer,main}

\definecolor{rvwvcq}{rgb}{0.08235294117647059,0.396078431372549,0.7529411764705882}
\definecolor{qqqqff}{rgb}{0,0,1}
\definecolor{ffqqqq}{rgb}{1,0,0}

\ifpdf
  \DeclareGraphicsExtensions{.eps,.pdf,.png,.jpg}
\else
  \DeclareGraphicsExtensions{.eps}
\fi

\newsiamremark{remark}{Remark}
\newsiamremark{hypothesis}{Hypothesis}
\crefname{hypothesis}{Hypothesis}{Hypotheses}
\newsiamthm{claim}{Claim}

\headers{A Graphical Representation of Membrane Filtration}{B. Gu, L. Kondic, and Linda J. Cummings}

\title{A Graphical Representation of Membrane Filtration\thanks{Submitted to the editors 06/07/2021.
\funding{This work was funded by NSF Grant No. DMS-1615719.}}}

\author{Binan Gu\thanks{Department of Mathematical Sciences, New Jersey Institute of Technology, Newark, NJ
  (\email{bg263@njit.edu}, \url{https://web.njit.edu/\~bg263/}).}
\and Lou Kondic\footnotemark[2]
\and Linda J. Cummings\footnotemark[2]
  }

\usepackage{amsopn}
